\documentclass[a4paper,11pt,reqno]{amsart}
\usepackage[latin1]{inputenc}
\usepackage{amsfonts,amssymb}
\usepackage{color}
\usepackage{amsthm, amssymb, latexsym, amsmath}
\usepackage{dsfont}
\usepackage{graphicx}

\newtheorem{theorem}{Theorem}[section]
\newtheorem{example}[theorem]{Example}
\newtheorem{lemma}[theorem]{Lemma}

\newtheorem{proposition}[theorem]{Proposition}

\newtheorem{conjecture}[theorem]{Conjecture}
\newtheorem{remark}{Remark}
\newtheorem{definition}{Definition}

\newcommand{\DETAILS}[1]{}

\newcommand{\Ran}{\operatorname{Ran}}
\newcommand{\tr}{\operatorname{Tr}\,}

\renewcommand{\H}{\operatorname{He}}

\newcommand{\1}{\mathds{1}}
\newcommand{\C}{\mathbb{C}}
\newcommand{\R}{\mathbb{R}}

\newcommand{\N}{\mathbb{N}}

\newcommand{\BR}{{\mathbb R}}

\newcommand{\D}{{\mathcal D}}

\newcommand{\cF}{{\mathcal{F}}}

\newcommand{\cG}{{\mathcal{G}}}
\newcommand{\cM}{{\mathcal{M}}}
\newcommand{\cR}{{\mathcal{R}}}

\renewcommand{\D}{D}
\newcommand{\Q}{{Q}}
\renewcommand{\O}{{O}}

\newcommand{\gH}{\mathfrak{H}}
\newcommand{\gS}{\mathfrak{S}}
\newcommand{\cE}{\mathcal{E}}

\newcommand{\cQ}{\mathcal{Q}}

\newcommand{\ii}{\infty}
\newcommand{\eps}{{\varepsilon}}

\newcommand{\al}{{L}}

\newcommand{\U}{{U}}
\newcommand{\V}{{V}}

\newcommand{\ran}{\rangle}
\newcommand{\lan}{\langle}

\newcommand\pscal[1]{{\ensuremath{\left\langle #1 \right\rangle}}}
\newcommand{\norm}[1]{ \left| \! \left| #1 \right| \! \right| }
\newcommand{\nn}{\nonumber}
\newcommand{\lk}{L_k}

\begin{document}

\title[Compactness of molecular reaction paths]{Compactness of molecular reaction paths in quantum mechanics}

\author[I. Anapolitanos]{Ioannis Anapolitanos}
\address{Dept.~of Math., Karlsruhe Institute of Technology, Karlsruhe, Germany} 
\email{ioannis.anapolitanos@kit.edu}

\author[M. Lewin]{Mathieu Lewin}
\address{CNRS \& CEREMADE, University Paris-Dauphine, PSL University, 75016 Paris, France} 
\email{mathieu.lewin@math.cnrs.fr}

\maketitle

\bigskip

\begin{abstract}
We study isomerizations in quantum mechanics. We consider a neutral molecule composed of $N$ quantum electrons and $M$ classical nuclei and assume that the first eigenvalue of the corresponding $N$-particle Schrödinger operator possesses two local minima with respect to the locations of the nuclei. An isomerization is a mountain pass problem between these two local configurations, where one minimizes over all possible paths the highest value of the energy along the path. Here we state a conjecture about the compactness of min-maxing sequences of such paths, which we then partly solve in the particular case of a molecule composed of two rigid sub-molecules that can move freely in space. More precisely, under appropriate assumptions on the multipoles of the two molecules, we are able to prove that the distance between them stays bounded during the whole chemical reaction. We obtain a critical point at the mountain pass level, which is called a transition state in chemistry. Our method requires to study the critical points and the Morse indices of the classical multipole interactions, as well as to generalize existing results about the van der Waals force. This paper generalizes previous works by the second author in several directions. 

\bigskip

\noindent\textsl{Final version to appear in \emph{Arch. Rat. Mech. Anal.}}
\end{abstract}

\bigskip\bigskip

\tableofcontents

\bigskip

In this paper we study isomerizations, which are chemical reactions with the property that the reactant is a molecule with the same atoms as the product, but in a different spacial configuration, for example HCN $\rightarrow$ CNH. The question how much energy is needed for such a reaction to take place is a very fundamental problem, which occupies a large amount of the numerical computations ran today in quantum chemistry. If the reaction is slow, it is customary to rely on the Born-Oppenheimer approximation of the $N$-body Schr\"odinger equation, that is, to assume that the nuclei are classical and pointwise particles, whereas the electrons are quantum and placed in their ground state. The two stable configurations of the molecule then correspond to local minima with respect to the positions of the nuclei. The chemical reaction is described by a path connecting these two states. The difference of the maximum and initial energies along one path corresponds to the amount of energy needed to bring the system from one end point to the other. The \emph{activation energy} is the lowest energy needed for the reaction to happen, and it is obtained by minimizing over all possible paths linking the two stable configurations. 

From the point of view of critical point theory, a chemical reaction corresponds to a \emph{mountain pass problem}~\cite{AmbRab-73,Struwe,Jabri-03}. For neutral molecules, it has been conjectured in~\cite{Lewin-04b,Lewin-PhD,Lewin-06} that \emph{all} possible isomerizations happen without breaking the molecule into pieces. Mathematically, this means that one can find sequences of reaction paths approaching the optimal activation energy, which only involve nuclei in a compact set. In this case, after passing to the limit, one obtains a critical point at the mountain pass level, which is called the \emph{transition state} in chemistry. As usual in critical point theory, it is not obvious to get a corresponding optimal path in the limit, but there is no doubt that this path should exist as well in most cases. 

The previously mentioned conjecture is partly motivated by the van der Waals force. This purely quantum effect occurs when two neutral molecules are placed at a large distance $L$ from each other. An attractive force of the order of $1/L^7$ is induced as a result of quantum correlations between the two systems. The existence of this force has been rigorously proved in a celebrated paper of Lieb and Thirring~\cite{LieThi-86} and its exact expression under some non-degeneracy assumptions has then been more recently derived for individual atoms by Sigal and the first author of this paper in~\cite{Anapolitanos-16,AnaSig-17}. The van der Waals force implies that \emph{all} neutral molecules can bind in the Born-Oppenheimer approximation, that is, the energy always has a global minimum with respect to the nuclear positions. This is because the energy of a molecule which splits into pieces is necessarily higher than the lowest possible energy, due to the van der Waals force. 

Isomerizations are however much more subtle than minima. Comparing energies is not sufficient to deduce that a path approaching the optimal activation energy will not involve nuclei very far from each other. For two molecules placed far away, there usually exist some orientations for which they repel, for instance when they have dipole moments oriented in opposite directions. Deforming a whole path involving such states requires more information than absolute energies. It is necessary to know the number of directions in which one can decrease the energy, that is, the Morse index of the ``critical points at infinity''~\cite{Bahri-89,Lewin-04b}. At least two such directions are necessary to deform a one-dimensional path and decrease its maximal energy.\footnote{Two directions are needed since the path could be already tangent to one direction of descent, as it is for a mountain pass.} In other words, the Morse index at infinity must be equal to two or higher, in order to eliminate the loss of compactness of optimal sequences of paths. 

In~\cite{Lewin-04b}, it was proved that the critical points of the dipole-dipole interaction which have a positive energy all have a Morse index greater than or equal to 2. This was used to prove the compactness of all reaction paths, for a system containing two rigid molecules, each having a non-degenerate ground state with a non-vanishing dipole. In~\cite{Lewin-06} the completely different situation of a molecule with only one moving atom (like for the reaction HCN $\rightarrow$ CNH) was treated. If the single atom escapes to infinity, then only the van der Waals force pertains since, by symmetry, the atomic ground state has no multipole in average. 
All the other cases were left open in~\cite{Lewin-04b,Lewin-06}. 

In this paper, we extend these results in many directions, and make important progress for the case of a molecule composed of two rigid sub-molecules. We are able to get several new results on the critical points of the (classical) multipolar interactions. Those dominate the energy in case the two molecules have low order multipoles, that is, such that the energy is of the order $1/L^p$ with $p<6$. On the other hand, if sufficiently many of their multipoles vanish, then the van der Waals force dominates the energy to leading order and it may be used to prove the compactness of reaction paths. Our approach then requires to improve several of the existing results on the van der Waals force and our findings in this direction may be of general interest. 

In the next section we define the system properly and describe the main open questions. Then we recall some results from~\cite{Lewin-04b,Lewin-06} and state our new results. The rest of the paper is devoted to the proof of our theorems.

\subsubsection*{Acknowledgements} 
M.L. would like to thank \'Eric S\'er\'e who suggested him this problem 15 years ago.  I.A. is grateful to Marcel Griesemer, Dirk Hundertmark and Semjon Wugalter for inspiring discussions on van der Waals forces between molecules, and to Michal Jex for discussions on the physical background of the problem. He gratefully acknowledges Matthias Roth for numerous stimulating discussions on the expansion of derivatives of the energy. Both authors thank Semjon Wugalter for useful remarks which have led us to include the spin into account. The research of I.A. was funded by the Deutsche Forschungsgemeinschaft (DFG, German Research Foundation) - Project-ID 258734477 - SFB 1173 and at an early stage under the grant number GR 3213/1-1. 
This project has received funding from the European Research Council (ERC) under the European Union's Horizon 2020 research and innovation programme (grant agreement MDFT No 725528 of M.L.).

\section{Model and main results}

\subsection{A conjecture for isomerizations in quantum mechanics}

\subsubsection{Schrödinger Hamiltonian for molecules}
We consider a system composed of $N$ quantum electrons and $M$ nuclei, of charges $z_1,\dots ,z_M\in\N$ and located at $y_1,\dots ,y_M\in\R^3$, with $y_j\neq y_k$ for $j\neq k$. In units where Planck's constant is $\hbar=1$, the electron mass is $m_e=\frac{1}{2}$ and the elementary charge is
$e=1$,  the Born-Oppenheimer Hamiltonian of the full system reads
\begin{multline}
H_N(Y,Z):=\sum_{j=1}^N-\Delta_{x_j}-\sum_{j=1}^N\sum_{m=1}^M\frac{z_m}{|x_j-y_m|}+\sum_{1\leq j<k\leq N}\frac{1}{|x_j-x_k|}\\
+\sum_{1\leq \ell<m\leq M}\frac{z_mz_\ell}{|y_m-y_\ell|}
\label{eq:Hamiltonian_full}
\end{multline}
where we have used the shorthand notation $Y=(y_1,\dots ,y_M)$ and $Z=(z_1,\dots ,z_M)$. 
The total nuclear charge will be denoted by
$$|Z|:=\sum_{m=1}^Mz_m.$$
Due to the fermionic nature of the electrons, the Hamiltonian~\eqref{eq:Hamiltonian_full} acts in the Hilbert space
$$L^2_a\left((\R^3\times\{\pm1/2\})^N,\C\right)\simeq \bigwedge_1^NL^2\left(\R^3\times\{\pm1/2\},\C\right)$$
of antisymmetric square-integrable wavefunctions $\Psi(x_1,s_1,\dots ,x_N,s_N)$ with spin, that is, such that
\begin{equation}\label{eq:fermions}
\pi\cdot \Psi(X_1,...,X_N):=\Psi(X_{\pi(1)},\dots ,X_{\pi(N)})=(-1)^\pi\,\Psi(X_{1},\dots ,X_{N}) 
\end{equation}
for any permutation $\pi\in\mathfrak{S}_N$, and where $X=(x,s)\in \R^3\times\{\pm1/2\}$. None of the next results turns out to depend on the statistics of the particles, nor of the presence of the spin, but we stick to this case for obvious physical reasons. The operator $H_N(Y,Z)$ is essentially self-adjoint on $C^\ii_c$ and its domain is the Sobolev space
$$D\big(H_N(Y,Z)\big)=H^2_a\left((\R^3\times\{\pm1/2\})^N,\C\right)\simeq \bigwedge_1^NH^2\left(\R^3\times\{\pm1/2\},\C\right).$$
Often we just write $L^2_a$ and $H^2_a$ when there is no possible confusion about the value of $N$.

We denote by 
$$\boxed{E_N(Y,Z):=\min\sigma\big(H_N(Y,Z)\big)}$$
the bottom of the spectrum of $H_N(Y,Z)$, which is a translation-invariant function of the nuclear positions $Y\in(\R^3)^M$. 
When $N<|Z|+1$ (neutral or positively charged molecules), the HVZ~\cite{Hunziker-66,VanWinter-64,Zhislin-60} and Zhislin-Sigalov~\cite{Zhislin-60,ZhiSig-65} theorems imply that $E_N(Y,Z)$ is an eigenvalue of $H_N(Y,Z)$, lying strictly below the essential spectrum:
$$E_N(Y,Z)<\min\sigma_{\rm ess}\big(H_N(Y,Z)\big)=E_{N-1}(Y,Z).$$
In other words, for every $Y$ there exists at least one eigenfunction $\Psi$, such that 
$$H_N(Y,Z)\Psi=E_N(Y,Z)\,\Psi.$$
When $y_j\to y_k$ for some $j\neq k$, then $E_N(Y,Z)\to+\ii$ due to the nuclear repulsion. We therefore use the convention $E_N(Y,Z)=+\ii$ when two nuclei are on top of each other.

\subsubsection{Binding and the van der Waals force}
It has been proved that all neutral molecules can bind in the Born-Oppenheimer approximation, which means that there is always a minimum with respect to the nuclear positions. This minimum is of course never unique, since the system is invariant under translations.

\begin{theorem}[All neutral molecules can bind~\cite{Morgan-79,MorSim-80,LieThi-86}]\label{thm:molbind}
Assume that $N=|Z|$ (neutral case). Then there exist some nuclear positions $\bar Y=(\bar{y}_1,\dots ,\bar{y}_M)\in(\R^3)^M$ such that 
\begin{equation}
E_N(\bar{Y},Z)=\min_{Y\in(\R^3)^M}E_N(Y,Z).
\label{eq:exist_min}
\end{equation}
More precisely, we have 
\begin{equation}
\min_{Y\in(\R^3)^M}E_N(Y,Z)< \liminf_{\sum_{m\neq\ell}|y_m-y_\ell|\to\ii}E_N(Y,Z).
\label{eq:binding}
\end{equation}
\end{theorem}

The inequality~\eqref{eq:binding} means that it is not energetically favorable to split the molecule in pieces. Since $Y\mapsto E_N(Y,Z)$ is continuous, this immediately implies the existence of a global minimizer, as stated  in~\eqref{eq:exist_min}. 

The behavior of the energy when the molecule splits into several subsystems receeding from each other as in~\eqref{eq:binding} is rather well understood. Let us explain this in the case when the molecule splits in two pieces, placed at a distance $L\to\ii$, for instance in the direction $e_1=(1,0,0)$. This corresponds to taking $Y=(Y_1,Y_2+L e_1)$ where it is understood in our notation that
$$(y_1,\dots ,y_M)+Le_1=(y_1+Le_1,\dots ,y_M+Le_1).$$
We assume for simplicity that the two molecules are rigid, that is, the nuclear positions $Y_1$ and $Y_2$ do not depend on $L$. 
Morgan and Simon have shown in~\cite{Morgan-79,MorSim-80} (see also~\cite{AveSei-75,Ahlrichs-76,ComSei-78,Lewin-04b}) that
\begin{equation}
\lim_{L\to\ii} E_N(Y_1,Y_2+Le_1,Z)=\min_{N_1+N_2=N}\big\{E_{N_1}(Y_1,Z_1)+E_{N_1}(Y_2,Z_2)\big\}
\label{eq:limit_Morgan_Simon}
\end{equation}
where we have denoted $Z=(Z_1,Z_2)$ for shortness. In other words, in the dissociation limit the energy becomes the sum of the energies of the two sub-molecules, but one should not forget to optimize over the ways of distributing the electrons. It is a famous conjecture that the minimum on the right side of~\eqref{eq:limit_Morgan_Simon} is attained only in the neutral case $N_1=|Z_1|$ and $N_2=|Z_2|$, see~\cite{AnaSig-17} and references therein.

If nevertheless~\eqref{eq:limit_Morgan_Simon} is attained for some $N_1\neq |Z_1|$, then placing the electrons in this manner one finds a Coulomb attraction between the two charged systems, which results in the following upper bound
\begin{multline}
E_N(Y_1,Y_2+Le_1,Z)\leq E_{N_1}(Y_1,Z_1)+E_{N_2}(Y_2,Z_2)\\+\frac{(N_1-|Z_1|)(N_2-|Z_2|)}{L}+o\left(\frac1L\right) 
\label{eq:split_non_neutral}
\end{multline}
where $(N_1-|Z_1|)(N_2-|Z_2|)<0$ since $N=N_1+N_2=Z_1+Z_2$. This is proved by considering $\Psi_1\wedge\Psi_2(\cdot-Le_1)$ for $\Psi_{1}$ and $\Psi_2$ two  ground states of $E_{N_1}(Y_1,Z_1)$ and $E_{N_2}(Y_2,Z_2)$, respectively\footnote{These ground states exist by~\cite[Lemma 6]{Lewin-04b}.}, see e.g.~\cite[Lemma 2]{Lewin-04b}. Hence one gets immediately~\eqref{eq:binding}.

In the neutral case $N_1=|Z_1|$ in~\eqref{eq:binding}, Lieb and Thirring have shown in~\cite{LieThi-86} that 
\begin{multline}
\int_{SO(3)}\int_{SO(3)} E_N(UY_1,VY_2+Le_1,Z)\,{\rm d}U\,{\rm d}V\\
\leq E_{N_1}(Y_1,Z_1)+E_{N_2}(Y_2,Z_2)-\frac{C}{L^6}
\label{eq:vdW_LT}
\end{multline}
for some $C>0$. Hence there is at least one orientation of each of the molecules for which 
$$E_N(UY_1,VY_2+Le_1,Z)\leq E_{N_1}(Y_1,Z_1)+E_{N_2}(Y_2,Z_2)-\frac{C}{L^6}.$$
A similar argument in case the molecule splits into several sub-systems implies~\eqref{eq:binding} and then~\eqref{eq:exist_min}.

In~\cite{Anapolitanos-16,AnaSig-17}, the exact value of the constant $C_{\rm vdW}$ in the expansion
\begin{equation}
E_N(Y_1,Y_2+Le_1,Z)= E_{N_1}(Y_1,Z_1)+E_{N_2}(Y_2,Z_2)-\frac{C_{\rm vdW}}{L^6}+o\left(\frac1{L^6}\right)
\label{eq:vdW_AnaSig}
\end{equation}
was found, assuming that the two individual molecules are atoms, with some irreducibility assumptions on the ground state eigenspaces. Under more stringent assumptions, Morgan and Simon had already proved in~\cite{MorSim-80} that $E_N(Y_1,Y_2+Le_1,Z)$ can be expanded as an infinite power series in $L^{-1}$ and mentioned the van der Waals interpretation of the 6th order term. 

In this paper we will generalize~\eqref{eq:vdW_LT} and~\eqref{eq:vdW_AnaSig} in several directions.

\subsubsection{Isomerizations}
Next we consider a molecule which possesses two locally stable configurations, that is, we assume that the function $Y\mapsto E_N(Y,Z)$ has two local minima, called $\bar{Y}_0$ and $\bar{Y}_1$. We then consider all the continuous paths linking them and define the mountain pass level by
\begin{equation}
\boxed{c:=\inf_{\substack{Y(t)\in C^0([0,1],(\R^3)^M)\\ Y(0)=\bar Y_0\\Y(1)=\bar Y_1}}\;\max_{t\in[0,1]} E_N\big(Y(t),Z\big).}
\label{eq:mountain_pass_level_general_case}
\end{equation}
Since the energy blows up when two nuclei go on top of each other, only paths with nuclei in 
$$(\R^3)^M\setminus\bigg\{\prod_{j\neq k}|y_j-y_k|=0\bigg\}$$
will be involved in~\eqref{eq:mountain_pass_level_general_case}. 

Of course, if $\bar{Y}_0$ and $\bar{Y}_1$ can be deduced from one another by translating or rotating the whole system, then the problem is trivial and we have $c=E_N\big(\bar{Y}_0,Z\big)=E_N\big(\bar{Y}_1,Z\big)$. This trivial situation is avoided by assuming that going from $\bar{Y}_0$ to $\bar{Y}_1$ actually requires to pass an energy barrier:
$$\boxed{c>\max\Big\{E_N\big(\bar{Y}_0,Z\big),E_N\big(\bar{Y}_1,Z\big)\Big\}.}$$
This condition is satisfied if, for instance, $\bar{Y}_0$ and $\bar{Y}_1$ are strict local minima, up to the natural symmetries of the system.

The mountain pass problem~\eqref{eq:mountain_pass_level_general_case} plays an important role in quantum chemistry. It is well understood in terms of the theory of ``rare events" in probability, see for instance \cite{EVan-06,LuNol-15} and  references therein. The idea is to add a Brownian motion to the nuclei and look for the most probable path they will take to go from one minimum to the other one, waiting an exponentially long time for this to happen. Here we concentrate on the (deterministic) definition of the mountain pass level~\eqref{eq:mountain_pass_level_general_case}. 

\begin{conjecture}[Compactness of all isomerizations~\cite{Lewin-04b,Lewin-PhD,Lewin-06}]\label{conjecture}
Let $N=|Z|$ (neutral case). Assume that $\bar Y_0$ and $\bar Y_1$ are two local minima of $Y\mapsto E_N(Y,Z)$ and that
$$c>\max\big\{ E_N(\bar Y_0,Z),E_N(\bar Y_1,Z)\big\}.$$
Then there exists a min-maxing sequence of paths $\{Y_n(t)\}\subset C^0([0,1],(\R^3)^M)$ with $Y(0)=\bar Y_0$ and $Y(1)=\bar Y_1$, that is, satisfying 
$$\lim_{n\to\ii}\max_{t\in[0,1]} E_N(Y_n(t),Z)=c,$$
which is compact in the sense that the nuclei do not escape to infinity:
\begin{equation}
 |y_{n,j}(t)|\leq R,
 \label{eq:compactness}
\end{equation}
for all $j$, all $n\geq1$ and all $t\in[0,1]$.
\end{conjecture}

Assuming that the conjecture is right, one would like to deduce the existence of a \emph{transition state}, that is, a critical point at the mountain pass level. Unfortunately, we have to face here the problem that the function $Y\mapsto E_N(Y,Z)$ is Lipschitz but not necessarily $C^1$ in case of degeneracy. If the mountain pass is attained at such a point where $E_N(Y,Z)$ is degenerate, the Born-Oppenheimer approximation becomes questionable, but it is still reasonable to ask what can be said. 

The way out found in~\cite{Lewin-04b,Lewin-PhD,Lewin-06} is to remark that the total energy 
$$(Y,\Psi)\mapsto \pscal{\Psi,H_N(Y,Z)\Psi}$$
is itself $C^1$. It indeed makes sense to let the electrons evolve along the path as well. For this reason, let us choose two ground states $\bar\Psi_0$ and $\bar\Psi_1$ corresponding to the eigenvalues $E_N(\bar Y_0,Z)$ and $E_N(\bar Y_1,Z)$, respectively, and define a new mountain pass level by
\begin{equation}
\boxed{c':=\inf_{\substack{
(Y,\Psi)(0)=(\bar Y_0,\bar\Psi_0)\\ (Y,\Psi)(1)=(\bar Y_1,\bar\Psi_1)}}\;\max_{t\in[0,1]} \pscal{\Psi(t),H_N\big(Y(t),Z\big)\Psi(t)}}
\label{eq:mountain_pass_level_general_case_with_Psi}
\end{equation}
which is obviously larger than or equal to $c$.
It is understood here that $Y(t)$ is a path as before and that $\Psi(t)$ is a continuous path of wavefunctions in $H^1_a((\R^3\times\{\pm1/2\})^N,\C)$ such that $\|\Psi(t)\|_{L^2}=1$ for all $t\in[0,1]$.

If $E_N(Y(t),Z)$ has a constant multiplicity along one path $Y(t)$ then, by usual perturbation theory~\cite{Kato}, one can choose a corresponding ground state $\Psi(t)$ in a smooth manner and the energy along the path is just $E_N(Y(t),Z)$ for all $t\in[0,1]$. Note that a path $\Psi(t)$ with only real-valued functions which starts at $\bar\Psi_0$ may end up at the wrong final state $-\bar\Psi_1$ instead of $\bar\Psi_1$, and then $c'$ could be different from $c$. However, for complex-valued wavefunctions, we can always choose the phase appropriately and bring $\bar\Psi_0$ to $\bar\Psi_1$, while staying at all times on the ground state. The fact that the wavefunctions are complex-valued therefore plays a role here. When the multiplicity is constant along one path, we conclude that dressing the path with electrons does not change the maximal energy. 

Even when the multiplicity of $E_N(Y(t),Z)$ is not constant along one path $Y(t)$, it was shown in~\cite[Thm.~4]{Lewin-04b} (see also~\cite[chapter~B.II]{Lewin-PhD}) that for every $\eps>0$, one can find a continuous path $\Psi(t)$ such that 
$$\max_{t\in[0,1]}\pscal{\Psi(t),H_N(Y(t),Z)\Psi(t)}\leq \max_{t\in[0,1]}E_N(Y(t),Z)+\eps.$$
Roughly speaking, the idea is to follow the ground state eigenspace on the pieces of the path where it is non-degenerate, and to paste these curves appropriately in the neighborhood of the degeneracies. This is explained in a slightly different manner in Appendix~\ref{app:dressed_path}. 
So we conclude after passing to the limit $\eps\to0$ that 
$$\boxed{c=c'}$$
in all cases.

Now that we can work with ``dressed'' paths where the electrons are not necessarily in their ground states, we can take benefit of the smoothness of the energy $(Y,\Psi)\mapsto \pscal{\Psi,H_N(Y,Z)\Psi}$. The following result was proved in~\cite[Thm.~4]{Lewin-04b}, using standard tools from critical point theory~\cite{Rabinowitz-86,Ghoussoub-93}, together with some known spectral properties of the $N$-body Schrödinger operator.

\begin{theorem}[Existence of a transition state~{\cite[Thm.~4]{Lewin-04b}}]\label{thm:existence_mountain_pass}
Let $N=|Z|$ (neutral case). Assume that $\bar Y_0$ and $\bar Y_1$ are two local minima of $Y\mapsto E_N(Y,Z)$ such that 
$$c>\max\big\{ E_N(\bar Y_0,Z),E_N(\bar Y_1,Z)\big\}$$
and that there exists a min-maxing sequence of paths 
$$\{Y_n(t)\}\subset C^0([0,1],(\R^3)^M)$$ 
for the mountain pass problem~\eqref{eq:mountain_pass_level_general_case}, satisfying the compactness property~\eqref{eq:compactness}. Then one can find a path of normalized wavefunctions $\Psi_n(t)\in H^2_a((\R^3\times\{\pm1/2\})^N,\C)$ such that $(Y_n(t),\Psi_n(t))$ is a min-maxing sequence for~\eqref{eq:mountain_pass_level_general_case_with_Psi},
\begin{multline*}
\lim_{n\to\ii}\max_{t\in[0,1]}E_N\big(Y_n(t),Z\big)=\lim_{n\to\ii}\max_{t\in[0,1]}\pscal{\Psi_n(t),H_N(Y_n(t),Z)\Psi_n(t)}\\=\lim_{n\to\ii}E_N\big(Y_n(t_n),Z\big)=c,
\end{multline*}
with, in addition, 
$$\Big(H_N(Y_n(t_n),Z)-c\Big)\Psi(t_n)\to0,$$
$$\nabla_Y\pscal{\Psi(t_n),H_N\big(Y_n(t_n),Z\big)\Psi(t_n)}\to0,$$
for some $t_n\in[0,1]$.
Extracting a subsequence, we find in the limit a critical point $(\tilde Y,\tilde\Psi)$ at the mountain pass level:
$$c=E_N(\tilde Y,Z),\quad H_N(\tilde Y,Z)\tilde \Psi=E_N(\tilde Y,Z)\tilde \Psi,\quad \pscal{\tilde \Psi,\nabla_YH_N\big(\tilde Y,Z\big)\tilde \Psi}=0.$$
\end{theorem}

The sequence $\Psi_n(t_n)$ is compact in $H^2_a$ since $E_N(Y_n(t_n),Z)$ stays below the essential spectrum and $Y_n(t_n)$ is bounded by assumption.

We remark that the energy $Y\mapsto \pscal{\Psi,H_N(Y,Z)\Psi}$ is $C^1$ in $Y$ but not necessarily $C^2$ for a fixed $\Psi\in H^2_a$. However the ground state energy $Y\mapsto E_N(Y,Z)$ is known to be real analytic~\cite{AveSei-75,ComSei-78,ComSei-80,ComDucSei-81,Hunziker-86} away from its degeneracies (at least without taking the spin into account). Therefore, if $\tilde\Psi$ is non degenerate, then using some Morse information as in~\cite{Ghoussoub-93}, we can get the additional property that the Hessian of $Y\mapsto E_N(Y,Z)$ at the mountain pass point $\tilde Y$ has at most one negative eigenvalue.

As usual in critical point theory, it is not obvious to pass to the limit in the sequence of paths $(Y_n(t),\Psi_n(t))_{t\in[0,1]}$, since the latter need not be equicontinuous in $t$. We however conjecture that there is always an optimal path (e.g. composed of gradient lines) on which the maximal energy is the transition state $(\tilde Y,\tilde \Psi)$. 

In the next section we consider the only setting for which the conjecture~\ref{conjecture} could be proven in some cases. Namely we assume that the molecule is composed of two rigid subsystems, which can only be rotated and translated with respect to each other.

\subsection{Case of two rigid molecules: model and some first cases}

We consider two rigid molecules, with nuclear positions $Y_1\in(\R^3)^{M_1}$, $Y_2\in (\R^3)^{M_2}$ and charges $Z_1\in\N^{M_1}$, $Z_2\in\N^{M_2}$. After an appropriate translation, we may assume that $0$ belongs to both $Y_1$ and $Y_2$. We then consider all possible ways of placing these two molecules in space. Without loss of generality, we can place the first molecule at the origin and the second one at a distance $L$ in the direction $e_1=(1,0,0)$, and simply rotate the two molecules using $U,V\in SO(3)$ (see Figure~\ref{fig:rigid_mol}). Our sole variables are therefore $(L,U,V)\in (0,\ii)\times SO(3)\times SO(3)$. For shortness we introduce the new variable 
$$\tau=(L,U,V)\in (0,\ii)\times SO(3)\times SO(3)$$
and denote by
\begin{equation}\label{def:Ytau}
Y(\tau)=(UY_1,VY_2+Le_1),
\end{equation}
the nuclear positions, as well as by
$$\boxed{\cE(\tau):=E_N\big(Y(\tau),Z\big)}$$
the corresponding ground state energy. As before, we use the notation
\begin{equation}\label{def:UY}
UY=U(y_1,\dots ,y_M):=(Uy_1,\dots ,Uy_M)\in(\R^3)^M.
\end{equation}

\begin{figure}[t]
\centering
\includegraphics[width=9cm]{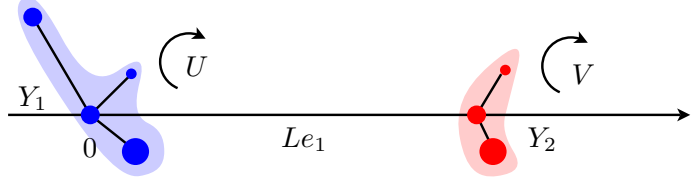}
\caption{The case of two rigid molecules.\label{fig:rigid_mol}}
\end{figure}

Now we assume that there exist two local strict minima $\tau_0,\tau_1\in (0,\ii)\times SO(3)\times SO(3)$ of $\cE$ and consider the mountain pass level
\begin{equation}
\boxed{ c:=\inf_{\substack{\tau(0)=\tau_0,\\ \tau(1)=\tau_1}}\;\max_{t\in[0,1]}\;\cE(\tau(t))}
 \label{eq:mountain_pass}
\end{equation}
where, as usual, it is understood that $\tau:[0,1]\to (0,\ii)\times SO(3)\times SO(3)$ is continuous. The assumption that the internal geometry of the two small molecules remains fixed during the chemical reaction is of course very restrictive, but it is reasonable in some particular cases (for example for the reaction HCN $\to$ CNH). We hope that this study will shed a new light on the problem and stimulate works with less stringent assumptions. 

From the results of~\cite{Morgan-79,MorSim-80,Lewin-04b}, we have
\begin{equation}
\boxed{\lim_{L\to\ii}\cE(\tau)=\min_{N_1+N_2=N}\big(E_{N_1}(Y_1,Z_2)+E_{N_2}(Y_2,Z_2)\big):=e_\ii}
\label{eq:limit_L_infinity}
\end{equation}
uniformly in $U,V\in SO(3)$. As we have said, it is a famous conjecture that the minimum on the right side is attained in the neutral case. If this is not the case, the mountain pass problem is actually rather easy, as stated in the following result. 

\begin{theorem}[Some first cases~{\cite[Thm.~4]{Lewin-04b}}]\label{thm:first_cases}
Assume that $N=Z_1+Z_2$ (neutral case). Assume that $\tau_0$ and $\tau_1$ are two local minima of $\tau\mapsto \cE(\tau)$ 
such that $c>\max\big\{ \cE(\tau_0),\cE(\tau_1)\big\}$ and define $e_\ii$ as in~\eqref{eq:limit_L_infinity}. If
\begin{itemize}
 \item[$\bullet$] \emph{either} $c\neq e_\ii$
 \item[$\bullet$] \emph{or} $c=e_\ii$ but there exists $N_1\neq |Z_1|$ such that $e_\ii=E_{N_1}(Y_1,Z_1)+E_{N_2}(Y_2,Z_2)$,
\end{itemize}
then one can find a compact min-maxing sequence of paths $\tau_n(t)$ for the min-max problem~\eqref{eq:mountain_pass}, that is, such that
$$L_n(t)\leq C$$
for all $n$ and all $t\in[0,1]$. In particular, there exists a mountain pass at the level $c$, satisfying similar properties as in Theorem~\ref{thm:existence_mountain_pass}, with $Y$ replaced by~$\tau$.
\end{theorem}

Since we will need similar arguments later, we quickly sketch the proof of Theorem~\ref{thm:first_cases} from~\cite{Lewin-04b}.

\begin{proof}[Sketch of the proof of Theorem~\ref{thm:first_cases}]
If $c<e_\ii$ then min-maxing sequences of paths can never have $L$ too large, and they are all compact. If $c>e_\ii$ then there might be non-compact min-maxing sequences, but one can remove the part going to infinity with an error that will not modify the maximal value of the energy on the path, which is larger than or equal to $c>e_\ii$. More precisely, let us choose $\lk$ so large that 
$$\max_{U,V\in SO(3)}\cE(\lk,U,V)\leq \frac{c+e_\ii}{2}<c$$
and $\max(L_0,L_1)<\lk$ for the two end points of the paths. 
If we have a path which has some $L(t)\geq \lk$, we may look at the first time $t_0$ and the last time $t_1$ for which $L(t)= \lk$. Then we replace the piece corresponding to $t\in[t_0,t_1]$ by a path which has $L(t)\equiv \lk$ on $[t_0,t_1]$ and connects $(U(t_0),V(t_0))$ to $(U(t_1),V(t_1))$ in any way, using the connectedness of $SO(3)$ (see Figure~\ref{fig:cut_path}). On this new piece, the maximal energy is $\leq (c+e_\ii)/2$, hence the global maximum on the path, which ought to be $\geq c>e_\ii$, is not changed. Using this method for any min-maxing sequence of paths, we get a new sequence which is compact.

The second case where $c=e_\ii$ but $e_\ii$ is attained in a non-neutral arrangement of electrons is similar. We now have
$$\max_{U,V\in SO(3)}\cE(U,V,L)\leq e_\ii-\frac{C}{L}+o\left(\frac1L\right)$$
as explained before in~\eqref{eq:split_non_neutral}.
By assumption we have $c=e_\ii$ and therefore we may choose $\lk$ such that 
$$\max_{U,V\in SO(3)}\cE(U,V,\lk)\leq c-\frac{C}{2\lk}<c.$$
The same surgery as before does not change the maximum value along the path.
\end{proof}

\begin{figure}[t]
\centering
\includegraphics[width=10cm]{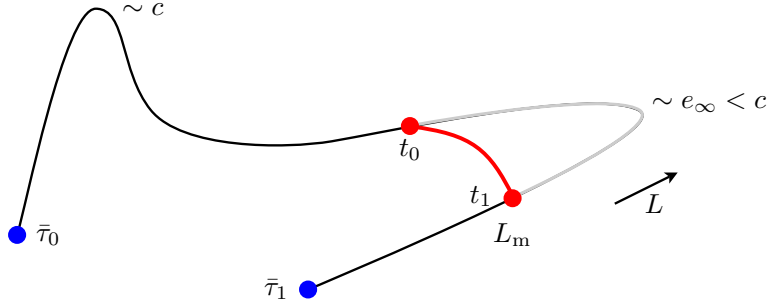}
\caption{Surgery for a path on which $L$ becomes too large, in the case when $c>e_\ii$.\label{fig:cut_path}}
\end{figure}

From the previous theorem, we see that it remains to study the case for which the mountain pass level equals the sum of the energies of the two neutral molecules, which is itself lower than all the non-neutral arrangements of electrons:
\begin{multline}
 \boxed{c=E_{|Z_1|}(Y_1,Z_1)+E_{|Z_2|}(Y_2,Z_2)<\min_{N_1'\neq |Z_1|}\big(E_{N'_1}(Y_1,Z_1)+E_{N'_2}(Y_2,Z_2)\big).}
 \label{eq:main_assumption}
\end{multline}
We work under this assumption in the rest of the paper and call for shortness
$$E_1:=E_{|Z_1|}(Y_1,Z_1),\qquad E_2:=E_{|Z_2|}(Y_2,Z_2)$$
the two ground state energies of the neutral sub-molecules, as well as 
$$H_1:=H_{|Z_1|}(Y_1,Z_1),\qquad H_2:=H_{|Z_2|}(Y_2,Z_2)$$
the Hamiltonians of the two neutral molecules.

\subsection{Case of two rigid molecules: new results}

In this section we state our new results, which generalize those in~\cite{Lewin-04b,Lewin-06}. 
We have to investigate the precise behavior of the system when the molecule splits. If the two molecules have multipoles whose interaction is of the order $L^{-p}$ with $p<6$, then the energy will be given by the corresponding term to leading order, and we have to study the critical points of this multipolar interaction. If all multipolar interactions with $p\leq 6$ vanish, then the leading term will be the van der Waals force and we may take advantage of the fact that it is always attractive, similarly as in the proof of Theorem~\ref{thm:first_cases}. 

\subsubsection{Expansion of the energy: multipoles and the van der Waals force}
In this section we expand the energy $\cE(L,U,V)$ up to order $L^{-6}$ and get the van der Waals energy as well as all the lower order multipolar energies. We therefore first need to define multipoles. One difficulty is that ground states can be degenerate, and there is then no canonical definition.

Suppose that we are given a bounded measure $\rho$ decaying faster than any polynomial at infinity. The multipoles of $\rho$ appear when expanding the associated Coulomb potential $\rho\ast|x|^{-1}$ at large distances. More precisely, the \emph{$2^n$--pole moment of $\rho$} is the tensor $\cM^{(n)}_\rho(h_1,\dots ,h_n)$ defined by
\begin{equation}\label{def:nmult}
\boxed{\cM^{(n)}_\rho(h_1,\dots ,h_n)=\frac{(-1)^n}{n!} \int_{\R^3} |z|^{2n+1} (h_1\cdot\nabla)\cdots (h_n\cdot\nabla)\left(\frac{1}{|z|}\right)\,{\rm d}\rho(z)}
\end{equation}
with the convention
$$\cM^{(0)}_\rho=\int_{\R^3}{\rm d}\rho(z),$$
see~\cite{StoSto-66} and~\cite[Section 4.2]{KhrSheMas-98}. 
Later in Section~\ref{sec:multipoles} we give explicit formulas for the dipole $D:=\cM_\rho^{(1)}$, the quadrupole $Q:=\cM_\rho^{(2)}$, the octopole $O:=\cM_\rho^{(3)}$ and the hexadecapole $\H:=\cM_\rho^{(4)}$. We will always identify the dipole $D:=\cM_\rho^{(1)}$ with a vector in $\R^3$ and the quadrupole $Q:=\cM_\rho^{(2)}$ with a $3\times3$ real symmetric matrix.
Using that for $z\neq0$ 
$$|Le_1-z| = L |z| \left|\frac{e_1}{|z|}-\frac{z}{L|z|}\right|= \frac{L}{|z|} \left|z-\frac{e_1}{L}|z|^2\right|,$$ 
Taylor's theorem gives
\begin{equation}\label{exp:Taylorformal}
\frac{1}{|Le_1-z|}=\frac{|z|}{L} \frac{1}{|z-\frac{e_1}{L}|z|^2|}=\frac{1}{L} \sum_{n=0}^{\infty}\frac{(-1)^n |z|^{2n+1}}{n! L^n} \left(e_1 \cdot \nabla \right)^n \frac{1}{|z|},
\end{equation}
a convergent series in a neighborhood of the origin. From this follows the (formal) multipolar expansion of the Coulomb potential
\begin{equation*}
\int \frac{{\rm d}\rho(z)}{|L e_1-z|}=\sum_{n=0}^\infty \frac{\cM^{(n)}_\rho(e_1,\dots ,e_1)}{L^{n+1}}.
\end{equation*}
Similarly, if we look at the classical Coulomb interaction of two measures $\rho_1$ and $\rho_2$, we have
\begin{equation}
\iint_{\R^3\times\R^3}\frac{{\rm d}\rho_1(x)\;{\rm d}\rho_2(y)}{|x-y+Le_1|}=\sum_{n,m\geq0}\frac{\cF^{(n,m)}(\rho_1,\rho_2,e_1)}{L^{n+m+1}}
\label{eq:def_multipole_expansion}
\end{equation}
where
\begin{multline}
\cF^{(n,m)}(\rho_1,\rho_2,e_1)=\frac{(-1)^m}{ (\prod_{j=1}^n(2j-1))\;(\prod_{i=1}^m(2i-1))}\sum_{j_1,\dots ,j_n}\sum_{k_1,\dots ,k_m}\\
\times \cM^{(n)}_{\rho_1}(e_{j_1},\dots ,e_{j_n})\cM^{(m)}_{\rho_2}(e_{k_1},\dots ,e_{k_m}) \left(\partial_{z_{j_1}} \cdots  \partial_{z_{k_m}}\frac{1}{|z+e_1|}\right)_{|z=0}
\label{def:nmult_int}
\end{multline}
depends on the two multipoles $\cM^{(n)}_{\rho_1}$ and $\cM^{(m)}_{\rho_1}$ as well as on the direction $e_1$. The expansion \eqref{eq:def_multipole_expansion} is well known (see for example ~\cite{BurBon-81}) but for convenience of the reader we sketch its proof in Appendix \ref{app:multipolar_expansion}. Later in Section~\ref{sec:multipoles} we give the exact expression of $\cF^{(n,m)}(\rho_1,\rho_2,e_1)$ for $n+m\leq5$.

Now we state a result which provides an upper bound to the energy in terms of two eigenfunctions of the neutral sub-molecules. 
To this end, we need to introduce the van der Waals correlation function. Let $\Psi_1$ and $\Psi_2$ be any two normalized eigenfunctions of, respectively, $H_1$ and $H_2$. We then introduce the corresponding rotated molecular densities
\begin{equation}
 \rho_{1,U}^{\Psi_1}(x):=\sum_{j=1}^{M_1}z_{1,j}\delta_{y_{1,j}}(U^{-1}x)-\rho_{\Psi_1}(U^{-1}x)
 \label{eq:def_rho_1_Psi}
\end{equation}
and
\begin{equation}
 \rho_{2,V}^{\Psi_2}(x):=\sum_{j=1}^{M_2}z_{2,j}\delta_{y_{2,j}}(V^{-1}x)-\rho_{\Psi_2}(V^{-1}x)
 \label{eq:def_rho_2_Psi}
\end{equation}
where we recall that the electronic density is defined by
\begin{multline}
\rho_\Psi(x)=N\sum_{s_1\in\{\pm1/2\}}\cdots \sum_{s_N\in\{\pm1/2\}}\\
\times \int_{\R^3}\cdots \int_{\R^3}|\Psi(x,s_1,x_2,s_2,\dots ,x_N,s_N)|^2\,{\rm d}x_2\cdots {\rm d}x_N. 
\end{multline}
Note that $\rho_{1,U}^{\Psi_1}$ and $\rho_{2,V}^{\Psi_2}$ depend on the chosen electronic states $\Psi_1$ and $\Psi_2$. If the first eigenvalue is degenerate, there are several non-equivalent possibilities. 
With a slight abuse of notation we write
$$\cF^{(n,m)}(\Psi_1,\Psi_2,U,V):=\cF^{(n,m)}(\rho_{1,U}^{\Psi_1},\rho_{2,V}^{\Psi_2})$$
for the corresponding multipolar energies, in order to emphasize the dependence on the wavefunctions. It is known that $\Psi_1$ and $\Psi_2$ decay exponentially at infinity, see e.g.~\cite{Combes-73,Simon-74,Hof-77,FroHer-82,Griesemer-04}. This implies that the multipoles $\cM^{(n)}$ and the multipolar interactions $\cF^{(n,m)}(\Psi_1,\Psi_2,U,V)$ are well defined for all $n,m\geq1$.

We will need a notion of non-degeneracy for the electronic eigenfunctions. Since the molecular Schrödinger Hamiltonian does not depend on the spin, the two kernels $\ker(H_1-E_1)$ and $\ker(H_2-E_2)$ are invariant under any permutation of the spins. We then say that the ground state eigenspace is irreducible when the group representation is itself irreducible.

\begin{definition}[Irreducibility]\label{def:irreducible}
For $k=1$ or $k=2$, we say that $H_k$ has an \emph{irreducible ground state eigenspace} when the only subspaces of $\ker(H_k-E_k)$ invariant under all the possible permutations of the spins are $\{0\}$ and the whole eigenspace $\ker(H_k-E_k)$. 
\end{definition}

We remark that this definition allows for a positive multiplicity. It also implies that for any given $\Psi_k\in\ker(H_k-E_k)$, 
$$\ker(H_k-E_k)={\rm span}\big(\pi\cdot \Psi_k,\ \pi\in\gS_{|Z_k|}\big)$$
where 
$$\pi\cdot \Psi_k(x_1,s_1,\dots ,x_N,s_N):=\Psi(x_1,s_{\pi^{-1}(1)},\dots ,x_N,s_{\pi^{-1}(N)}).$$
That is, all the vectors are cyclic and the ground state is indeed unique up to spin relabeling.

Let us denote by $\cG_k:=\ker(H_k-E_k)$ the two eigenspaces, by $\Pi_{12}$ the orthogonal projection onto $\cG_1\otimes \cG_2$ in the tensor space 
\begin{equation}
L^2_a((\R^3\times\{\pm1/2\})^{|Z_1|},\C)\otimes L^2_a((\R^3\times\{\pm1/2\})^{|Z_2|},\C)
\label{eq:tensor_space}
\end{equation}
and by $\Pi_{12}^\perp=1-\Pi_{12}$ its orthogonal complement.
We introduce as well the dipolar interaction function
\begin{multline}
f_{(U,V)}(x_1,\dots ,x_{N}):= UD_1(x_1,\dots ,x_{|Z_1|})\cdot VD_2(x_{|Z_1|+1},\dots ,x_N)\\
-3\left(e_1\cdot UD_1(x_1,\dots ,x_{|Z_1|})\right)\left(e_1\cdot VD_2(x_{|Z_1|+1},\dots ,x_N)\right)
\label{eq:def_f_U_V}
\end{multline}
where
\begin{align}
D_1(x_1,\dots ,x_{|Z_1|})&:=\sum_{j=1}^{|Z_1|}x_j-\sum_{m=1}^{M_1}z_{m}y_{m},\\
D_2(x_{|Z_1|+1},\dots ,x_N)&:=\sum_{j=|Z_1|+1}^{N}x_j-\sum_{m=M_1+1}^{M_1+M_2}z_{m}y_{m}
\end{align}
are the `instantaneous' dipoles of the two molecules. The \emph{van der Waals correlation function} is the positive function of $\Psi_1$, $\Psi_2$, $U$ and $V$ defined by
\begin{multline}
 C_{\rm vdW}(\Psi_1,\Psi_2,U,V):=\bigg\langle \Pi_{12}^\perp f_{(U,V)}\Psi_1\otimes \Psi_2,\\
 \big(H_1\otimes\1+\1\otimes H_2-E_1-E_2\big)^{-1}_{|(\cG_1\otimes\cG_2)^\perp}\Pi_{12}^\perp f_{(U,V)}\Psi_1\otimes\Psi_2\bigg\rangle,
 \label{eq:def_C_vdW_Psi_1_2}
\end{multline}
where the scalar product is in the tensor space~\eqref{eq:tensor_space}. We remark that although $\Psi_1\otimes\Psi_2$ belongs to the ground state space $\cG_1\otimes\cG_2$, the function $f_{(U,V)}\Psi_1\otimes\Psi_2$ can be shown to always have a component in its orthogonal complement (Proposition~\ref{prop:positivity_vdW} below). Therefore $\Pi_{12}^\perp f_{(U,V)}\Psi_1\otimes\Psi_2\neq0$ and $C_{\rm vdW}(\Psi_1,\Psi_2,U,V)>0$.
Then we have the following result.

\begin{theorem}[Multipolar/van der Waals expansion of the energy]\label{thm:expansion_energy}
Let $\Psi_1$ and $\Psi_2$ be any two normalized ground states of, respectively, $H_1$ and $H_2$. Then the function 
$$(U,V)\mapsto C_{\rm vdW}(\Psi_1,\Psi_2,U,V)$$
is continuous and positive on $SO(3)\times SO(3)$. 
We have the upper bound
\begin{multline}
 \cE(L,U,V)\leq E_1+E_2+\sum_{2\leq n+m\leq 5}\frac{\cF^{(n,m)}(\Psi_1,\Psi_2,U,V)}{L^{n+m+1}}\\
 -\frac{C_{\rm vdW}(\Psi_1,\Psi_2,U,V)}{L^6}+O\left(\frac{1}{L^7}\right)
 \label{eq:upper}
\end{multline}
where the $O(1/L^{7})$ is uniform in $U,V\in SO(3)$.
If in addition the two ground state eigenspaces are irreducible and
\begin{multline}
E_{|Z_1|}(Y_1,Z_1)+E_{|Z_2|}(Y_2,Z_2)<\min_{N_1'\neq |Z_1|}\Big(E_{N'_1}(Y_1,Z_1)+E_{N'_2}(Y_2,Z_2)\Big)
 \label{eq:main_assumption2}
\end{multline}
then~\eqref{eq:upper} is an equality, where each term is independent of the chosen $\Psi_1$ and $\Psi_2$.
\end{theorem}

This theorem gives an upper bound which is unconditional and is sharp in the irreducible case under assumption~\eqref{eq:main_assumption2}. It is an interesting question to find the expansion of the energy in the non-irreducible case. 

In the case of two molecules, our upper bound~\eqref{eq:upper} improves the upper bound~\eqref{eq:vdW_LT} obtained by Lieb and Thirring in \cite{LieThi-86}, where only the average over rotations $U,V\in SO(3)$ was considered. This has the effect of eliminating all the multipolar terms since
$$\int_{SO(3)}\cF^{(n,m)}(\Psi_1,\Psi_2,U,V)\,{\rm d}U=\int_{SO(3)}\cF^{(n,m)}(\Psi_1,\Psi_2,U,V)\,{\rm d}V=0$$
for all $n,m\geq1$. Lieb and Thirring have indeed considered the case of several molecules, but combining our method with the one used in \cite[Theorem 1.5]{Anapolitanos-16}, our upper bound~\eqref{eq:upper} can be generalized to a system of several molecules as well. The corresponding result is stated with a sketch of the proof in Appendix \ref{app:vdWseveral} for completeness. 
Note that results in the direction of Theorem~\ref{thm:expansion_energy} have been proven in the case of systems of several atoms in \cite{MorSim-80,AnaSig-17,Anapolitanos-16}. However, an upper bound of the form \eqref{eq:upper} without any assumptions on the multiplicity of the ground state energies did not appear there, even in the case of individual atoms.

We note that Theorem \ref{thm:expansion_energy} holds if the spin of the electrons is not taken into account or when the fermionic statistics is dropped. In the absence of spin and fermionic statistics, ground states are automatically non-degenerate (see, e.g.,~\cite[Sec.~XIII.12]{ReeSim4}), and thus \eqref{eq:upper} is always an equality, provided that~\eqref{eq:main_assumption2} holds.

In \cite{AnaSig-17} the energy expansion was provided for a system of atoms using the Feshbach map. In \cite[Theorem 1.5]{Anapolitanos-16} an upper bound similar to~\eqref{eq:upper} was provided for several atoms but without spin and with the irreducibility assumption. The test function used to get the bound was inspired by a formula for the ground states of the system given by the Feshbach map, and the idea was to replace the full resolvent of the system in the formula with the non-interacting one. Here we use similar ideas but we deal with the case of molecules, for which the multipolar energies then appear.

\subsubsection{The case with no multipoles $n+m\leq 5$ in average}

Next we come back to our mountain pass problem. First, we use the attractivity of the van der Waals force for all $U,V\in SO(3)$ and therefore require the vanishing of all the multipolar interactions for $n+m\leq5$  in~\eqref{eq:upper} for all $U,V\in SO(3)$. We are still free to choose the eigenfunctions $\Psi_1$ and $\Psi_2$ as we wish. Indeed, we may even average over several possibilities and consider a mixed state for each of the two molecules. Although we could require the vanishing of the appropriate number of multipoles for one (unknown) mixed state, checking the validity of this assumption for real systems could be difficult. For this reason, we will make one specific choice of mixed state, which seems physically reasonable and has good symmetry properties. As we will explain, this can be used to easily prove the vanishing of the appropriate number of multipoles.

We replace the pure state $|\Psi_k\rangle\langle\Psi_k|$ by the uniform average over the corresponding eigenspace
$$\frac{\1_{\{E_k\}}(H_k)}{r_k}=\frac{1}{r_k}\sum_{j=1}^{r_k}|\Psi_{k,j}\rangle\langle\Psi_{k,j}|$$ 
where $r_k:=\dim\ker(H_k-E_k)$ and $(\Psi_{k,j})_{j=1}^{r_k}$ is any chosen orthonormal basis of $\cG_k=\ker(H_k-E_k)$, for $k\in\{1,2\}$. We therefore call
\begin{equation}
 \rho_{1}(x):=\sum_{j=1}^{M_1}z_{1,j}\delta_{y_{1,j}}(x)-\frac1{r_1}\sum_{j=1}^{r_1}\rho_{\Psi_{1,j}}(x)
 \label{eq:def_rho_1}
\end{equation}
and
\begin{equation}
 \rho_{2}(x):=\sum_{j=1}^{M_2}z_{2,j}\delta_{y_{2,j}}(x)-\frac1{r_2}\sum_{j=1}^{r_2}\rho_{\Psi_{2,j}}(x)
 \label{eq:def_rho_2}
\end{equation}
the two corresponding averaged molecular densities. The second term is just the density of the trace-class operators $(r_k)^{-1}\1_{\{E_k\}}(H_k)$ and therefore the above two densities do not depend on the chosen orthonormal basis. They are canonical densities associated with each molecule. 
The rotated densities are then defined by 
$$\rho_{1,U}(x):=\rho_1(U^{-1}x),\qquad \rho_{2,V}(x):=\rho_2(V^{-1}x)$$
and we denote by $\cM^{(m)}_{\rho_{1,U}}$ and $\cM^{(n)}_{\rho_{2,V}}$ the corresponding multipoles, as well as by
$$\cF^{(n,m)}(U,V):=\cF^{(n,m)}(\rho_{1,U},\rho_{2,V})$$
the associated (averaged) multipolar energies. 

\begin{definition}[First non-zero average multipole]
For $k\in\{1,2\}$, we call $n_k$ the smallest integer $n\geq1$ such that $\cM^{(n)}_{\rho_{k}}\neq0$. If all the multipoles vanish, we let $n_k=+\ii$.
\end{definition}

Since the two molecules are neutral, we have $n_k\geq1$. We remark that if $H_k$ is invariant under the action of a subgroup $G$ of $O(3)$, then so is its first eigenspace, and therefore the corresponding averaged density $\rho_{k}(x)$. In particular the multipole moments of the $k$-th molecule are also invariant under $G$. It follows for instance that for an atom the density  $\rho_k$ is spherically symmetric and thus all its multipoles moments vanish: $n_k=+\ii$. 

\begin{example}
In the case of an oxygen molecule there is no average dipole moment because there is no vector that is invariant under the symmetry group of the oxygen molecule. Hence $n_k\geq2$. Similarly the methane molecule is symmetric with respect to the tetrahedral group and thus in this case the average dipole and quadrupole moments vanish: $n_k\geq3$. Indeed, the quadrupole moment is traceless so if it does not vanish it has a simple eigenvalue. But this cannot happen since the tetrahedral group does not have any invariant one dimensional subspace. In fact it is experimentally known that oxygen and methane have non-vanishing quadrupole and octopole, respectively. Therefore it is expected that $n_k=2$ and $n_k=3$ in these two cases.
\end{example}

We now denote by 
\begin{multline}
 C_{\rm vdW}(U,V):=\frac{1}{r_1r_2}\sum_{j_1=1}^{r_1}\sum_{j_2=1}^{r_2}\bigg\langle \Pi_{12}^\perp f_{(U,V)}\Psi_{1,j_1}\otimes \Psi_{2,j_2},\\
 \left(H_1\otimes\1+\1\otimes H_2-E_1-E_2\right)^{-1}_{|(\cG_1\otimes\cG_2)^\perp}\Pi_{12}^\perp f_{(U,V)}\Psi_{1,j_1}\otimes\Psi_{2,j_2}\bigg\rangle
 \label{eq:def_C_vdW}
\end{multline}
the averaged van der Waals correlation function. By averaging the upper bound~\eqref{eq:upper} we immediately get the similar inequality with averaged quantities
\begin{multline}
 \cE(L,U,V)\leq E_1+E_2+\sum_{2\leq n+m\leq 5}\frac{\cF^{(n,m)}(U,V)}{L^{n+m+1}}\\
 -\frac{C_{\rm vdW}(U,V)}{L^6}+O\left(\frac{1}{L^7}\right).
 \label{eq:upper_averaged}
\end{multline}
From this we can derive the following result.

\begin{theorem}[Compactness for $n_1+n_2\geq6$]\label{thm:no-multipoles_in_average}
Assume that $N=|Z_1|+|Z_2|$ (neutral case). 

Assume also that~\eqref{eq:main_assumption} holds. If 
$$\boxed{n_1+n_2\geq 6,}$$ 
then one can find a compact min-maxing sequence of paths $\tau_n(t)$ for the min-max problem~\eqref{eq:mountain_pass}, that is, such that
$$L_n(t)\leq C$$
for all $n$ and all $t\in[0,1]$. In particular, there exists a mountain pass at the level $c$, satisfying similar properties as in Theorem~\ref{thm:existence_mountain_pass}, with $Y$ replaced by~$\tau$.
\end{theorem}

Theorem~\ref{thm:no-multipoles_in_average} applies to the case when one of the two sub-molecules is an atom, since then $n_k=+\ii$. This was the main result of~\cite{Lewin-06}. All the other cases are new. For example, as we explained above in the case of methane-methane we have $n_i \geq 3$ so  $n_1+n_2 \geq 6$. Thus two methane molecules attract each other unconditionally with interaction energy at least as strong as $L^{-6}$ and Theorem \ref{thm:no-multipoles_in_average} applies.
The proof of the theorem goes along the same lines as that of Theorem~\ref{thm:first_cases}, with the help of the energy expansion in Theorem~\ref{thm:expansion_energy}.

\begin{proof}
The assumption $n_1+n_2\geq 6$ and Theorem~\ref{thm:expansion_energy} imply that 
$$\max_{U,V\in SO(3)}\cE(L,U,V)\leq E_1+E_2-\frac{\min_{U,V\in SO(3)}C_{\rm vdW}(U,V)}{L^6}+O\left(\frac{1}{L^7}\right).$$
However, from Theorem~\ref{thm:expansion_energy}, the function  $U,V\mapsto C_{\rm vdW}(U,V)$ is continuous and positive on the compact set $SO(3)^2$, hence its minimum $C$ is positive. Therefore, we may find a large enough $\lk$ such that 
$$\max_{U,V\in SO(3)}\cE(\lk,U,V)\leq E_1+E_2-\frac{C}{2(\lk)^6}$$
and the argument is the same as in the proof of Theorem~\ref{thm:first_cases}.
\end{proof}

\begin{remark}
Theorem~\ref{thm:no-multipoles_in_average} holds the same if we assume that there exist some possibly non-uniform mixed states $\sum_{j=1}^{r_k}\alpha_{k,j}|\Psi_{k,j}\rangle\langle\Psi_{k,j}|$ for which the corresponding multipoles vanish for all $n,m$ with $n+m \leq 5$. 
\end{remark}

\subsubsection{The irreducible case with $n_1+n_2< 5$}
Next we turn to the case when the two sub-molecules have multipoles such that $n+m\leq 5$. This is the hardest part of the proof, where it becomes necessary to study some fine properties of the multipolar interaction energies. In this case we will need the additional assumption that the two molecules have irreducible ground state eigenspaces (Definition~\ref{def:irreducible}). We are able to treat all cases for $n_1+n_2<5$ with a reasonable additional non-degeneracy assumption on the octopole moment. We are unfortunately not able to treat the case $n_1+n_2=5$ where the leading multipolar term is of the same order as the van der Waals interaction. We are indeed using two completely different proofs for $n_1+n_2<5$ and $n_1+n_2\geq 6$ and we are not able to match them when $n_1+n_2=5$.

\begin{theorem}[Compactness for $n_1+n_2<5$ in the irreducible case]\label{thm:with-multipoles}
Assume that $N=Z_1+Z_2$ (neutral case) and that the two Hamiltonians $H_1$ and $H_2$ have irreducible ground state eigenspaces, such that 
$$\boxed{n_1+n_2<5.}$$
Assume that $\tau_0$ and $\tau_1$ are two local minima of $\tau\mapsto \cE(\tau)$ such that $c>\max\big\{ \cE(\tau_0),\cE(\tau_1)\big\}$. 
Assume also that~\eqref{eq:main_assumption} holds. If $n_k=3$ for $k=1$ or $k=2$, we make the additional non-degeneracy assumption on the octopole
\begin{equation}
\boxed{O_k(v,\cdot,\cdot)\equiv 0\implies v=0.}
\label{eq:octopole_assumption}
\end{equation}
Then one can find a compact min-maxing sequence of paths $\tau_n(t)$ for the min-max problem~\eqref{eq:mountain_pass}, that is, such that
$$L_n(t)\leq C$$
for all $n$ and all $t\in[0,1]$. In particular, there exists a mountain pass at the level $c$, satisfying similar properties as in Theorem~\ref{thm:existence_mountain_pass}, with $Y$ replaced by~$\tau$. 
\end{theorem}

Theorem \ref{thm:with-multipoles} is an extension of  ~\cite[Theorem~7]{Lewin-04b} where the case $n_1=n_2=1$ was considered, without spin. All the other cases are new. Our proof is based on a careful study of the properties of the critical points of the functions $(U,V)\mapsto \cF^{(n,m)}(U,V)$ for all $n+m<5$.

\begin{remark}
A similar theorem holds in the case that we do not take spin into account. In this case the nondegeneracy assumption means that the ground state energy is a simple eigenvalue. If we neglect completely the fermionic statistics, i.e., the electrons are treated as bosons, then a similar theorem holds but the non-degeneracy assumption of the ground state energies is automatically satisfied, as mentioned before.  
\end{remark}

Before we explain the main strategy of the proof, we comment on the non-degeneracy assumption~\eqref{eq:octopole_assumption} on the octopole, in case this is the first non-vanishing multipole. Typically the dipole and quadrupole moments vanish only if the molecule is symmetric with respect to a  finite subgroup $G$ of $O(3)$ that does not have one-dimensional invariant subspaces. As mentioned before, the octopole $O_k$ is also invariant under $G$.   But if $O_k(v,\cdot,\cdot)=0$, and $G$ does not have invariant one dimensional subspaces, it follows that $O_k\equiv 0$ hence $n_k\neq 3$, i.e. the octopole moment cannot be the leading multipole moment. Hence in this situation the assumption~\eqref{eq:octopole_assumption} is always satisfied.

\begin{lemma}[Symmetry and non-degeneracy of the octopole]\label{lem:physas}
If the octopole $O\neq0$ is invariant under a symmetry group $G\subset SO(3)$ that does not have one-dimensional invariant subspaces, then~\eqref{eq:octopole_assumption} is satisfied.
\end{lemma}
\begin{proof}
Assume that there exists a $v_1 \neq 0$ so that $O(v_1,\cdot,\cdot)=0$. The $G$-invariance implies that $O(Uv_1,\cdot,\cdot)=0$ for all $U\in G$. 
If $G\subset SO(3)$ has no one-dimensional invariant subspace, then it must act irreducibly on $\R^3$, because it can also not have any two-dimensional invariant subspaces, by taking the orthogonal complement. This implies that $Uv_1$ spans the whole of $\R^3$, that is, $v_1$ is cyclic. Hence we conclude that $O(v,\cdot,\cdot)\equiv0$ for all $v\in\R^3$, or equivalently $O\equiv0$, which is a contradiction. 
\end{proof}

\subsection{Sketch of the proof of Theorem~\ref{thm:with-multipoles}}
In this section we give the proof of Theorem~\ref{thm:with-multipoles}, but defer the proof of many intermediate results to the rest of the paper. We need to study the critical points of the multipolar interactions, which is rather lengthy.

We consider $\tau_n(t)=(L_n(t),U_n(t),V_n(t))$ a min-maxing sequence of paths for the mountain pass problem~\eqref{eq:mountain_pass} and a large distance 
$$\lk>\max\{L_n(0),L_n(1)\}$$ 
 (the initial and final distances of the molecules, which actually do not depend on $n$) to be chosen later. As in the proof of Theorem~\ref{thm:first_cases}, we look at the first time $t_0$ and the last time $t_1$ for which $L_n(t)=\lk$. These times depend on $\lk$ and $n$, but we do not emphasize this in our notation, unless needed for clarity. 

Our goal is to replace the path by a new one on $[t_0,t_1]$, on which $L(t)\equiv \lk$ but the difficulty is to link the two molecular orientations $(U_n(t_0),V_n(t_0))$ and $(U_n(t_1),V_n(t_1))$ by a continuous path in $SO(3)\times SO(3)$ on which the energy does not increase too much. More precisely, we want to find a continuous path $(U(t),V(t))_{t\in[t_0,t_1]}$ of rotations such that
$$(U(t_0),V(t_0))=(U_n(t_0),V_n(t_0)),\qquad (U(t_1),V(t_1))=(U_n(t_1),V_n(t_1)),$$
and
\begin{multline}
 \max_{t\in[t_0,t_1]}\cE(\lk,U(t),V(t))\\
 \leq \max\big\{\cE(\lk,U_n(t_{0}),V_n(t_{0})),\cE(\lk,U_n(t_{0}),V_n(t_{0}))\big\}+\frac1n.
 \label{eq:goal_path_rotations}
\end{multline}
An arbitrary path will not work here since there are orientations for which the energy is $>c=e_\ii$, contrary to the cases handled in Theorems~\ref{thm:first_cases} and~\ref{thm:no-multipoles_in_average} where the energy was always $<e_\ii$ at infinity. We proceed as follows. 

\subsubsection*{\bf Step 1. Replacing the two points by local (pseudo) minima} 
Our first step is to rotate the molecules at each point $t_0,t_1$ without changing the distance $\lk$ so that we reach a kind of local minimum of the energy of the system with respect to rotations. In general this is impossible if in all directions the energy takes values higher and lower than that of the current point, at arbitrary small distances. This pathological behavior should not occur for our function $(U,V)\mapsto \cE(\lk,U,V)$ which is real-analytic~\cite{Hunziker-86} under appropriate non-degeneracy assumptions. Here we use a softer argument which does not rely on analyticity. 

\begin{definition}[Local pseudo-minimum]
Let $\cE$ be any continuous function on $SO(3)\times SO(3)$. A point $(\bar U, \bar V) \in SO(3)\times SO(3)$ is called a \emph{local pseudo-minimum of $\cE$} if there exists a small constant $\eps>0$ such that for every normalized vector $v$ in the plane tangent to $(\bar U,\bar V)$ the associated geodesic $\gamma_v:[0,\eps]\to SO(3)\times SO(3)$ in the direction $v$ satisfies 
$$\cE(\gamma_v(t_n)) \geq \cE(\gamma_v(0))$$
for a sequence $t_n\to 0^+$, $t_n\neq0$. 
\end{definition}

In other words, a local pseudo-minimum is a point at which one cannot find a direction in which the energy strictly decreases. If $\cE$ is $C^2$, one must then have 
$$\nabla\cE(\bar U,\bar V)=0,\qquad {\rm Hess}\;\cE(\bar U,\bar V)\geq0$$
but $(\bar U,\bar V)$ does not need to be a true local minimum. Think of $f(x)=x^4\sin(1/x) \chi_{\{x<0\}}(x) + x^4 \chi_{\{x>0\}}(x)$, where $\chi_A$ denotes the characteristic function of $A$. Then $f$ has a local pseudo-minimum at zero. 

That we can reach a local pseudo-minimum can be shown with the help of the following lemma, which we state for simplicity on $SO(3)\times SO(3)$ but is of course much more general.

\begin{lemma}[Linking any point to a local pseudo-minimum]\label{lem:flowpseudolocmin}
Let $\cE$ be a continuous function on $SO(3)\times SO(3)$. Let $(U,V)$ be a point that is not a local pseudo-minimum of $\cE$. Then there exists a continuous path $\tau$ on $SO(3)\times SO(3)$ linking $(U,V)$ to a local pseudo-minimum $(\bar U,\bar V)$, such that the maximum value of $\cE$ on this path is $\cE(U,V)$.
\end{lemma}

\begin{proof}
From the assumption, the energy decreases in at least one direction $v$ when we start from the point $X=(U,V)$. Let therefore $\gamma_v$ be the associated geodesic and $X'=(U',V'):=\gamma_v(\eps)$ for $\eps$ small enough such that $\cE(\gamma_v(t))<\cE(U,V)$ for all $t\in[0,\eps]$. In particular $\cE(U',V')<\cE(U,V)$. 

Next we look at the set $A$ of all the points $Y\in SO(3)\times SO(3)$ that can be reached from $X$ by a continuous path on which the energy does not exceed $\cE(X)$. Let $a=\inf\{\cE(Y),\ Y\in A\}$. Since $X'\in A$ and $\cE(X')<\cE(X)$, we have $a<\cE(X)$. We consider now a sequence $Y_n$ in $A$ with $a\leq \cE(Y_n)\leq a + 1/n$. After passing to a subsequence if necessary, we have $Y_n \rightarrow Y$ and $\cE(Y)=a$, by continuity. For $n$ large enough, $Y_n$ and $Y$ are so close that we can connect them by a geodesic along which $\cE$ varies so little around the value $a<\cE(X)$, that it can never exceed $\cE(X)$. Hence $Y$ is actually in $A$, which is thus a closed set.

Now $Y$ has to be a local pseudo-minimum of $\cE$, otherwise we would be able to find a new direction where the energy decreases further, contradicting the definition of $a$.
\end{proof}

Using Lemma \ref{lem:flowpseudolocmin}  we can connect the point $(\U(t_0), \V(t_0))$ to a local pseudo-minimum $(\U_0',\V_0')$ of $\cE(L,\cdot,\cdot)$, with a path along which the energy stays below $\cE(L,\U(t_0), \V(t_0))$. We proceed analogously for $t_1$ and connect $(U(t_1),V(t_1))$ to a local pseudo-minimum $(\U_1',\V_1')$. In the next steps we will study the properties of these two points and finally connect them by a path, without increasing the energy too much.

\subsubsection*{\bf Step 2. Properties of the leading multipolar interaction at the new points $(U'_0,V'_0)$ and $(U'_1,V'_1)$}
In step 1 we have managed to reach two points $(U'_0,V'_0)$ and $(U'_1,V'_1)$ which are local pseudo-minima of the energy $\cE(\lk,\cdot,\cdot)$. In this step we prove that 
$$\cE(\lk,U'_0,V'_0)<c=e_\ii,\qquad \cE(\lk,U'_1,V'_1)<c=e_\ii,$$
that is, the points have an energy below the dissociation threshold, provided we had chosen $\lk$ large enough (independently of $U$ and $V$). More precisely, we show that 
$$\cF^{(n_1,n_2)}(U'_0,V'_0)<0,\qquad \cF^{(n_1,n_2)}(U'_1,V'_1)<0$$
where we recall that $\cM^{(n_1)}_{\rho_1}$ and $\cM^{(n_2)}_{\rho_2}$ are the first non-vanishing multipoles of the two neutral sub-molecules. The estimate on the energy follows from the expansion in Theorem~\ref{thm:expansion_energy}.

\begin{proposition}[Leading multipolar energy at a local pseudo-minimum]\label{prop:fullexp}
We make the same assumptions as in Theorem~\ref{thm:with-multipoles}. There exists an $L_0>0$ and a constant $C$ such that, if $L>L_0$ and $(\U',\V')$ is any local pseudo-minimum of $\cE(L,\cdot,\cdot)$, then 
\begin{equation}\label{eq:almostlocmin}
\left|\nabla_{U,V}\cF^{(n_1,n_2)}(\U',\V')\right| \leq \frac{C}{L}, \qquad  {\rm Hess}_{U,V}\;\cF^{(n_1,n_2)}(\U',\V') \geq -\frac{C}{L}.
\end{equation} 
\end{proposition}

Intuitively, Proposition \ref{prop:fullexp} says that a local minimum of $\cE(L,\cdot,\cdot)$ is close to being a local minimum of the leading multipole-multipole term $\cF^{(n_1,n_2)}$. 
One then expects that $\cF^{(n_1,n_2)}$, will be negative at such a point because its average over all orientations is zero. More precisely, the following says that close to critical points of $\cF^{(n,m)}$ with a vanishing Morse index we have a negative energy, for $n+m\leq 4$. 

\begin{proposition}[Critical points of vanishing Morse index have a negative multipolar energy]\label{prop:localmin}
Let $n,m \in \mathbb{N}$ with $n+m \in \{2,3,4\}$. If $n=3$ or $m=3$ we also assume that the octopole satisfies the non-degeneracy assumption~\eqref{eq:octopole_assumption}. Then
 there exists $\delta>0$  such that if $|\nabla_{\U,\V}\cF^{(n,m)}(\U,\V)| \leq \delta$ and  ${\rm Hess}\;\cF^{(n,m)}(\U,\V) \geq -\delta$ for some $U,V\in SO(3)$, then 
\begin{equation}\label{eq:leadingneggen}
\cF^{(n,m)}(\U,\V) \leq -\delta. 
\end{equation}
\end{proposition}

Using Propositions~\ref{prop:fullexp} and~\ref{prop:localmin} we obtain that for $\lk$ large enough (independently of the orientations), there exists $\delta>0$ such that 
\begin{equation}\label{eq:leadingneg}
	\cF^{(n_1,n_2)}(\U_0',\V_0') \leq -\delta<0, \quad   	\cF^{(n_1,n_2)}(\U_1',\V_1') \leq -\delta<0. 
\end{equation}

\subsubsection*{\bf Step 3. Linking the new points $(U'_0,V'_0)$ and $(U'_1,V'_1)$}
We have so far managed to show that the leading multipole-multipole interaction of the two new points $(U'_0,V'_0)$ and $(U'_1,V'_1)$ is negative. We now connect $(\U_0',\V_0')$ and $(\U_1',\V_1')$ with a path along which $\cF^{(n_1,n_2)}$ stays negative, which then implies $\cE(\lk,\cdot,\cdot)<c$, by Theorem~\ref{thm:expansion_energy}. Note that for this part the non-degeneracy assumption \eqref{eq:octopole_assumption} for the octopole moment is not needed.

\begin{proposition}[Connectedness of $\{\cF^{(n,m)}\leq -\delta\}$]\label{prop:connected}
Let $n,m \in \mathbb{N}$ with $n+m \in \{2,3,4\}$. Then there exists $\delta_0>0$ such that for all $0<\delta<\delta_0$ the set $\{(\U,\V) \in SO(3): \cF^{(n,m)}(\U,\V) <-\delta\}$ is nonempty and pathwise connected.
\end{proposition}

Due to Proposition \ref{prop:connected} and \eqref{eq:leadingneg} we can connect $(\U_0',\V_0')$ and $(\U_1',\V_1')$ with a path along which $\cF^{(n_1,n_2)}$ stays below a negative constant independent of $\lk$. Thus if  $\lk$ is chosen large enough in the beginning, it follows that along this path the interaction energy $\cE(L,\cdot,\cdot)-E_1-E_2$ remains negative, which ends the proof of Theorem~\ref{thm:with-multipoles}.\qed

\begin{figure}[t]
\centering
\includegraphics[width=10cm]{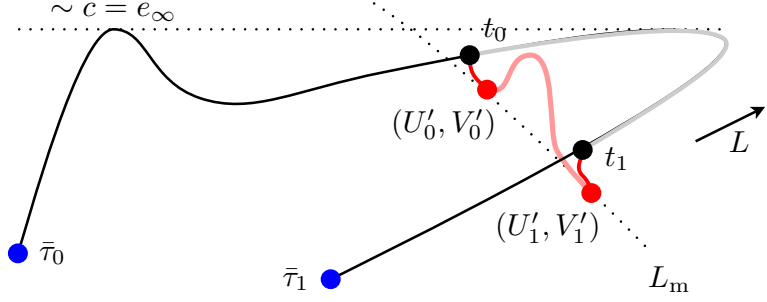}
\caption{Graphical representation of the proof of Theorem~\ref{thm:with-multipoles}.\label{fig:cut_path2}}
\end{figure}

\bigskip

Note that Propositions \ref{prop:localmin} and \ref{prop:connected} are results involving the Coulomb potential and thus can be used for other types of kinetic energies. 

The paper is now organized as follows. In Section \ref{sec:ProofvdWdom} we prove Theorem \ref{thm:expansion_energy}. In Section \ref{sec:Proofevergyexpsmooth} we prove Proposition \ref{prop:fullexp}. In Section \ref{sec:prooflocmincon} we prove Propositions \ref{prop:localmin} and \ref{prop:connected}. Since the proof of Propositions \ref{prop:localmin} and \ref{prop:connected} is rather lengthy, we treat the different multipolar interactions in separate subsections. 

\section{Proof of Theorem  \ref{thm:expansion_energy}: expansion of the energy}\label{sec:ProofvdWdom}

\subsection{Multipolar expansion}\label{sec:multipoles}

Our first goal is to provide an expansion of the Coulomb energy of two neutral densities placed at a distance $L\gg1$. For this we simply need to use Taylor's formula~\eqref{exp:Taylorformal}. Recalling the definition of the multipoles $\cM_\rho^{(n)}$ in~\eqref{def:nmult}, the precise statement is as follows.

\begin{lemma}[Expansion of the Coulomb potential]\label{lem:multipexp}
	For any $N \geq1$ there exists a constant $C_N$ such that, for all $h \in \mathbb{R}^3, L >0$ with $|h| \leq \frac{L}{2}$ and for all multi-indices $\alpha$ with $ 0 \leq |\alpha|\leq 2$,
	\begin{equation}
	\bigg|\partial_h^\alpha \bigg( \frac{1}{|Le_1-h|}-
	\sum_{n=0}^N \frac{\cM^{(n)}_{\delta_h}(e_1,\dots ,e_1)}{L^{n+1}}  \bigg) \bigg|\leq C_N \frac{ 1 + |h|^{N+1}}{L^{N+2}}
	\label{eq:expansion_Coulomb},
	\end{equation}
	where $\delta_h$ is the Dirac delta at $h$ and where we recall that the multipole $\cM^{(n)}_{\rho}$ is defined in~\eqref{def:nmult}.
\end{lemma}

The terms on the left hand side of \eqref{eq:expansion_Coulomb} can be computed with the help of formula \eqref{def:nmult}. For $N=5$ we find 
\begin{multline}
	\sum_{n=0}^5 \frac{\cM^{(n)}_{\delta_h}(e_1,\dots ,e_1)}{L^{n+1}} =\frac{1}{L}+\frac{h \cdot e_1}{L^2}+ \frac{3(h \cdot e_1)^2-|h|^2}{2 L^3}\\
	+\frac{5(h \cdot e_1)^3-3 (h \cdot e_1)|h|^2 }{2 L^4}
	+ \frac{3 |h|^4- 30 (h \cdot e_1)^2 |h|^2 +35 (h \cdot e_1)^4}{ 8 L^5 }\\
	+ \frac{15 (h \cdot e_1) |h|^4- 70 (h \cdot e_1)^3 |h|^2+63 (h \cdot e_1)^5}{8 L^6}.
\end{multline}
We can now give the expression of the first four multipole moments of a neutral charge distribution $\rho$:
\begin{align}
 &\D(\rho)_i=\big(\cM^{(1)}_\rho\big)_i= \int_{\R^3} z_i \,{\rm d}\rho(z)\nn\\
 &\Q(\rho)_{ij}=\big(\cM^{(2)}_\rho\big)_{ij} =\frac{1}{2} \int_{\R^3} (3 z_i z_j -|z|^2 \delta_{ij}) \,{\rm d}\rho(z)\nn \\
 &\O(\rho)_{ijk}=\big(\cM^{(3)}_\rho\big)_{ijk}= \frac{1}{2} \int_{\R^3} ( 5 z_i z_j z_k - |z|^2  (z_i \delta_{jk}+ z_j \delta_{ik}+ z_k \delta_{ij})) \,{\rm d}\rho(z)\nn\\
 &\H(\rho)_{ijk\ell}=\big(\cM^{(4)}_\rho\big)_{ijk\ell}= \frac{1}{8} \int_{\R^3}\Big(35 z_i z_j z_k z_\ell\nn\\
  &\qquad\qquad\qquad\qquad\qquad- 30 |z|^2 P_S(z_i z_j \delta_{k\ell}) + 3 |z|^4 P_S(\delta_{ij} \delta_{k\ell})\Big)\,{\rm d}\rho(z),
  \label{def:multipoles}
\end{align}
where
$$P_S(z_i z_j \delta_{k\ell}):=\frac{z_i z_j \delta_{k\ell} + z_i z_k \delta_{j\ell} + z_i z_\ell \delta_{jk} + z_j z_k \delta_{i\ell} + z_j z_\ell \delta_{ik} + z_k z_\ell \delta_{ij}}{6}$$
and
$$P_S(\delta_{ij} \delta_{k\ell}):=\frac{\delta_{ij} \delta_{k\ell} + \delta_{ik} \delta_{j\ell} + \delta_{i\ell} \delta_{jk}}{3}.$$
In other words $P_S$ symmetrizes a given expression with respect to the indices $i,j,k,\ell$. 

From the Taylor expansion in Lemma~\ref{lem:multipexp} we can deduce an expansion for the interaction of two charge distributions placed far away. 

\begin{lemma}[Multipolar expansion]\label{lem:multipolar_expansion}
Let $\rho_1$ and $\rho_2$ be two bounded measures with support in the ball $B(0,L/3)$, such that $\rho_1(\R^3)=\rho_2(\R^3)=0$. For any $N\geq1$ there exists a constant $C$ such that
	\begin{multline}
	\bigg\|\int_{\R^3}\int_{\R^3}\frac{{\rm d}\rho_1(x)\,{\rm d}\rho_2(y)}{|Ux-Vy-Le_1|}-\sum_{2\leq n+m\leq N}\frac{\cF^{(n,m)}(\rho_{1,U},\rho_{2,V},e_1)}{L^{n+m+1}}\bigg\|_{C^2(SO(3)^2,\R)}\\
\leq \frac{C}{L^{N+2}}\int_{\R^3}(1+|x|^{N+1})\,{\rm d}|\rho_1|(x)\; \int_{\R^3}(1+|y|^{N+1})\,{\rm d}|\rho_2|(y).
\label{eq:multipolar_expansion}
	\end{multline}
We recall that the multipolar interactions $\cF^{(n,m)}(\rho_{1,U},\rho_{2,V},e_1)$ are defined above in~\eqref{def:nmult_int}.
\end{lemma}

We recall that  we use the shorter notation $\rho_{1,U}(x):=\rho_1(U^{-1}x)$ and $\rho_{2,V}(x):=\rho_2(V^{-1}x)$.
We have stated in Lemma~\ref{lem:multipolar_expansion} only the expansion for the interaction and its first two derivatives with respect to $U,V\in SO(3)$ because this is what we will use later on. But the interaction is actually smooth in $U,V$ and the expansion~\eqref{eq:multipolar_expansion} holds in $C^k(SO(3)^2)$ for all $k\geq2$. The proof of the lemma is just Taylor's expansion~\eqref{eq:expansion_Coulomb} together with some algebraic manipulations (based on the neutrality of the two charges) to express everything in terms of the multipolar energies $\cF^{(n,m)}(\rho_{1,U},\rho_{2,V},e_1)$. The proof of Lemmas~\ref{lem:multipexp} and~\ref{lem:multipolar_expansion} is provided for completeness in Appendix~\ref{app:multipolar_expansion}.

For $n+m\leq5$, the multipolar interaction energies $\cF^{(n,m)}(\rho_1,\rho_2,e_1)$ appearing in the expansion~\eqref{eq:def_multipole_expansion} are given by the following formulas, whose derivation is also explained in Appendix~\ref{app:multipolar_expansion}. The dipole-dipole interaction is defined by
\begin{equation}\label{eq:dd}
\cF^{(1,1)}(\rho_1,\rho_2,e_1)= -3 (D_{1} \cdot e_1) (e_1 \cdot  D_{2}) + D_{1} \cdot  D_{2},
\end{equation}
with an obvious notation for the dipole $D_k$ of the charge distribution $\rho_k$. 
Similarly, the  dipole-quadrupole interaction term is
\begin{equation}\label{eq:dq}
\cF^{(1,2)}(\rho_1,\rho_2,e_1)=  5 (D_{1} \cdot e_1) Q_{2}(e_1, e_1) - 2 Q_{2}(D_{1},e_1),
\end{equation}
the dipole-octopole interaction term is
\begin{equation}\label{eq:do}
\cF^{(1,3)}(\rho_1,\rho_2,e_1) = O_{2}\big(e_1, e_1, 3 D_{1} - 7 (e_1\cdot  D_{1})e_1\big),
\end{equation}
and the quadrupole-quadrupole interaction term is
\begin{equation}\label{eq:qq}
\cF^{(2,2)}(\rho_1,\rho_2,e_1)= \frac{1}{3}\; \tr\Big(\big(35 p Q_{1} p-10 p Q_{1} -10 Q_{1} p + 2 Q_{1}\big)Q_{2}\Big).
\end{equation}
where $p:=|e_1\rangle\langle e_1|$. Finally, the dipole-hexadecapole and quadrupole-octopole terms are defined by
\begin{equation}\label{eq:dh}
\cF^{(1,4)}(\rho_1,\rho_2,e_1) = 9 \H_{2}(e_1, e_1, e_1, e_1) (\D_{1} \cdot e_1)-4 \H_{2}(e_1, e_1, e_1, D_{1})
\end{equation}
and
\begin{multline}\label{eq:qo}
\cF^{(2,3)}(\rho_1,\rho_2,e_1) = - 21\, O_{2}(e_1, e_1, e_1) Q_{1}(e_1, e_1) + 14\, O_{2}(e_1, e_1, Q_{1} e_1)\\ - 2 \tr(O_{2}(e_1,.,.) Q_{1}).
\end{multline}
Here we have only defined the multipolar energies $\cF^{(n,m)}(\rho_1,\rho_2,e_1)$ for $n\leq m$. For $n\geq m$ we have by definition
$$\cF^{(n,m)}(\rho_1,\rho_2,e_1):=(-1)^{n+m}\cF^{(m,n)}(\rho_2,\rho_1,e_1).$$

\subsection{Proof of Theorem \ref{thm:expansion_energy}: upper bound and positivity of $C_{\rm vdW}$}

In this section we provide the proof of the upper bound~\eqref{eq:upper} and of the positivity of $C_{\rm vdW}$.

\subsubsection{Localizing the eigenspaces}
To prove Theorem \ref{thm:expansion_energy} it turns out to be convenient to introduce a cut-off function. Let $\chi_1: \mathbb{R}^3\rightarrow \mathbb{R}$ be a spherically symmetric $C^\infty$ function
supported in the ball $B(0,\frac{1}{4})$ and equal to $1$ on the ball
$B(0,\frac{1}{5})$ with $ 0 \leq \chi_1 \leq 1$. Define then
\begin{equation}\label{chiR}
\chi_L(x):=\chi_1(L^{-1} x).
\end{equation}
For all multi-indices $\gamma$ there exists a constant $C>0$, so that
for all $L$ large enough we have that
\begin{equation}\label{chiader}
\|\partial^\gamma \chi_L\|_{L^\infty} \leq C \al^{-|\gamma|}.
\end{equation}
Next we consider the eigenspaces $\cG_k:=\ker(H_k-E_k)$ for $k=1,2$ and call 
$$\cG_{k,L}:=(\chi_L)^{\otimes N_k}\cG_k$$ 
the localized eigenspaces. 

As pointed out in the introduction the ground states of the Hamiltonians $H_k$ are exponentially decaying, namely there exists $c>0$ such that
\begin{equation}\label{eq:expdecay}
H_k \Psi= E_k \Psi \implies \|e^{c|x|} \partial^\alpha\Psi\|_{L^2} < \infty, \quad |\alpha|\leq2,\quad k=1,2
\end{equation}
where $\alpha$ denotes a multi-index. To any orthonormal basis $\Psi_{k,1},...,\Psi_{k,r_k}$ of $\cG_k$, we can thus associate the basis $(\chi_L)^{\otimes N_k}\Psi_{k,1},...,(\chi_L)^{\otimes N_k}\Psi_{k,r_k}$ of $\cG_{k,L}$, which is exponentially close to being orthonormal. If we introduce the two projectors $\Pi_k$ and $\Pi_{k,L}$ onto $\cG_k$ and $\cG_{k,L}$ respectively, then we obtain 
\begin{equation}
\norm{(1-\Delta)(\Pi_k-\Pi_{k,L})(1-\Delta)}=O(e^{-cL})
\label{eq:norm_diff_projections}
\end{equation}
in operator norm. Since by Theorem \ref{thm:molbind}  $E_k$ is in the discrete spectrum of $H_k$ we have that 
\begin{equation}\label{eq:trivialgap}
\Pi_k^\bot (H_k-E_k) \Pi_k^\bot \geq \delta \Pi_k^\bot \text{ for some } \delta>0.
\end{equation}
Since the domain of $H_k$ is the Sobolev space $H^2$, the bounds~\eqref{eq:norm_diff_projections} and \eqref{eq:trivialgap} imply that 
\begin{equation}
\Pi_{k,L}^\perp (H_k-E_k)\Pi_{k,L}^\perp \geq \epsilon\Pi_{k,L}^\perp  
\label{eq:resolvent_estimate}
\end{equation}
for some $\epsilon>0$. This means that 
$$(H_k-E_k)_{|(\cG_{k,L})^\perp}:=\Pi_{k,L}^\perp (H_k-E_k)\Pi_{k,L}^\perp$$ 
is invertible in the space $(\cG_{k,L})^\perp$. 
In the following, we denote by 
$$\Pi_{k,L}^\perp (H_k-E_k)^{-1}_{|(\cG_{k,L})^\perp}\Pi_{k,L}^\perp$$
the associated inverse operator in the space $(\cG_{k,L})^\perp$. 
By~\eqref{eq:norm_diff_projections} we have
\begin{align*}
\Pi_{k}^\perp(H_k-E_k)\Pi_{k}^\perp\Pi_{k,L}^\perp (H_k-E_k)^{-1}_{|(\cG_{k,L})^\perp}\Pi_{k,L}^\perp&=\Pi_{k,L}^\perp+O(e^{-cL})\\
&=\Pi_{k}^\perp+O(e^{-cL}) 
\end{align*}
which proves that 
\begin{equation}
\norm{\Pi_{k}^\perp(H_k-E_k)^{-1}_{|(\cG_{k})^\perp}\Pi_{k}^\perp-\Pi_{k,L}^\perp (H_k-E_k)^{-1}_{|(\cG_{k,L})^\perp}\Pi_{k,L}^\perp}=O(e^{-cL}).
\label{eq:estim_difference_resolvents}
\end{equation}
Hence, using the cut-off eigenspaces will always only introduce exponentially small errors. 

We recall that the two molecules are rotated using $U,V\in SO(3)$ and that the second one is translated by $Le_1$. Our variable is $\tau:=(L,U,V)\in(0,\ii)\times SO(3)\times SO(3)$. We therefore introduce the corresponding rotated and translated spaces
$$\cG_{1,\tau}:=U\cG_{1,L},\qquad \cG_{2,\tau}:=T_LV\cG_{2,L}$$
where 
\begin{equation*}
\U(\Psi)(x_1,s_1,\dots ,x_k,s_k)=\Psi(\U^{-1}x_1,s_1,\dots ,\U^{-1}x_k, s_k),
\end{equation*}
\begin{equation*}
T_L(\Psi)(x_1,s_1,\dots ,x_k,s_k)=\Psi(x_1 - L e_1, s_1, \dots ,x_k - L e_1, s_k).
\end{equation*}
Similarly, we denote by $H_{1,\tau}$ and $H_{2,\tau}$ the two Hamiltonians of the individual molecules, with the nuclei rotated and translated by $\tau$. In this manner, $\cG_{k,\tau}$ is an approximate ground state eigenspace of $H_{k,\tau}$. 

\subsubsection{The test function}
Now we are ready to prove the upper bound~\eqref{eq:upper} in Theorem~\ref{thm:expansion_energy}. We consider as in the statement of the theorem two eigenfunctions $\Psi_1\in\cG_1$ and $\Psi_2\in\cG_2$ of, respectively, $H_1$ and $H_2$. We denote by
$$\Psi_{k,L}:=(\chi_L)^{\otimes N_k}\Psi_k\in \cG_{k,L}$$
the cut-off approximate eigenfunctions and by
$$\Phi_{1,\tau}:=U\left((\chi_L)^{\otimes N_1}\Psi_1\right)\in\cG_{1,\tau},\qquad \Phi_{2,\tau}:=T_LV\left((\chi_L)^{\otimes N_2}\Psi_2\right)\in\cG_{2,\tau}$$
the rotated and translated functions. Next we pick the following test function
\begin{equation}\label{def:phitau}
\boxed{\Phi_{\tau}:=  \Phi_{1,\tau}\otimes\Phi_{2,\tau} - \chi_{\tau} \Pi_{12,\tau}^\bot R_\tau \Pi_{12,\tau}^\bot I_{\tau}  \Phi_{1,\tau}\otimes\Phi_{2,\tau}.}
\end{equation}
Here 
\begin{equation}
R_\tau=\Big(H_{1,\tau}\otimes \1+\1\otimes H_{2,\tau}-E_1-E_2\Big)^{-1}_{|(\cG_{1,\tau}\otimes\cG_{2,\tau})^\perp}
\end{equation} 
is the inverse of the operator $H_{1,\tau}\otimes \1+\1\otimes H_{2,\tau}-E_1-E_2$ restricted to the orthogonal complement of $\cG_{1,\tau}\otimes\cG_{2,\tau}$, and $\Pi_{12,\tau}$ is the projection onto $\cG_{1,\tau}\otimes\cG_{2,\tau}$ in $\bigwedge_1^{N_1}L^2\otimes \bigwedge_1^{N_2}L^2$. We have also introduced 
\begin{align}
I_\tau:=&\sum_{j=1}^{|Z_1|}\sum_{k=|Z_1|+1}^N\frac{1}{|x_j-x_k|} +\sum_{\ell=1}^{M_1}\sum_{m=M_1+1}^{M_1+M_2}\frac{z_\ell z_m}{|Uy_\ell-Vy_m-Le_1|}\nn\\
&-\sum_{j=1}^{|Z_1|}\sum_{m=M_1+1}^{M_1+M_2}\frac{z_m}{|x_j-Vy_m-Le_1|}-\sum_{k=|Z_1|+1}^N\sum_{\ell=1}^{M_1}\frac{z_\ell}{|Uy_\ell-x_k|}
\label{eq:def_I_tau}
\end{align}
which is the interaction between the two molecules, and the cut-off function
$$\chi_\tau:=(\chi_{4L/3})^{\otimes N_1}\otimes (\chi_{4L/3}(\cdot-Le_1))^{\otimes N_2}$$
(the rotations $U$ and $V$ are not necessary since $\chi_1$ is radial). Note that we use the length $4L/3$ instead of $L$. This choice is made to ensure that $\chi_\tau$ is equal to 1 on the support of $\Phi_{1,\tau}\otimes\Phi_{2,\tau}$,
\begin{equation}
 \chi_\tau\Phi_{1,\tau}\otimes\Phi_{2,\tau}=\Phi_{1,\tau}\otimes\Phi_{2,\tau},
 \label{eq:choice_cut_off}
\end{equation}
and at the same time ensure that the two localized regions stay at a distance $L/3$ on the support of $\chi_\tau$. From this we conclude that the two functions in~\eqref{def:phitau} are orthogonal to each other, 
\begin{equation}
 \pscal{\Phi_{1,\tau}\otimes\Phi_{2,\tau},\chi_{\tau} \Pi_{12,\tau}^\bot R_\tau \Pi_{12,\tau}^\bot I_{\tau}  \Phi_{1,\tau}\otimes\Phi_{2,\tau}}=0
 \label{eq:orthogonality_test_fn_1}
\end{equation}
since 
$\Pi_{12,\tau}^\bot\chi_{\tau}\Phi_{1,\tau}\otimes\Phi_{2,\tau}=\Pi_{12,\tau}^\bot\Phi_{1,\tau}\otimes\Phi_{2,\tau}=0.$

A test function defined similarly as in \eqref{def:phitau} was used for spinless fermions in the nondegenerate case in \cite{Anapolitanos-16} to provide (in the case of atoms) an upper bound for the interaction energy which is sharp to leading order. It is an approximation of the exact ground state of the system as given by the Feshbach map, if we replace the full resolvent of the system with the non-interacting resolvent. See the brief sketch of the proof of Theorem 1.5 in the introduction of \cite{Anapolitanos-16}.

The test function $\Phi_\tau$ in~\eqref{def:phitau} is not anti-symmetric. However, its antisymmetrization is a sum of functions with disjoint support having the same norm, hence anti-symmetrizing it will not change anything. More precisely, we have
$$ \frac{1}{\norm{\cQ\Phi_\tau}_{L^2}^2}\pscal{\cQ\Phi_\tau,H_N(Y(\tau),Z)\cQ\Phi_\tau}= \frac{1}{\norm{\Phi_\tau}_{L^2}^2} \pscal{\Phi_\tau,H_N(Y(\tau),Z)\Phi_\tau},$$
where $\cQ$ is the projector onto the fermionic subspace,
\begin{equation}\label{def:Q}
 \cQ=\frac{1}{N!}\sum_{\pi \in \gS_{N}}(-1)^\pi \pi,
\end{equation}
with the action of the permutation $\pi$ on wavefunctions as in~\eqref{eq:fermions}. The Hamiltonian $H_N(Y(\tau),Z)$ is here extended to $\bigwedge_1^{N_1}L^2\otimes \bigwedge_1^{N_2}L^2$ in the obvious manner.
From this we get an upper bound on the fermionic energy, as follows:
\begin{equation}\label{testfunction}
\cE(\tau)-E_{1}-E_2 \leq  \frac{1}{\| \Phi_\tau\|^2} \pscal{ \Phi_\tau, \Big(H_N(Y(\tau),Z)-E_1-E_2\Big)\Phi_\tau}.
\end{equation}
It therefore remains to compute the right side of~\eqref{testfunction}, which can be done similarly as in the proof of \cite[Theorem 1.5]{Anapolitanos-16}. A main difference is the presence of multipole moments, due to which the term $ \langle \Phi_{\tau}, I_\tau \Phi_{\tau} \rangle $ is no longer exponentially decaying, but gives the multipole-multipole expansion as stated in Lemma \ref{lem:multipolar_expansion}.

First we compute the $L^2$ norm of $\Phi_\tau$ using the orthogonality~\eqref{eq:orthogonality_test_fn_1} as follows:
\begin{equation}
\norm{\Phi_\tau}^2= \norm{\Phi_{1,\tau}}^2\norm{\Phi_{2,\tau}}^2 +\norm{\chi_{\tau} \Pi_{12,\tau}^\bot R_\tau \Pi_{12,\tau}^\bot I_{\tau}  \Phi_{1,\tau}\otimes\Phi_{2,\tau}}^2.
\label{eq:comput_L2_norm}
\end{equation}
The exponential decay implies that 
$$\norm{\Phi_{1,\tau}}^2=1+O(e^{-cL}),\qquad \norm{\Phi_{2,\tau}}^2=1+O(e^{-cL}).$$
Even if we will not use it later, let us explain in detail how to deal with the second term in~\eqref{eq:comput_L2_norm}, which turns out to be of order $L^{-6}$. We change variables and obtain
\begin{multline}
\norm{\chi_{\tau} \Pi_{12,\tau}^\bot R_\tau \Pi_{12,\tau}^\bot I_{\tau}  \Phi_{1,\tau}\otimes\Phi_{2,\tau}}^2\\
=\norm{(\chi_{4L/3})^{\otimes N} \Pi_{12,L}^\bot \big(H_1+H_2-E_1-E_2\big)^{-1}_{|(\cG_{1,L}\otimes\cG_{2,L})^\perp}\Pi_{12,L}^\bot \widetilde{I_{\tau}}  \Psi_{1,L}\otimes\Psi_{2,L}}^2
\label{eq:comput_L2_norm_2nd_term}
\end{multline}
where now $\Pi_{12,L}=\Pi_{1,L}\otimes\Pi_{2,L}$ is the projection onto $\cG_{1,L}\otimes \cG_{2,L}$ and
\begin{align}
&\widetilde{I_{\tau}}=\sum_{j=1}^{|Z_1|}\sum_{k=|Z_1|+1}^N\frac{1}{|Ux_j-Vx_k-Le_1|} +\sum_{\ell=1}^{M_1}\sum_{m=M_1+1}^{M_1+M_2}\frac{z_\ell z_m}{|Uy_\ell-Vy_m-Le_1|}\nn\\
&-\sum_{j=1}^{|Z_1|}\sum_{m=M_1+1}^{M_1+M_2}\frac{z_m}{|Ux_j-Vy_m-Le_1|}-\sum_{k=|Z_1|+1}^N\sum_{\ell=1}^{M_1}\frac{z_\ell}{|Uy_\ell-Vx_k-Le_1|}.
\label{eq:def_tilde_I_tau}
\end{align}
Only $\widetilde{I_\tau}$ contains the rotations $U,V\in SO(3)$. On the support of $\Psi_{1,L}\otimes\Psi_{2,L}$, $\widetilde{I_{\tau}}$ is indeed a smooth function of $U,V\in SO(3)^2$ and all our next arguments apply in the same way to derivatives with respect to $U,V$ of the operator estimated in~\eqref{eq:comput_L2_norm_2nd_term}. 

Since by construction $(\chi_{4L/3})^{\otimes N}=1$ on the supports of the functions in $\Ran \Pi_{12,L}$, we have
$$\Pi_{12,L}^\bot(\chi_{4L/3})^{\otimes N} \Pi_{12,L}^\bot=(\chi_{4L/3})^{\otimes N}\Pi_{12,L}^\bot=\Pi_{12,L}^\bot(\chi_{4L/3})^{\otimes N}$$
hence
\begin{multline*}
\Pi_{12,L}^\bot(H_1+H_2-E_1-E_2)\Pi_{12,L}^\bot(\chi_{4L/3})^{\otimes N} \Pi_{12,L}^\bot\\
=(\chi_{4L/3})^{\otimes N}\Pi_{12,L}^\bot(H_1+H_2-E_1-E_2)\Pi_{12,L}^\bot -\Pi_{12,L}^\bot \big[\Delta,(\chi_{4L/3})^{\otimes N}\big] \Pi_{12,L}^\bot.
\end{multline*}
From this we conclude that 
\begin{multline}
(\chi_{4L/3})^{\otimes N} R_{12,L} \widetilde{I_{\tau}}  \Psi_{1,L}\otimes\Psi_{2,L}\\
=R_{12,L} \widetilde{I_{\tau}}  \Psi_{1,L}\otimes\Psi_{2,L}-R_{12,L}\big[\Delta,(\chi_{4L/3})^{\otimes N}\big] R_{12,L}\widetilde{I_{\tau}}  \Psi_{1,L}\otimes\Psi_{2,L}
\label{eq:simplification_cut_off}
\end{multline}
and that 
\begin{multline}
\Pi_{12,L}^\bot (H_1+H_2-E_1-E_2)(\chi_{4L/3})^{\otimes N} R_{12,L} \widetilde{I_{\tau}}  \Psi_{1,L}\otimes\Psi_{2,L}\\
=\Pi_{12,L}^\perp \widetilde{I_{\tau}}  \Psi_{1,L}\otimes\Psi_{2,L}-\Pi_{12,L}^\perp\big[\Delta,(\chi_{4L/3})^{\otimes N}\big] R_{12,L}\widetilde{I_{\tau}}  \Psi_{1,L}\otimes\Psi_{2,L},
\label{eq:simplification_cut_off2}
\end{multline}
where we have used the simplified notation
\begin{equation}\label{eq:R12L}
R_{12,L}:=\Pi_{12,L}^\bot \big(H_1+H_2-E_1-E_2\big)^{-1}_{|(\cG_{1,L}\otimes\cG_{2,L})^\perp}\Pi_{12,L}^\bot.
\end{equation}
From Lemma~\ref{lem:multipexp} (see also \cite{AnaSig-17} and \cite{Anapolitanos-16}) we have
$$\widetilde{I_{\tau}}=\frac{f_{(U,V)}}{L^3}+O_{C^2(SO(3)^2)}\left(\frac{1}{L^4}\right)$$ 
uniformly on the support of $\Psi_{1,L}\otimes\Psi_{2,L}$, where we recall that $f_{(U,V)}$ is defined in~\eqref{eq:def_f_U_V}. Using the exponential decay of $\Psi_1$ and $\Psi_2$, i.e. \eqref{eq:expdecay}, we obtain that
\begin{equation}\label{eq:Ifuv}
\norm{\widetilde{I_{\tau}}\Psi_{1,L}\otimes\Psi_{2,L}-\frac{f_{(U,V)}}{L^3}\Psi_{1}\otimes\Psi_{2} }_{C^2(SO(3)^2)}=O\left(\frac{1}{L^4}\right).
\end{equation}
From~\eqref{eq:estim_difference_resolvents} we have 
\begin{equation}
\Big\| R_{12,L}- \Pi_{12}^\bot \big(H_1+H_2-E_1-E_2\big)^{-1}_{|(\cG_{1}\otimes\cG_{2})^\perp}\Pi_{12}^\bot\Big\|_{C^2(SO(3)^2)}=O(e^{-cL})
\label{eq:estim_resolvents}
\end{equation}
with $\Pi_{12}=\Pi_1\otimes\Pi_2$ the orthogonal projection onto $\cG_1\otimes\cG_2$. Using $$[\Delta,f]=(\Delta f)+2(\nabla f)\cdot\nabla$$ for the commutator involving $(\chi_{4L/3})^{\otimes N}$
together with 
\eqref{eq:simplification_cut_off}, \eqref{eq:Ifuv}, \eqref{eq:estim_resolvents} and  \eqref{chiader} we find that
\begin{multline}
\bigg\|(\chi_{4L/3})^{\otimes N} R_{12,L} \widetilde{I_{\tau}}  \Psi_{1,L}\otimes\Psi_{2,L}\\
-\frac1{L^{3}}\Pi_{12}^\bot \big(H_1+H_2-E_1-E_2\big)^{-1}_{|(\cG_{1}\otimes\cG_{2})^\perp}\Pi_{12}^\bot f_{(U,V)}  \Psi_{1}\otimes\Psi_{2}\bigg\|_{C^2(SO(3))}\!\!=O\left(\frac{1}{L^4}\right).
\end{multline}
This allows us to compute the norm of $\Phi_\tau$ with the help of \eqref{def:phitau} and \eqref{eq:orthogonality_test_fn_1} to leading order:
\begin{multline}
\norm{\Phi_\tau}^2=1
+\frac1{L^6}\norm{\Pi_{12}^\bot \big(H_1+H_2-E_1-E_2\big)^{-1}_{|(\cG_{1}\otimes\cG_{2})^\perp}\Pi_{12}^\bot f_{(U,V)}  \Psi_{1}\otimes\Psi_{2}}^2\\
+O_{C^2(SO(3))}\left(\frac{1}{L^7}\right).
\label{eq:final_estim_norm}
\end{multline}
It turns out that the precise form of the term of order $L^{-6}$ does not matter for our proof.  However, with very similar arguments, we find for the energy that 
\begin{multline}
\pscal{ \Phi_\tau, \Big(H_N(Y(\tau),Z)-E_1-E_2\Big)\Phi_\tau}=\pscal{\Psi_{1,L}\otimes\Psi_{2,L},\tilde{I_\tau} \Psi_{1,L}\otimes\Psi_{2,L}}\\
-\frac{C_{\rm vdW}(\Psi_1,\Psi_2,U,V)}{L^6}
+O_{C^2(SO(3))}\left(\frac{1}{L^7}\right),
\label{eq:energy_expansion}
\end{multline}
where $C_{\rm vdW}(\Psi_1,\Psi_2,U,V)$ is the van der Waals corellation term defined in \eqref{eq:def_C_vdW_Psi_1_2}. A corresponding detailed calculation in the case of two atoms was done in \cite{AnaLewRot-19}. We further observe that from \eqref{eq:def_tilde_I_tau} and \eqref{eq:def_rho_1_Psi}-\eqref{eq:def_rho_2_Psi} for $U=V=I$ we have that
\begin{equation}\label{eq:psiIpsirho}
\pscal{\Psi_{1,L}\otimes\Psi_{2,L},\tilde{I_\tau} \Psi_{1,L}\otimes\Psi_{2,L}} = \int \frac{\rho_{\Psi_{1,L}}(x) \rho_{\Psi_{2,L}}(y)}{|Ux-Vy-Le_1|} dx dy.
\end{equation}
 Using now \eqref{eq:psiIpsirho}, Lemma~\ref{lem:multipolar_expansion} and the fact that 
$$\int|\Psi_{k,L}|^2=1+O_{C^2(SO(3))}(e^{-cL})$$
we can obtain
\begin{multline}\label{eq:psiIpsiexp}
\pscal{\Psi_{1,L}\otimes\Psi_{2,L},\tilde{I_\tau} \Psi_{1,L}\otimes\Psi_{2,L}}=\sum_{2\leq n+m\leq 5}\frac{\cF^{(n,m)}(\Psi_1,\Psi_2,U,V)}{L^{n+m+1}}\\+O_{C^2(SO(3))}\left(\frac{1}{L^7}\right).
\end{multline}
Thus, since the left hand side of \eqref{eq:psiIpsiexp}  is at most of the order of $O_{C^2(SO(3))}(L^{-3})$, dividing \eqref{eq:energy_expansion} by the norm $\|\Phi_\tau\|^2=1+O_{C^2(SO(3))}(L^{-6})$ will only change the remainder of our expansion. This concludes the proof of the upper bound~\eqref{eq:upper}. 

We remark that from \eqref{eq:psiIpsiexp} and \eqref{eq:energy_expansion} we obtain that
\begin{multline}\label{eq:testsmoothness}
\bigg\| \frac{1}{\| \Phi_\tau\|^2} \pscal{ \Phi_\tau, \Big(H_N(Y(\tau),Z)-E_1-E_2\Big)\Phi_\tau}
-\!\!\!\sum_{2\leq n+m\leq 5}\frac{\cF^{(n,m)}(\Psi_1,\Psi_2,U,V)}{L^{n+m+1}}
\\ +\frac{C_{\rm vdW}(\Psi_1,\Psi_2,U,V)}{L^6}\bigg\|_{C^2(SO(3)^2)}\!\!=O\left(\frac{1}{L^7}\right),
\end{multline}
which is going to be very useful in Section \ref{sec:Proofevergyexpsmooth}.
Note that $C_{\rm vdW}(\Psi_1,\Psi_2,U,V)$ is smooth in $U,V$ as it is a polynomial in $U,V$. 

\subsubsection{Proof that $C_{\rm vdW}(\Psi_1,\Psi_2,U,V)>0$}
It is clear from the definition \eqref{eq:def_C_vdW_Psi_1_2} that $C_{\rm vdW}(\Psi_1,\Psi_2,U,V)\geq0$. Furthermore, $C_{\rm vdW}(\Psi_1,\Psi_2,U,V)=0$ if and only if $f_{(U,V)}\Psi_1\otimes\Psi_2$ is a ground state of $H_1\otimes \1+\1\otimes H_2$. This does not happen by the following

\begin{proposition}[Positivity of the van der Waals correlation function]\label{prop:positivity_vdW}
Suppose that $\Psi_1, \Psi_2$ are ground states of $H_1$ and $H_2$, respectively. Then $f_{(U,V)} \Psi_1 \otimes \Psi_2$ cannot be a ground state of $H_1\otimes \1+\1\otimes H_2$. In other words, $\Pi_{12}^\perp f_{(U,V)} \Psi_1 \otimes \Psi_2\neq0$ hence $C_{\rm vdW}(\Psi_1,\Psi_2,U,V)>0$.
\end{proposition}

We first state and prove the following lemma
\begin{lemma}\label{lem:affine}
Let $\Psi$ be a ground state of $H_k$ and $L$ a nonconstant affine function of $x_1,\dots ,x_N$. Then $L \Psi$ cannot be a ground state of $H_k$.
\end{lemma}
\begin{proof}[Proof of Lemma~\ref{lem:affine}]
If $(H_j-E_j) L\Psi=0$, then since $\Delta L=0$ and $\nabla L$ is a constant vector $\vec{b}$, from the Leibniz rule we would find that
$\vec{b} \cdot \nabla \Psi=0$. This would imply a translation invariance in the direction of the vector $\vec{b}$ which contradicts the exponential decay of $\Psi$ (in fact it already contradicts the fact that $\Psi \in L^2$).
\end{proof}

	We are now able to write the 
	
	\begin{proof}[Proof of Proposition~\ref{prop:positivity_vdW}] The function $f_{U,V} \Psi_1 \otimes \Psi_2$ can be written in the form
$$f_{U,V} \Psi_1 \otimes \Psi_2=\sum_{n,m\geq0} a_{nm} \Psi_{1,n} \otimes \Psi_{2,m},$$ 
for  some coefficients $a_{nm}$. Here for $n>0$
	$\Psi_{k,n}$ has the form $x_i^{(j)} \Psi_k$ where $x_i^{(j)}$ denotes the $j$-th component of $x_i$. For $n=0$ we have  $\Psi_{k,0}:=\Psi_k$. Assuming that $f_{U,V} \Psi_1 \otimes \Psi_2$ is a ground state of $H_1+H_2$ we find that
	\begin{equation}
	\sum_{n,m\geq1} a_{nm} \Psi_{1,n}^\bot \otimes \Psi_{2,m}^\bot =0,
	\end{equation}
	where $\Psi_{k,n}^\bot=\Pi_k^\perp\Psi_{k,n}$ and recall that $\Pi_k$ is the orthogonal projection onto the ground state eigenspace $\cG_k$. Note that we have also used that $\Psi_{1,0}^\bot=\Psi_{2,0}^\bot=0$.
	 Since the $(\Psi_{1,n}^\bot \otimes \Psi_{2,m}^\bot)_{n,m\geq1}$ are linearly dependent, it follows that one of the two sets $(\Psi_{1,n}^\bot)_{n\geq1}$ or $(\Psi_{2,m}^\bot)_{m\geq1}$ has to be linearly dependent as well. This is because $\dim(E_1\otimes E_2)=\dim(E_1)\times \dim(E_2)$. But if for instance $\sum_{n\geq1} b_n \Psi_{1,n}^\bot=0$, then $\sum_{n\geq1} b_n \Psi_{1,n}$ is a ground state of $H_1$ and we obtain a contradiction by Lemma \ref{lem:affine}. 
\end{proof}

This concludes the proof of the upper bound~\eqref{eq:upper} and of the properties of the van der Waals correlation energy.

\subsection{Proof of Theorem \ref{thm:expansion_energy}: lower bound in the irreducible case}

In this section we provide the proof of the reverse inequality in~\eqref{eq:upper} in the irreducible case, under the assumption~\eqref{eq:main_assumption2}.

\subsubsection{The Feshbach-Schur method}
In the irreducible case we are going to follow the strategy employed in \cite{AnaSig-17}, which is based on the Feshbach-Schur method. Let us recall that if we have a self-adjoint operator $A$ with lowest eigenvalue $E$ and a projection $P$ such that $P^\perp (A-E)P^\perp\geq \eps P^\perp$ for some $\eps>0$, then we can rewrite the eigenvalue equation
$$(A-E)f=0$$
using the Schur complement formula in the form
\begin{equation}
\Big(A^{++}-A^{+-}(A^{--}-E)^{-1}A^{-+}-E\Big)f_+=0
\label{eq:F-S-eigenvalue_eq}
\end{equation}
with $A^{++}=PAP$, $A^{--}=P^\perp AP^\perp$, $A^{-+}=P^\perp AP=(A^{+-})^*$ and $f_+=P f$. The other component is given by 
$$f_-=P^\perp f=(A^{--}-E)^{-1}A^{-+}f_+.$$
More precisely, $E$ is the smallest real number for which 0 is an eigenvalue of the operator
$$F(e):=A^{++}-A^{+-}(A^{--}-e)^{-1}A^{-+}-e$$
for $e< E+\eps$. Note that $F(e)$ is monotone-decreasing with respect to $e$, hence all its eigenvalues are decreasing as well. 
Taking the scalar product with $f_+$ in~\eqref{eq:F-S-eigenvalue_eq} we find 
\begin{align}
E&=\frac{\pscal{f_+,A^{++}f_+}-\pscal{f_+,A^{+-}(A^{--}-E)^{-1}A^{-+}f_+}}{\|f_+\|^2}\nn\\
&\geq \frac{\pscal{f_+,A^{++}f_+}-\eps^{-1}\norm{A^{-+}f_+}^2}{\|f_+\|^2}.\label{eq:lower_abstract_F-S}
\end{align}
In other words,
\begin{equation}
E\geq \min\sigma\Big(A^{++}-\eps^{-1}A^{+-}A^{-+}\Big) \geq \min\sigma\big(A^{++}\big)-\eps^{-1}\|A^{-+}\|^2,
\label{eq:lower_abstract_F-S_spectrum}
\end{equation}
where the operator $A^{++}$ is considered on ${\rm Ran}(P)$ only. This equation will give us the sought-after lower bound, if we choose $P$ appropriately. 

Coming back to our problem, we recall that the spectrum of $H_N(Y(\tau),Z)$ is known to converge in the limit $L\to\ii$ to the union of the spectra of the operators 
$$H_{n}(Y_1,Z_1)\otimes \1+\1\otimes H_{N-n}(Y_2,Z_2),$$ 
with $n=0,...,N$, see, e.g.,~\cite{MorSim-80,AnaSig-17} and~\cite[Thm. 2 \& 3]{Lewin-04b}. Our assumption~\eqref{eq:main_assumption2} that 
$$E_{|Z_1|}(Y_1,Z_1)+E_{|Z_2|}(Y_2,Z_2)<\min_{n\neq |Z_1|}\Big(E_{n}(Y_1,Z_1)+E_{N-n}(Y_2,Z_2)\Big)$$
implies that the ground state energy $\cE(\tau)$ of $H_N(Y(\tau),Z)$ is obtained in the limit $L \to \infty$ in the neutral case $n=|Z_1|$ and that there is a gap above it. Then, using the IMS localization method, it has been proved in~\cite{MorSim-80,AnaSig-17} that the ground state eigenspace of $H_N(Y(\tau),Z)$ is close to 
$$\cQ(U\cdot \cG_{1}\otimes T_LV\cdot \cG_{2})$$ 
where we use here the notations of the previous section. We recall that the eigenspaces of $H_k$ are denoted by $\cG_k$, whereas $\cQ$ is the projection onto antisymmetric functions. The precise statement is that there exists $\eps>0$ (independent of $\tau$ for $L$ large enough) such that 
\begin{equation}
 P^\perp\Big(H_N(Y(\tau),Z)-\cE(\tau)\Big)P^\perp\geq \eps P^\perp
 \label{eq:F-S_original_space}
\end{equation}
where $P$ is the projection onto $\cQ(U\cdot \cG_{1}\otimes T_LV\cdot \cG_{2})$, so that the Feshbach-Schur method can be applied on this space. We can also replace the projection $P$ by the projection $\Pi_{12,L}$ onto $\cQ(U\cdot \cG_{1,L}\otimes T_LV\cdot \cG_{2,L})$, which is known to be exponentially close to $P$. This is the choice made in~\cite{AnaSig-17}. 

In our proof we are going to use a different projection, in order to catch the leading orders completely in the range of $P$. We look at the space spanned by our test functions~\eqref{def:phitau} when $\Psi_1$ and $\Psi_2$ span $\cG_1$ and $\cG_2$:
\begin{multline}
\Upsilon_\tau:={\rm span}\bigg\{\cQ\left(\Phi_{1,\tau}\otimes\Phi_{2,\tau} - \chi_{\tau} \Pi_{12,\tau}^\bot R_\tau \Pi_{12,\tau}^\bot I_{\tau}  \Phi_{1,\tau}\otimes\Phi_{2,\tau}\right)\ :\\ \Phi_{1,\tau}=U\left((\chi_L)^{\otimes N_1}\Psi_1\right)\in\cG_{1,\tau},\ \Phi_{2,\tau}=T_LV\left((\chi_L)^{\otimes N_2}\Psi_2\right)\in\cG_{2,\tau}\bigg\}.
\label{eq:space_used_for_F-S}
\end{multline}
As we have proved in the previous section, 
\begin{multline*}
\norm{U\cdot \Psi_{1}\otimes V T_L\Psi_{2}-\Phi_{1,\tau}\otimes\Phi_{2,\tau} + \chi_{\tau} \Pi_{12,\tau}^\bot R_\tau \Pi_{12,\tau}^\bot I_{\tau}  \Phi_{1,\tau}\otimes\Phi_{2,\tau}}_{H^2}\\
\leq C\frac{\|\Psi_1\|_{H^1}\|\Psi_2\|_{H^1}}{L^3} 
\end{multline*}
and this shows that our space $\Upsilon_\tau$ in~\eqref{eq:space_used_for_F-S} is close to $\cQ(U\cdot \cG_{1}\otimes T_LV\cdot \cG_{2})$ and to $\cQ(U\cdot \cG_{1,L}\otimes T_LV\cdot \cG_{2,L})$, in the Sobolev norm $H^2$. More precisely, denoting by $\Pi_\tau$ the orthogonal projection onto $\Upsilon_\tau$, we have
\begin{equation}
\norm{(1-\Delta)\left(\Pi_\tau-\Pi_{12,\tau}\right)(1-\Delta)}=O\left(\frac1{L^3}\right).
\label{eq:diff_projection_F-S}
\end{equation}
From~\eqref{eq:F-S_original_space} and~\eqref{eq:norm_diff_projections} this implies again
\begin{equation}
\Pi_\tau^\perp\Big(H_N(Y(\tau),Z)-\cE(\tau)\Big)\Pi_\tau^\perp\geq \eps \Pi_\tau^\perp
\label{eq:F-S_new_space}
\end{equation}
for some $\eps>0$. Hence we may as well use the projection $\Pi_\tau$ for the Schur complement method. 

\subsubsection{Estimate on the off-diagonal term $A^{-+}$}\label{sec:setimate_off_diagonal_term}
Here we prove that, in operator norm,
\begin{align}
\|A^{-+}\|&=\norm{\Pi_\tau^\perp H_N(Y(\tau),Z)\Pi_\tau}\\
&=\norm{\Pi_\tau^\perp \Big(H_N(Y(\tau),Z)-E_1-E_2\Big)\Pi_\tau}=O\left(\frac{1}{L^4}\right)
\label{eq:estim_off_diagonal_terms}
\end{align}
which will allow us to neglect the term $\|A^{-+}\|^2=O(L^{-8})$ in~\eqref{eq:lower_abstract_F-S_spectrum}. To this end we will prove that
\begin{multline}\label{eq:H1estusefull'}
\Big\| \cQ \Pi_\tau^\perp \big(H_N(Y(\tau),Z)-E_1-E_2\big)\times\\
\times\left(\Phi_{1,\tau}\otimes\Phi_{2,\tau} - \chi_{\tau} \Pi_{12,\tau}^\bot R_\tau \Pi_{12,\tau}^\bot I_{\tau}  \Phi_{1,\tau}\otimes\Phi_{2,\tau}\right)\Big\|_{H^1}=O\left(\frac{1}{L^4}\right).
\end{multline}
In fact, estimating the $L^2$ norm would already be enough, however the estimate in the $H^1$ norm is going to be useful later for the  proof of Proposition \ref{prop:fullexp}.
Due to the locality of the Hamiltonian $H_N(Y(\tau),Z)$ and using that the antisymmetrization $\cQ$ on any function of the form $\Phi_{1,\tau}\otimes\Phi_{2,\tau} - \chi_{\tau} \Pi_{12,\tau}^\bot R_\tau \Pi_{12,\tau}^\bot I_{\tau}  \Phi_{1,\tau}\otimes\Phi_{2,\tau}$ produces a sum of functions that have disjoint support from each other, we find that
 \eqref{eq:H1estusefull'} reduces to showing that
 \begin{multline}\label{eq:H1estusefull}
 \Big\|  \Pi_{0,\tau}^\perp \big(H_N(Y(\tau),Z)-E_1-E_2\big)\times\\
 \times \left(\Phi_{1,\tau}\otimes\Phi_{2,\tau} - \chi_{\tau} \Pi_{12,\tau}^\bot R_\tau \Pi_{12,\tau}^\bot I_{\tau}  \Phi_{1,\tau}\otimes\Phi_{2,\tau}\right)\Big\|_{H^1}=O\left(\frac{1}{L^4}\right),
 \end{multline}
 where $\Pi_{0,\tau}$ is the orthogonal projection onto the function $\Phi_{1,\tau}\otimes\Phi_{2,\tau} - \chi_{\tau} \Pi_{12,\tau}^\bot R_\tau \Pi_{12,\tau}^\bot I_{\tau}  \Phi_{1,\tau}\otimes\Phi_{2,\tau}$.
A simple change of variables shows, similarly as in \eqref{eq:comput_L2_norm_2nd_term},  that
\begin{multline}\label{eq:H1change}
\Big\| \Pi_{0,\tau}^\perp \big(H_N(Y(\tau),Z)-E_1-E_2\big)\times\\
\times\left(\Phi_{1,\tau}\otimes\Phi_{2,\tau} - \chi_{\tau} \Pi_{12,\tau}^\bot R_\tau \Pi_{12,\tau}^\bot I_{\tau}  \Phi_{1,\tau}\otimes\Phi_{2,\tau}\right)\Big\|_{H^1} \\=
	\Big\| \Pi_L^\perp \big(H_1 + H_2 +  \widetilde{I_\tau}-E_1-E_2\big)\times\\ \times\left(\Phi_{1,L}\otimes\Phi_{2,L} - \chi_{4L/3}^{\otimes N} R_{12,L} \widetilde{I_{\tau}}  \Phi_{1,L}\otimes\Phi_{2,L}\right)\Big\|_{H^1},
\end{multline}
where recall that $R_{12,L}$ was defined in \eqref{eq:R12L} and $\Pi_L$ is the orthogonal projection onto
$$\text{span}\{\Phi_{1,L}\otimes\Phi_{2,L} - \chi_{4L/3}^{\otimes N} R_{12,L} \widetilde{I_{\tau}}  \Phi_{1,L}\otimes\Phi_{2,L}, \Phi_{j,L} \in \cG_{j,L}, j=1,2\}.$$
Using~\eqref{eq:expdecay} we have $\|(H_1+H_2-E_1-E_2)\Phi_{1,L}\otimes\Phi_{2,L}\|_{H^1}=O(e^{-cL})$ for some $c>0$, and we find
\begin{align*}
& \Pi_L^\perp \big(H_1 + H_2 +  \widetilde{I_\tau}-E_1-E_2\big)\left(\Phi_{1,L}\otimes\Phi_{2,L} - \chi_{4L/3}^{\otimes N} R_{12,L} \widetilde{I_{\tau}} \Phi_{1,L}\otimes\Phi_{2,L}\right)\\
&=\Pi_L^\perp \widetilde{I_\tau} \Phi_{1,L}\otimes\Phi_{2,L} - \Pi_L^\perp \widetilde{I_\tau} \chi_{4L/3}^{\otimes N}  R_{12,L}  \widetilde{I_{\tau}}  \Phi_{1,L}\otimes\Phi_{2,L}\\
&\qquad -\Pi_L^\perp(H_{1}+H_{2}-E_1-E_2)\chi_{4L/3}^{\otimes N}  R_{12,L} \widetilde{I_\tau}  \Phi_{1,L}\otimes\Phi_{2,L} + O_{H^1}(e^{-cL}).
\end{align*}
Therefore, using
$$(H_{1}+H_{2}-E_1-E_2) \chi_{4L/3}^{\otimes N}= \chi_{4L/3}^{\otimes N} (H_{1}+H_{2}-E_1-E_2) -[ \Delta, \chi_{4L/3}^{\otimes N}],$$
it follows that
\begin{align} \nonumber
&\Pi_L^\perp \big(H_1 + H_2 +  \widetilde{I_\tau}-E_1-E_2\big)\left(\Phi_{1,L}\otimes\Phi_{2,L} - \chi_{4L/3}^{\otimes N} R_{12,L} \widetilde{I_{\tau}} \Phi_{1,L}\otimes\Phi_{2,L}\right)\\ \nonumber 
&= \Pi_L^\perp\Pi_{12,L} \widetilde{I_\tau} \Phi_{1,L}\otimes\Phi_{2,L}
- \Pi_L^\perp \widetilde{I_\tau} \chi_{4L/3}^{\otimes N}  R_{12,L}  \widetilde{I_\tau}  \Phi_{1,L}\otimes\Phi_{2,L}\\ \label{eq:H1change1}
&\qquad +\Pi_L^\perp[\Delta,\chi_{4L/3}^{\otimes N}] R_{12,L} \widetilde{I_\tau} \Phi_{1,L}\otimes\Phi_{2,L} + O_{H^1}(e^{-cL}),
\end{align}
where we have also used that $$(H_1 + H_2- E_1- E_2)R_{12,L}  \widetilde{I_\tau} \Phi_{1,L}\otimes\Phi_{2,L} = \Pi_{12,L}^\bot \widetilde{I_\tau} \Phi_{1,L}\otimes\Phi_{2,L}  + O_{H^1}(e^{-cL}).$$
Since $\chi_{4L/3}^{\otimes N}$ is supported $O(L)$ away from the singularities of $\widetilde{I_\tau}$, we see that 
$$\|\widetilde{I_\tau} \chi_{4L/3}^{\otimes N}\|_{H^1} =O\left(\frac1{L}\right).$$
Hence we obtain that
\begin{equation}\label{eq:H1change2}
\norm{\Pi_L^\perp \widetilde{I_\tau} \chi_{4L/3}^{\otimes N}  R_{12,L}  \widetilde{I_\tau}  \Phi_{1,L}\otimes\Phi_{2,L}}_{H^1}=O\left(\frac{\|\Psi_1\|\,\|\Psi_2\|}{L^4}\right),
\end{equation}
because, as seen in the previous section,
$$\norm{\widetilde{I_\tau}  \Phi_{1,L}\otimes\Phi_{2,L}}_{L^2}=O\left(\frac{\|\Psi_1\|\|\Psi_2\|}{L^3}\right).$$
Therefore, from the equality $[\Delta,\chi_{4L/3}^{\otimes N}]=(\Delta \chi_{4L/3}^{\otimes N})-2(\nabla \chi_{4L/3}^{\otimes N})\cdot\nabla$ and \eqref{chiader}, we can similarly show that
\begin{equation}\label{eq:H1change3}
\norm{\Pi_L^\perp[\Delta,\chi_{4L/3}^{\otimes N}] R_{12,L}  \widetilde{I_\tau} \Phi_{1,L}\otimes\Phi_{2,L}}_{H^1}=O\left(\frac{\|\Psi_1\|\,\|\Psi_2\|}{L^4}\right).
\end{equation}
Finally, using that $\Pi_L^\perp\Pi_{12,L}=\Pi_L^\perp(\Pi_{12,L}-\Pi_L)$ together with \eqref{eq:diff_projection_F-S},
we obtain that
\begin{equation}\label{eq:H1change4}
\norm{\Pi_L^\perp\Pi_{12,L} \widetilde{I_\tau} \Phi_{1,L}\otimes\Phi_{2,L}}_{H^1}=O\left(\frac{\|\Psi_1\|\,\|\Psi_2\|}{L^6}\right).
\end{equation}
From  \eqref{eq:H1change}, \eqref{eq:H1change1}, \eqref{eq:H1change2}, \eqref{eq:H1change3} and \eqref{eq:H1change4}
we obtain \eqref{eq:H1estusefull} and therefore \eqref{eq:H1estusefull'}.
Since
$$\norm{\Phi_{1,L}\otimes\Phi_{2,L} - \chi_{4L/3}^{\otimes N}  R_{12,L}  \widetilde{I_\tau} \Phi_{1,L}\otimes\Phi_{2,L}}=\|\Psi_1\|\|\Psi_2\|\left(1+O(L^{-3})\right),$$
 the bound \eqref{eq:H1estusefull'}  implies  ~\eqref{eq:estim_off_diagonal_terms}. Thus, the last term in the Schur complement estimate~\eqref{eq:lower_abstract_F-S_spectrum} does not contribute to the order we are interested in. 

\begin{remark}\label{rem:powerx}
The estimates \eqref{eq:H1change2}, \eqref{eq:H1change3} and \eqref{eq:H1change4} and therefore as well \eqref{eq:H1estusefull} remain true in the $L^2$ norm if we multiply with some power of $|(x_1,\dots,x_N)|$ as well. This is because due to \eqref{eq:expdecay} we can write for example
$$|x|^n\widetilde{I_\tau} \Phi_{1,L}\otimes\Phi_{2,L}= |x|^n e^{-\frac{c|x|}{2}} \widetilde{I_\tau} e^{\frac{c|x|}{2}}\Phi_{1,L}\otimes\Phi_{2,L}.$$ 
Then $|x|^n e^{-\frac{c|x|}{2}}$ is bounded and 
$$\widetilde{I_\tau} e^{\frac{c|x|}{2}}\Phi_{1,L}\otimes\Phi_{2,L}=O\left(\frac{1}{L^3}\right)$$ 
still because the function $e^{\frac{c|x|}{2}}\Phi_{1,L}\otimes\Phi_{2,L}$ is still exponentially decaying.  Even in terms where the resolvent is present there is no problem as the exponential $e^{-\frac{c|x|}{2}}$  can be pushed through with the help of boosted Hamiltonians: we can write, for example 
$$R_{12,L}  \widetilde{I_\tau} \Phi_{1,L}\otimes\Phi_{2,L}= e^{-\frac{c|x|}{2}} R_{12,L,c}  e^{\frac{c|x|}{2}}\widetilde{I_\tau} \Phi_{1,L}\otimes\Phi_{2,L},$$ 
where $R_{12,L,c}= e^{\frac{c|x|}{2}} R_{12,L} e^{\frac{-c|x|}{2}}$. Standard arguments (see for example \cite[Sec.~5.4]{Anapolitanos-16}) give that $R_{12,L,c}$ is bounded.  This remark is going to be useful for the proof of Proposition \ref{prop:fullexp}.
\end{remark}

\subsubsection{Irreducibility and computation of $A^{++}$}
Now we prove that the operator localized to the space $\Upsilon_\tau$
$$\Pi_\tau H_N(Y(\tau),Z)\Pi_\tau$$
is actually a multiple of the identity on $\Upsilon_\tau$, that is,
\begin{equation}
 \Pi_\tau H_N(Y(\tau),Z)\Pi_\tau=\frac{\pscal{\Phi_\tau,H_N(Y(\tau),Z)\Phi_\tau}}{\|\Phi_\tau\|^2}\Pi_\tau
 \label{eq:multiple_identity}
\end{equation}
where $\Phi_\tau$ is any chosen non-zero vector of $\Upsilon_\tau$. This is of course obvious when $\dim\cG_1=\dim\cG_2=1$ since then $\dim \Upsilon_\tau=1$ as well. In the general case this follows from the irreducibility condition.
Indeed, for any $\Psi_1\in\cG_1$ and $\Psi_2\in\cG_2$, we have for the corresponding $\cQ\Phi_\tau\in\Upsilon_\tau$
\begin{align*}
\frac{\pscal{\cQ\Phi_\tau,H_N(Y(\tau),Z)\cQ\Phi_\tau}}{\norm{\cQ\Phi_\tau}^2}&=\frac{\pscal{\Phi_\tau,(H_{1,\tau}+H_{2,\tau}+I_\tau)\Phi_\tau}}{\norm{\Phi_\tau}^2} \\
&=\frac{\pscal{\Phi_{1,\tau}\otimes\Phi_{2,\tau}, K_\tau \Phi_{1,\tau}\otimes\Phi_{2,\tau}}}{\pscal{\Phi_{1,\tau}\otimes\Phi_{2,\tau}, T_\tau \Phi_{1,\tau}\otimes\Phi_{2,\tau}}} 
\end{align*}
where
\begin{multline*}
K_\tau:=  \Pi_{12,\tau} \left(1- I_{\tau} \Pi_{12,\tau}^\bot R_\tau \Pi_{12,\tau}^\bot \chi_{\tau}\right)\times\\
\times(H_{1,\tau}+H_{2,\tau}+I_\tau)\left(1- \chi_{\tau} \Pi_{12,\tau}^\bot R_\tau \Pi_{12,\tau}^\bot I_{\tau}\right)  \Pi_{12,\tau}
\end{multline*}
and
\begin{equation*}
T_\tau:=  \Pi_{12,\tau} \left(1- \chi_{\tau}  \Pi_{12,\tau}^\bot R_\tau \Pi_{12,\tau}^\bot I_{\tau} \right)^2  \Pi_{12,\tau}
\end{equation*}
on $\cG_{1,\tau} \otimes\cG_{2,\tau}$. One important property of the operators $K_\tau$ and $T_\tau$ is that they commute with all permutations of the spins acting solely on the first and on the second molecule, separately,
$$\pi_1\otimes\pi_2 K_\tau=K_\tau\pi_1\otimes\pi_2 , \qquad \pi_1\otimes\pi_2 T_\tau=T_\tau\pi_1\otimes\pi_2. $$ 
From the irreducibility condition, this implies that $K_\tau, T_\tau$ are  multiples of the identity on $\cG_{1,\tau} \otimes\cG_{2,\tau}$. Indeed,  any eigenspace of e.g. $K_\tau$ in $\cG_{1,\tau} \otimes\cG_{2,\tau}$ is invariant under all the $\pi_1\otimes\pi_2$, so must be equal to the whole space $\cG_{1,\tau} \otimes\cG_{2,\tau}$. This concludes the proof of~\eqref{eq:multiple_identity}.

We can now conclude the argument, using~\eqref{eq:multiple_identity} and~\eqref{eq:estim_off_diagonal_terms}. In the previous section we have computed the energy of any such fixed vector $\Phi_\tau\neq0$ with $\|\Psi_1\|=\|\Psi_2\|=1$ and we found
\begin{multline*}
\frac{\pscal{\Phi_\tau,H_N(Y(\tau),Z)\Phi_\tau}}{\|\Phi_\tau\|^2}=E_1+E_2+\sum_{2\leq n+m\leq 5}\frac{\cF^{(n,m)}(\Psi_1,\Psi_2,U,V)}{L^{n+m+1}}\\
 -\frac{C_{\rm vdW}(\Psi_1,\Psi_2,U,V)}{L^6} +O\left(\frac1{L^7}\right).
\end{multline*}
Inserting then in the Schur complement lower bound~\eqref{eq:lower_abstract_F-S} and using \eqref{eq:estim_off_diagonal_terms}, we conclude, as we wanted, that
\begin{multline*}
\cE(\tau)\geq E_1+E_2+\sum_{2\leq n+m\leq 5}\frac{\cF^{(n,m)}(\Psi_1,\Psi_2,U,V)}{L^{n+m+1}}\\
 -\frac{C_{\rm vdW}(\Psi_1,\Psi_2,U,V)}{L^6} +O\left(\frac1{L^7}\right)
\end{multline*}
and thus we have equality.

\section{Proof of Proposition \ref{prop:fullexp}}\label{sec:Proofevergyexpsmooth}

By adapting the arguments in~\cite{Hunziker-86}, it might be possible to prove that our function $\cE(\tau)$ is smooth. Therefore, at any local (pseudo-)minimum $(U_0,V_0)$ with respect to variations of $U,V\in SO(3)$ only, we must have
$$\nabla_{(U,V)}\cE(L,U_0,V_0)=0,\qquad {\rm Hess}_{(U,V)}\cE(L,U_0,V_0)\geq0.$$
In order to deduce a similar property for the leading multipolar energy, we first need to show that the expansion $\cE(\tau)$  in powers of $L$ is also valid for the first two derivatives with respect to $U,V\in SO(3)$. This was done under some more stringent non-degeneracy assumptions in~\cite[Sec.~3.3.3]{Lewin-04b} for the first derivative. An alternative approach  is provided in \cite{AnaLewRot-19} in the case of two atoms, which works for the second derivative as well. Here we use  ideas of \cite{AnaLewRot-19}, but significant modifications are necessary, mainly because we have to take spin into account. This is harder since the ground state energies can be degenerate, even under the irreducibility assumption.

We provide here a simpler argument which only uses the smoothness of the Feshbach map at fixed energy and does not rely on  smoothness of $\cE(\tau)$. 
Let $L\gg1$ and assume that $(U_0,V_0)$ is a local pseudo-minimum of $(U,V)\mapsto \cE(L,U,V)$. Our goal is to show that 
$$\left|\nabla_{U,V}\left(\sum_{2\leq n+m\leq 5}\frac{\cF^{(n,m)}(\Psi_1,\Psi_2,U_0,V_0)}{L^{n+m+1}} -\frac{C_{\rm vdW}(\Psi_1,\Psi_2,U_0,V_0)}{L^6}\right)\right|\leq \frac{C}{L^7}$$
and
\begin{multline*}
{\rm Hess}_{U,V}\left(\sum_{2\leq n+m\leq 5}\frac{\cF^{(n,m)}(\Psi_1,\Psi_2,U_0,V_0)}{L^{n+m+1}} -\frac{C_{\rm vdW}(\Psi_1,\Psi_2,U_0,V_0)}{L^6}\right)\\
\geq -\frac{C}{L^7} 
\end{multline*}
for some $C$ independent of $U_0,V_0,L$. Multiplying by $L^{n_1+n_2+1}$ where $n_1$ and $n_2$ are the indices of the first non-vanishing multipoles, we will conclude, as desired, that 
$$\left|\nabla_{U,V}\cF^{(n_1,n_2)}(\Psi_1,\Psi_2,U_0,V_0)\right|\leq \frac{C}{L}$$
and
$${\rm Hess}_{U,V}\cF^{(n_1,n_2)}(\Psi_1,\Psi_2,U_0,V_0)\geq -\frac{C}{L}.$$

From the Feshbach-Schur method recalled in the beginning of the previous section, we know that $E=\cE(\tau)$ is the unique solution to the equation
\begin{equation}
 \min\sigma\left(\Pi_\tau H_\tau\Pi_\tau-\Pi_\tau H_\tau\Pi_\tau^\perp \big(H_\tau-E\big)^{-1}_{|(\Upsilon_\tau)^\perp}\Pi_\tau^\perp H_\tau\Pi_\tau\right)=E
 \label{eq:reminder_FS}
\end{equation}
in a neighborhood of $E_1+E_2$, where we use here the shorthand notation
$$H_\tau:=H_N(Y(\tau),Z).$$
We recall that $\Pi_\tau$ is the orthogonal projection onto the space $\Upsilon_\tau$ defined above in~\eqref{eq:space_used_for_F-S}.
Note that the left side of~\eqref{eq:reminder_FS} is decreasing in $E$ whereas the right side is increasing. 
Let now 
\begin{multline*}
\Phi_{(U_0,V_0)}=\cQ \sum_{m=1}^{\dim(\cG_1)}  \sum_{n=1}^{\dim(\cG_2)} a_{mn}\times\\ \times \left(\Phi_{1,\tau_0}^{(m)}\otimes\Phi_{2,\tau_0}^{(n)} - \chi_{\tau_0} \Pi_{12,\tau_0}^\bot R_{\tau_0} \Pi_{12,\tau_0}^\bot I_{\tau_0}  \Phi_{1,\tau_0}^{(m)} \otimes \Phi_{2,\tau_0}^{(n)}\right)\in\Upsilon_{\tau_0} 
\end{multline*}
be a normalized eigenfunction of the Feshbach map at the energy $E_0:=\cE(L,U_0,V_0)$. This time $\Phi_{i,\tau}^{(k)}, k=1,...,\dim(\cG_i)$ are exact orthonormal bases of the respective cut off ground state eigenspaces $\cG_{i,\tau}.$
We define 
\begin{equation}\label{def:PhiUV}
\Phi_{(U,V)}:=\cQ  \sum_{m=1}^{\dim(\cG_1)}  \sum_{n=1}^{\dim(\cG_2)} a_{mn} \Psi_{(U,V)}^{(m,n)} \in\Upsilon_{\tau} ,
\end{equation}
where
\begin{multline}\label{def:PsiUV}
\Psi_{(U,V)}^{(m,n)}= UU_0^*\Phi_{1,\tau_0}^{(m)}\otimes VV_0^*\Phi_{2,\tau_0}^{(n)}\\ - \chi_{\tau} \Pi_{12,\tau}^\bot R_\tau \Pi_{12,\tau}^\bot I_{\tau}  UU_0^*\Phi_{1,\tau_0}^{(m)}\otimes VV_0^*\Phi_{2,\tau_0}^{(n)}.
\end{multline}
The function $\Phi_{(U,V)}$ is the function $\Phi_{(U_0,V_0)}$
 appropriately rotated. Note that $\Phi_{(U,V)}$ is not necessarily  the first eigenfunction of the Feshbach map corresponding to $U,V$, but it can always be used as a trial function. 
We finally introduce the function
\begin{equation}
\boxed{\cF(L,U,V):=\pscal{\frac{\Phi_{(U,V)}}{\|\Phi_{(U,V)\|}}, \left(H_\tau- H_\tau\Pi_\tau^\perp \big(H_\tau-E_0\big)^{-1}_{|(\Upsilon_\tau)^\perp}\Pi_\tau^\perp H_\tau \right) \frac{\Phi_{(U,V)}}{\|\Phi_{(U,V)}\|}}.}
\label{eq:trick_F-S-upper-bound}
\end{equation}
In other words, we rotate the trial function, keep the energy $E$ fixed to its value $E_0=\cE(L,U_0,V_0)$, and apply the Feshbach-Schur map. If $U,V$ are such that 
$$\cE(L,U,V)\geq \cE(L,U_0,V_0)=E_0=\cF(L,U_0,V_0),$$ 
then we have
\begin{align*}
\cE(L,U,V)&\leq \frac{\pscal{\Phi_{(U,V)}, \left(H_\tau- H_\tau\Pi_\tau^\perp \big(H_\tau-\cE(L,U,V)\big)^{-1}_{|(\Upsilon_\tau)^\perp}\Pi_\tau^\perp H_\tau \right)\Phi_{(U,V)}}}{\|\Phi_{(U,V)}\|^2} \\
 &\leq \frac{\pscal{\Phi_{(U,V)}, \left(H_\tau- H_\tau\Pi_\tau^\perp \big(H_\tau-E_0\big)^{-1}_{|(\Upsilon_\tau)^\perp}\Pi_\tau^\perp H_\tau \right)\Phi_{(U,V)}}}{\|\Phi_{(U,V)}\|^2} \\
 &=\cF(L,U,V).
\end{align*}
In particular, 
$$\cF(L,U,V)\geq \cF(L,U_0,V_0)=E_0$$
and we deduce that $(U_0,V_0)$ is a local pseudo-minimum of $\cF(L,\cdot,\cdot)$ as well. The function $\cF$ has the same expansion as $\cE$, to leading order, but what we have gained here is that it is much easier to prove that $\cF$ is regular, since $E_0$ appears in its definition instead of $\cE(\tau)$ and since the test function is appropriately chosen.

\begin{lemma}[Regularity of $\cF$ along geodesics]\label{lem:regularity_f}
	Under the assumptions of Theorem~\ref{thm:expansion_energy}, there exists $L_0>0$ and $C>0$ such that for any $L>L_0$ and any $C^2$ geodesic $(U_t, V_t)$ in $SO(3)^2$  
	\begin{multline}
	\bigg\|\cF(L,U_t,V_t)-\sum_{2\leq n+m\leq 5}\frac{\cF^{(n,m)}(\Psi_1,\Psi_2,U_t,V_t)}{L^{n+m+1}}\\
	+\frac{C_{\rm vdW}(\Psi_1,\Psi_2,U_t,V_t)}{L^6}\bigg\|_{C^2([-1,1])}\leq \frac{C}{L^7}.
	\label{eq:expansion_f}
	\end{multline}
\end{lemma}

The end of the argument follows and it remains to provide the proof of the Lemma. Because $\cF^{(m,n)}$ and $C_{\rm vdW}$ are  smooth functions of $(U,V)$, proving inequalities on their gradient and Hessian is reduced to showing the respective inequalities for their time derivatives along geodesics.   Thus in the rest of Section we prove Lemma~\ref{lem:regularity_f}.

\begin{proof}[Proof of Lemma~\ref{lem:regularity_f}]
	Under the irreducibility condition, we have proved in the previous section that $\Pi_\tau H_\tau\Pi_\tau$ is a multiple of $\Pi_\tau$, hence using $\Phi_{(U,V)}$ or a true eigenstate of the Feschbach map does not matter for this term. Thus, from \eqref{eq:testsmoothness}  it follows  that for any geodesic $(U_t,V_t)$ in $SO(3)$, we have 
	\begin{multline}
	\bigg\| \frac{1}{\|\Phi_{(U_t,V_t)}\|^2} \pscal{\Phi_{(U_t,V_t)},H_\tau\Phi_{(U_t,V_t)}}-E_1-E_2 -\!\!\sum_{2\leq n+m\leq 5}\frac{\cF^{(n,m)}(\Psi_1,\Psi_2,U_t,V_t)}{L^{n+m+1}}\\
	+\frac{C_{\rm vdW}(\Psi_1,\Psi_2,U_t,V_t)}{L^6}\bigg\|_{C^2([-1,1])}\leq \frac{C}{L^7}.
	\label{eq:expansion_f_diagonal_term}
	\end{multline}
	Therefore, we only have to show that the off-diagonal term
	$$Q(U,V):=\frac{1}{\|\Phi_{(U,V)}\|^2} \pscal{\Phi_{(U,V)}, H_\tau\Pi_\tau^\perp \big(H_\tau-E_0\big)^{-1}_{|(\Upsilon_\tau)^\perp}\Pi_\tau^\perp H_\tau \Phi_{(U,V)}}$$
	fullfills
	\begin{equation}\label{eq:QUVest}
	\|Q(U_t,V_t)\|_{C^2([-1,1])} = O\left(\frac{1}{L^8}\right).
	\end{equation}
	uniformly in the geodesics $(U_t, V_t)$ in $SO(3)^2$. In Section~\ref{sec:setimate_off_diagonal_term} we have already shown that 
	$$|Q(U,V)|\leq C\norm{\Pi_\tau^\perp H_\tau\Pi_\tau}^2=O\left(\frac{1}{L^8}\right).$$
	Using that $\Phi_{U,V}=\cQ \Phi_{U,V}$ constists, due to the antisymmetrization, of a sum of functions that have disjoint supports and all have the same norm, we obtain that
	$$\|\Phi_{(U,V)}\|^2=\frac{1}{ {N_1 + N_2 \choose N_1}} \left\|\sum_{m=1}^{\dim(G_1)}  \sum_{n=1}^{\dim(G_2)} a_{nm} \Psi_{U,V}^{(m,n)}\right\|^2.$$
	Since $K_\tau$ is a multiple of the identity on $\cG_{1,\tau} \otimes \cG_{2,\tau}$, we find that $\Psi_{U,V}^{(m_1,n_1)}$ and $\Psi_{U,V}^{(m_2,n_2)}$ are orthogonal to each other when $(m_1,n_1) \neq (m_2,n_2)$. Therefore,
	$$\|\Phi_{(U,V)}\|^2= \frac{1}{ {N_1 + N_2 \choose N_1}} \sum_{m=1}^{\dim(G_1)}  \sum_{n=1}^{\dim(G_2)} |a_{nm}|^2 \|\Psi_{U,V}^{(m,n)}\|^2,$$
	and, since we assumed that $\Phi_{(U_0,V_0)}$ is normalized arguing as in the proof of  \eqref{eq:final_estim_norm}, it follows that
	$$\|\Phi_{(U,V)}\|^2 = 1 + O_{C^2(S0(3)^2)}\left(\frac{1}{L^6}\right) .$$
	As a consequence to conclude the proof of \eqref{eq:QUVest} and therefore of Lemma \ref{lem:regularity_f} it is enough to show that
	\begin{equation}\label{eq:secdersmall}
	\big\|X(U_t,V_t)\big\|_{C^2([-1,1])}=O\left(\frac{1}{L^8}\right),
	\end{equation}
uniformly on geodesics $(U_t, V_t)$	where 
	$$X(U,V):= \pscal{\Phi_{(U,V)}, H_\tau\Pi_\tau^\perp \big(H_\tau-E_0\big)^{-1}_{|(\Upsilon_\tau)^\perp}\Pi_\tau^\perp H_\tau \Phi_{(U,V)}}$$
	is almost the same as $Q(U,V)$ except that the normalization of $\Phi_{(U,V)}$ is dropped. 
	The following lemma is similar to Lemma 3.2 in \cite{AnaLewRot-19} and will help us estimate $ \frac{d}{dt}\Pi_{\tau_t}^\perp H_{\tau_t} \Phi_{(U_t,V_t)}$, where $\tau_t=(L,U_t,V_t)$.
	For a function $\Psi \in L^2(\R^{3(N_1+N_2)})$ we define 
	$$(s_{U,V}\Psi)(X_1, X_2)=\Psi(U^{-1} X_1, V^{-1} (X_2 -  Le_1)),$$
	where $X_1=(x_1,\dots,x_{N_1})$, $X_2=(x_{N_1+1},\dots,x_{N_1+N_2})$ and $UX_1, VX_2$ is defined similarly as in \eqref{def:UY}.
	\begin{lemma}\label{lem:compodiff}
		Let $\Psi_{(.,.)}: SO(3)^2 \rightarrow L^2(\R^{3(N_1+N_2)}) $ $(U,V) \rightarrow \Psi_{(U,V)}$ such that $\Psi(U_t,V_t)$ is in $C^1([-1,1];L^2)$, for any geodesic $(U_t,V_t)$ and  so that $\Psi_{(U,V)} \in H^1$ for all $U,V$. Then 
		$s_{U_t,V_t} \Psi_{(U_t,V_t)}$ is in  $C^1([-1,1];L^2)$  for all geodesics as well and
		\begin{equation}\label{eq:compodiff}
		\frac{d}{dt}(s_{U_t,V_t} \Psi_{(U_t,V_t)})\Big|_{t=0}
		=   s_{(U_t,V_t)}  \frac{d}{dt}\Psi_{(U_t,V_t)}\Big|_{t=0} + \vec{w} \cdot \nabla (s_{(U_0,V_0)} \Psi_{(U_0,V_0)}),
		\end{equation}
		where  $\vec{w}(X_1,X_2):=\frac{d}{dt}(U_t^{-1} X_1, V_t^{-1}(X_2-L e_1))|_{t=0}$.
	\end{lemma}
	\begin{proof}
 It suffices to observe that 
\begin{multline*}
\frac{s_{U_t, V_t} \Psi_{(U_t, V_t)}- s_{U_0,V_0} \Psi_{(U_0,V_0)}}{t}=\frac{(s_{U_t, V_t} - s_{U_0,V_0}) }{t} \Psi_{(U_0,V_0)}\\ +  s_{U_t, V_t} \frac{d}{ds}\Psi_{(U_s,V_s)}\Big|_{s=0} 
+ s_{U_t, V_t}  \left(\frac{ (\Psi_{(U_t, V_t)}- \Psi_{(U_0,V_0)})}{t}- \frac{d}{ds}\Psi_{(U_s,V_s)}\Big|_{s=0}\right)
\end{multline*} 
and to take the limit $t \rightarrow 0$.
	\end{proof}
	Recall that for a geodesic $(U_t,V_t)$ in $SO(3)^2$ we defined $\tau_t=(L,U_t,V_t)$, where $L$ is fixed.
	Our next goal is to use Lemma \ref{lem:compodiff} in order to prove that
	\begin{equation}\label{est:PbotHPder}
	\frac{d}{dt}\Pi_{\tau_t}^\perp H_{\tau_t} \Phi_{(U_t,V_t)}\big|_{t=0}=O\left(\frac{1}{L^4}\right),
	\end{equation}
	uniformly in all possible geodesics. 
	
    We observe that
    $$\Pi_\tau=\sum_{m,n} \frac{1}{\|\cQ \Psi_{(U,V)}^{(m,n)}\|^2} | \cQ \Psi_{(U,V)}^{(m,n)} \rangle \langle \cQ \Psi_{(U,V)}^{(m,n)}|, $$
    which, together with \eqref{def:PhiUV}, the orthogonality of the functions $\Psi_{(U,V)}^{(m,n)}$,
	the fact that $\cQ$ commutes with $ H_\tau$ and the irreducibility of $\cG_{1,\tau}, \cG_{2,\tau}$, implies
	\begin{multline}
	\Pi_\tau^\perp H_\tau \Phi_{(U,V)}= 
	\Pi_\tau^\perp (H_\tau-E_1-E_2) \Phi_{(U,V)}
	\\= \cQ \sum_{m=1}^{\dim(\cG_1)} \sum_{n=1}^{\dim(\cG_2)} a_{mn} \Pi_{\Psi_{(U,V)}^{(m,n)}}^\perp   (H_\tau-E_1-E_2)   \Psi_{(U,V)}^{(m,n)},
	\end{multline}
	where  $\Pi_{\Psi_{(U,V)}^{(m,n)}}$ denotes the orthogonal projection into the span of $\Psi_{(U,V)}^{(m,n)}$.
	Thus, proving \eqref{est:PbotHPder} reduces in light of \eqref{def:Q} to proving that for all $\pi \in S_N$, $m \in \{1,\dots,\dim(\cG_1)\}$, $n \in \{1,\dots,\dim(\cG_2)\}$ 
	\begin{equation}\label{eq:HminusEpsismall}
	\left\|\frac{d}{dt} \pi \Pi_{\Psi_{(U_t,V_t)}^{(m,n)}}^\bot (H_{\tau_t}-E_1-E_2) \Psi_{(U_t,V_t)}^{(m,n)}\right\|=O\left(\frac{1}{L^4}\right).
	\end{equation}
	We will  prove \eqref{eq:HminusEpsismall} for the case that $\pi$ is the identity as the other permutations can be similarly handled.
	We may also assume without loss of generality that $U_0,V_0$ are both the identity. 
	Using \eqref{def:PsiUV}, we obtain similarly as in \eqref{eq:H1change} that
	\begin{multline}\label{est:nablaUV1}
	\Pi_{\Psi_{(U,V)}^{(m,n)}}^\bot (H_\tau-E_1-E_2) \Psi_{(U,V)}^{(m,n)}\\
	=s_{U,V} \left( \Pi_{\Theta_{(U,V)}^{(m,n)}}^\bot (H_{1} +  H_2 + \widetilde{I_\tau}-E_1-E_2) \Theta_{(U,V)}^{(m,n)}\right),
	\end{multline}
	where
	\begin{equation}
	\Theta_{(U,V)}^{(m,n)}=
	\Phi_{1,L}^{(m)}\otimes \Phi_{2,L}^{(n)} - \chi_{4L/3}^{\otimes N}  R_{12,L}  \widetilde{I_\tau}  \Phi_{1,L}^{(m)}\otimes \Phi_{2,L}^{(n)}. 
	\end{equation}
	Moreover \eqref{eq:H1estusefull}-\eqref{eq:H1change} give  that
	\begin{equation}\label{est:nablaUV2}
	\left\| \Pi_{\Theta_{(U,V)}^{(m,n)}}^\bot (H_1 + H_2 + \widetilde{I_\tau}-E_1-E_2) \Theta_{(U,V)}^{(m,n)}\right\|_{H^1}=O\left(\frac{1}{L^4}\right).
	\end{equation}
	The  dependence of
	$\Pi_{\Theta_{(U,V)}^{(m,n)}}^\bot (H_1 + H_2 + \widetilde{I_\tau}-E_1-E_2) \Theta_{(U,V)}^{(m,n)}$
	on $U,V$ appears only in $\widetilde{I_\tau}$. Thus with the help of \eqref{eq:Ifuv} the estimate
	$$ \Pi_{\Theta_{(U,V)}^{(m,n)}}^\bot (H_1 + H_2 + \widetilde{I_\tau}-E_1-E_2) \Theta_{(U,V)}^{(m,n)}=O_{L^2}\left(\frac{1}{L^4}\right)$$
	can be upgraded to 
	\begin{equation}\label{est:nablaUV3}
	\nabla_{(U,V)} \Pi_{\Theta_{(U,V)}^{(m,n)}}^\bot ( H_1 + H_2 + \widetilde{I_\tau}-E_1-E_2)\Theta_{(U,V)}^{(m,n)}=O_{L^2}\left(\frac{1}{L^4}\right).
	\end{equation}
	Using \eqref{est:nablaUV1}, \eqref{est:nablaUV2}, \eqref{est:nablaUV3}, Remark \ref{rem:powerx} and Lemma \ref{lem:compodiff}
	we arrive at \eqref{eq:HminusEpsismall} when $\pi$ is identity. For any other permutation the estimate \eqref{eq:HminusEpsismall}
	can be proven similarly.
	Note that Remark \ref{rem:powerx} is needed because  in the second summand of the right hand side of \eqref{eq:compodiff} we have $\vec{w}(X_1,X_2) =\|(X_1, X_2-Le_1)\|$, and because of the exponential decay of $s_{(U_0,V_0)} \Psi_{(U_0,V_0)}$ is  in the variables $X_1, X_2-Le_1$.

	Iterating the above argument we can find with Lemma \ref{lem:compodiff} that 
	$$\left\|(1-\sum_{j=1}^N \Delta_{x_j})^{-\frac{1}{2}} \Pi_{\tau_t}^\perp H_{\tau_t} \Phi_{(U_t,V_t)}\right\|_{C^2([-1,1])}=O\left(\frac{1}{L^4}\right).$$
	Indeed,  we have that 
	$$\left\|\Pi_{\Theta_{(U_t,V_t)}^{(m,n)}}^\bot (H_{1} +  H_2 + \widetilde{I_{\tau_t}}-E_1-E_2) \Theta_{(U_t,V_t)}^{(m,n)}\right\|_{C^2([-1,1])}=O\left(\frac{1}{L^4}\right)$$ because as we mentioned above the dependence of the left hand side on $U,V$ appears  only in $\widetilde{I_\tau}$. So only the directional derivative 
	$\vec{w} \cdot \nabla$
	 could cause a problem in the iteration, but the presence of  the gradient is remedied by $(1-\sum_{j=1}^N \Delta_{x_j})^{-\frac{1}{2}}$.
	Therefore, since we can introduce $(1-\sum_{j=1}^N \Delta_{x_j})^{-1/2}$ due to the presence of the resolvent in the definition of $X(U,V)$, it turns out that the only term which can pose problems in proving \eqref{eq:secdersmall} is the resolvent 
	$$\cR(U,V)=\Pi_\tau^\perp \big(H_\tau -E_0\big)^{-1}_{|\Upsilon_\tau^\perp} \Pi_\tau^\perp= \big(H_\tau^\perp -E_0\big)^{-1}_{|\Upsilon_\tau^\perp} ,$$
	when we differentiate with respect to $U,V$.  In order to explain the method, we remark that, by Leibniz' rule, it is sufficient to show that the map
	$$(U,V)\mapsto \pscal{\Psi,\cR(U,V)\Psi}$$
	is $C^2$, for any fixed $\Psi\in Q\Upsilon_\tau^\bot$. A calculation shows that on $Q\Upsilon_\tau^\bot$ we have 
	\begin{align*}
	\nabla_{(U,V)} \cR(U,V)=&\cR(U,V)(\nabla_{(U,V)} \Pi_\tau) H_\tau \cR(U,V)\\
	&+\cR(U,V) H_\tau (\nabla_{(U,V)} \Pi_\tau)\cR(U,V)\\
	&-\cR(U,V)(\nabla_{(U,V)} H_\tau)\cR(U,V)
	\end{align*}
	where we recall that $\Pi_\tau$ is a projection onto a finite space of $H^2$ functions. The first two terms are differentiable once more. The last one is not, because the second derivative of $\nabla_{(U,V)} H_\tau$ is too singular. By Leibniz's rule again, we are finally left with the study of the differentiability of the function 
	$$(U,V)\mapsto \pscal{\Psi',\nabla_{(U,V)} H_\tau \Psi'}=\nabla_{(U,V)}\pscal{\Psi',I^\tau\Psi'}$$
	for a fixed $\Psi'\in H^2(\R^{3N})$, where $I^\tau$ contains all the interaction terms even within a molecule. Let $X=(x_1,\cdots,x_N)$. After changing variables we see that the inner product on the right is equal to (discarding the spin for simplicity of notation)
	\begin{align*}
	\pscal{\Psi', I^\tau \Psi'}= R  
	&-\sum_{j=1}^{N}\sum_{m=M_1+1}^{M_1+M_2}\int_{\R^{3N}}\frac{z_m}{|x_j|}\big|\Psi'(X+Vy_m+Le_1)\big|^2\\
	&-\sum_{j=1}^N\sum_{\ell=1}^{M_1}\int_{\R^{3N}}\frac{z_\ell}{|x_j|}\big|\Psi'(X+Uy_\ell)\big|^2\\
	&+\sum_{\ell=1}^{M_1}\sum_{m=M_1+1}^{M_1+M_2}\frac{z_\ell z_m}{|Uy_\ell-Vy_m-Le_1|}\int_{\R^{3N}}|\Psi'(X)|^2,
	\end{align*}
	where $R$ consists of terms that are invariant with respect to $U,V$.
	This is twice differentiable for $\Psi'\in H^2(\R^{3N})$. 
	By arguing as before, we can then verify that the terms obtained after differentiation are all $O(L^{-8})$, which ends the proof of Lemma~\ref{lem:regularity_f}.
\end{proof}

This now concludes the proof of Proposition~\ref{prop:fullexp}.\qed

\section{Proof of Propositions \ref{prop:localmin} and \ref{prop:connected}}\label{sec:prooflocmincon}

In this section we prove Propositions \ref{prop:localmin} and \ref{prop:connected} for all $\cF^{(n,m)}$ with $n+m\leq 5$, by looking at all the possible values of $n$ and $m$.

\subsection{Dipole-dipole term}

The results mentioned here for the dipole-dipole interaction follow from the fact, proved in~\cite[Lemma~12]{Lewin-04b}, that $\cF^{(1,1)}$ has no critical point of zero energy, and that all its critical points of positive energy have a Morse index larger than or equal to $2$. We however give a different proof of Propositions \ref{prop:localmin} and \ref{prop:connected} here, since the arguments will be useful for all the higher order terms involving a dipole. 

In light of \eqref{eq:dd} and assuming without loss of generality that the dipole moments are unit vectors
the problem reduces to studying the function
$$f(p_1,p_2):=-3(p_1 \cdot e_1) (p_2 \cdot e_1)+ p_1 \cdot p_2,$$
on $S^2 \times S^2$. This was proved to have the critical levels $\{-2,-1,1,2\}$ in~\cite[Lemma~12]{Lewin-04b} but we will not directly use this information.

We first prove Proposition \ref{prop:localmin}. Observe that 
$f(p_1,p_2)=X_{p_1} \cdot p_2$, where
 $X_{p_1}:=p_1-3(p_1 \cdot e_1) e_1$ is never $0$. Then there exists a unit vector $v$ orthogonal to $X_{p_1}$ and an angle $\theta$
such that $p_2=\cos(\theta) R_{p_1} +\sin(\theta) v$, where $R_{p_1}:=-\frac{X_{p_1}}{\|X_{p_1}\|}$. Let $p_2(t)=\cos(\theta+t) R_{p_1}+\sin(\theta + t)v$ and $g(t):=f(p_1,p_2(t))= -\cos(\theta+t)\|X_{p_1}\|$. 
Then 
$$g(0)=-\cos(\theta)\|X_{p_1}\|,\qquad  g'(0)=\sin(\theta)\|X_{p_1}\|,\qquad g''(0)=-g(0).$$
From these equalities and the fact that $\|X_{p_1}\| \geq 1$ it follows that if $|g'(0)| \leq \frac{1}{3}$ and $g''(0) \geq -\frac{1}{3}$ then since $g'(0)^2+g''(0)^2 \geq 1$ we have that $g''(0) \geq \frac{2 \sqrt{2}}{3}$ and thus $g(0)=-g''(0) \leq -\frac{2 \sqrt{2}}{3}$. Thus if the Hessian is bigger than $-\frac{1}{3}$ and the derivative has norm less then $\frac{1}{3}$ then the value of the function is less than 
$-\frac{2 \sqrt{2}}{3}$ which concludes the proof of Proposition \ref{prop:localmin} in the case of the dipole-dipole term.

We now prove Proposition \ref{prop:connected}.  We choose
\begin{equation}
\delta_0:= \frac{1}{2}.
\end{equation}
We now consider $\delta<\delta_0$, and two points $(p_1,p_2)$ and $(p_1',p_2')$ such that $f(p_1,p_2)<-\delta$, $f(p_1',p_2')<-\delta$. We connect   $p_2$  to 
$R_{p_1}$, following the great circle of $S^2$ 
 in the plane $(R_{p_1},p_2)$, in the direction which avoids $-R_{p_1}$. Similarly we connect $p_2'$ to $R_{p_1'}$.
 Along these paths the function $f$ is decreasing.
Thus we can assume without loss of generality that $p_2=R_{p_1}$, $p_2'=R_{p_1'}$.
We denote by  $\gamma:[0,1] \rightarrow S^2$ an arbitraty path starting from $p_1$ and ending to $p_1'$. Then  the path 
$(\gamma(t),R_{\gamma(t)})$ connects $(p_1,p_2)$ to $(p_1',p_2')$ and since $f(\gamma(t),R_{\gamma(t)})=-\|X_{\gamma(t)}\| \leq-1 < -\delta_0$, along this path we have that $f<-\delta$. This concludes the proof of Proposition \ref{prop:connected} in the case of the dipole-dipole term.

\subsection{Dipole-quadrupole term}
When the dipole term is involved the argument is simplified considerably, as everything can be written as a function of the two normalized vectors $$e=V^{-1}e_1,\qquad p=V^{-1}UD_1/\|D_1\|.$$ 
Then any dipole-multipole term becomes a function on $S^2 \times S^2$ and in particular so does the dipole-quadrupole term. With this observation one can see that Propositions \ref{prop:localmin} and \ref{prop:connected} for the case of the dipole-quadrupole term follow immediately from the following lemma. The rest of the section is devoted to its proof. 
\begin{lemma}[Dipole-quadrupole]\label{lem:dipquad}
	Consider the function $f: S^2 \times S^2 \subset \R^3 \times \R^3 \to \R$ defined by
	\begin{equation}
	f(e,p)= 5 \lan p,e \ran \lan e, Q e \ran-2 \lan e, Q p \ran,
	\end{equation}
	where $\lan\cdot ,\cdot \ran$ is the standard inner product in $\R^3$ and the $3\times 3$ matrix $Q$ is real-symmetric with $Q \neq 0$ and $\tr Q=0$.
\begin{itemize}
 \item[a)] There exists $\delta_0>0$ such that for all $\delta<\delta_0$ the set $A=\{(e,p): f(e,p)<-\delta\}$ is non-empty and pathwise connected.
 
 \smallskip
 
\item[b)] There exists $\delta >0$ such that for all  $(e,p) \in S^2 \times S^2$ the following holds: if $\|\nabla_{e,p}f(e,p)\| \leq \delta$ and ${\rm Hess}_{e,p} f(e,p) \geq -\delta$ then we have that $f(e,p) \leq -\delta$. 
\end{itemize}
\end{lemma}
\begin{proof}
	a) The function $f$ can be rewritten as
	\begin{equation}\label{eq:finner}
	f(e,p)= \lan 5 P_e Q e-2 Q e, p \ran,
	\end{equation}
	where $P_e$ denotes the orthogonal projection onto $e\R$.
	
	 Assume first that $\ker(Q)=\{0\}$. Then the vector $5 P_e Qe-2Qe$ is always nonzero. We choose
	 \begin{equation}\label{eq:d0}
	  \delta_0:= \min\{\|5 P_e Qe-2Qe\|: e \in S^2\}.
	 \end{equation}
	  We now consider $\delta<\delta_0$, and two points $(e_1,p_1)$ and $(e_2,p_2)$ such that $f(e_i,p_i)<-\delta$. Here $e_1, e_2$ are not to mix up with the standard basis vectors. We connect each $p_i$  to 
$$R_{e_i}:=- \frac{5 P_{e_i} Q e_i-2 Q e_i}{\|5 P_{e_i} Q e_i-2 Q e_i \|}$$ 
following the great circle in the plane $(R_{e_i},p_i)$, in the direction which avoids $-R_{e_i}$, similarly as in the case of the dipole-dipole term. Along this path the function $f$ is decreasing to $f(e_i,R_{e_i})=-\|5 P_{e_i} Q{e_i}-2Q{e_i}\|$.
	Thus we can assume without loss of generality that $p_i=R_{e_i}$. Consider now any continuous path $e(t)$ on $S^2$ which connects $e_1$ and $e_2$. Taking $p(t)=R_{e(t)}$ we obtain a continuous path (since $e\mapsto R_e$ is continuous), with 
	$$f(e(t),R_{e(t)})=-\|5 P_{e(t)} Q{e(t)}-2Q{e(t)}\|\leq-\delta_0$$
	as desired.
	
	Assume now that $Q$ has the eigenvalue $0$. Since $Q$ is nonzero and traceless 
	the eigenvalue $0$ has multiplicity 1. Let $\pm e_0$ denote the eigenvectors in $S^2$ of the eigenvalue $0$. Note that on $S^2$ we have $ 5 P_{e} Q e-2 Q e=0$ if and only if $e=\pm e_0$. 
	We choose $$\delta_0:=\|Q\|.$$ 
	We consider   $(e_1,p_1)$ and $(e_2,p_2)$ as before and deduce from the assumption $f(e_i,p_i)\leq -\delta<0$ that $e_i\neq \pm e_0$ for $i=1,2$. 
We may thus assume, as in the previous paragraph, that $p_i= R_{e_i}$. Then we construct a path starting from   $(e_1,p_1)$  as follows: By replacing $e_0$ with $-e_0$ if necessary, we can assume without loss of generality that  $e_1=e(\theta):= \cos(\theta)e_0 + \sin(\theta) e_1^\bot$,
	where $\theta \in (0,\frac{\pi}{2}]$ and
	$e_1^\bot$ is in the orthogonal complement of $e_0$ (and not of $e_1$).
	Then we consider the path $(e(\phi), R_{e(\phi)})$, $\phi \in [\theta, \frac{\pi}{2}]$. After taking the square it follows that $\|5 P_{e_1^\bot}Qe_1^\bot-2 Q e_1^\bot\| \geq 2 \|Q e_1^\bot\|= 2 |\lambda|=2 \|Q\|$, where $\pm \lambda$ are the other eigenvalues of $Q$. Thus, $f(e_1^\bot, R_{e_1^\bot}) \leq -2 \|Q\| < -\delta_0$. Using that  
	$$\|5 P_{e(\phi)} Qe(\phi)-2 Q e(\phi)\|^2=5 \langle e(\phi), Q e(\phi) \rangle^2 + 4 \|Q e(\phi)\|^2,$$ 
	we can show that $f$ is decreasing along the path $(e(\phi), R_{e(\phi)})$.
	We connect in a similar way $(e_2,p_2)$ to the point
	$(e_2^\bot, R_{e_2^\bot})$. Finally we connect $(e_1^\bot, R_{e_1^\bot})$ to $(e_2^\bot, R_{e_2^\bot})$ along a path $(e(t)^\bot, R_{e(t)^\bot})$, where $e(t)^\bot$ is in the orthogonal complement of $e_0$. Along this path $f$ remains smaller than or equal to $-\delta_0$.
	
	b) As in part a) we first assume that $\ker(Q)=\{0\}$. Let $(e,p) \in S^2 \times S^2$ be arbitrary. Then there exists an angle $\theta \in [0,2 \pi)$ and a unit vector $R_e^\bot$ which is orthogonal to $R_e$ such that $p=\cos(\theta) R_e + \sin(\theta) R_e^\bot$. We now consider the function $g(t)=f(e,\cos(t) R_e + \sin(t) R_e^\bot)$. Then 
	$$g(t)=-\|5 P_e Qe-2Qe\|\cos(t),\qquad g'(t)=\|5 P_e Qe-2Qe\|\sin(t),$$
	$$g''(t)=\|5 P_e Qe-2Qe\|\cos(t).$$ 
	Now choose $\delta$ to be ${\delta_0}/{3}$, where $\delta_0$ is the same as in \eqref{eq:d0}. If $\|\nabla_{e,p}f(e,p)\| \leq \delta$ and ${\rm Hess}_{e,p}f(e,p) \geq -\delta$ then $|g'(\theta)| \leq \delta$ and $g''(\theta) \geq -\delta$. But since $g'(\theta)^2 + g''(\theta)^2=\delta_0^2$ it follows actually that $|g''(\theta)| > \delta$ and since $g''(\theta) \geq -\delta$ we find $g''(\theta)> \delta$. Therefore, $f(e,p)=g(\theta)=-g''(\theta)<-\delta$, completing the proof in the case that $0$ is not an eigenvalue of $Q$. 
	
	Assume now that $Q$ has the eigenvalue zero and let as before $\pm e_0$ denote the eigenvectors  in $S^2$ of the eigenvalue $0$.  We may assume without loss of generality that the other eigenvalues are $\lambda=\pm1$. We will show that there are neighborhoods $U_1,U_2$ of $e_0,-e_0$ and $\delta>0$ such that if $e \in U_1 \cup U_2$ then it is impossible to have $\|\nabla_{e,p}f(e,p)\| \leq \delta$ and ${\rm Hess}_{e,p} f(e,p) \geq -\delta$. Once this is shown we can repeat the argument of the previous paragraph on $(S^2/(U_1 \cup U_2)) \times S^2$ with $\delta_0$ replaced by $\min\{\|5P_e Q_e-2Qe\|: e \in S^2/(U_1 \cup U_2) \}$. 
	  
	  We define 
$U_1=\{\cos(\theta) e_0 + \sin(\theta) e_0^\bot: e_0^\bot \bot e_0, \|e_0^\bot\|=1,  \theta \leq \epsilon\}$, where $\epsilon>0$ is to be chosen and $U_2=-U_1$. Let $e=\cos(\theta) e_0 + \sin(\theta) e_0^\bot \in U_1$ and write any $p$ as $p=\cos(\phi) Qe_0^\bot+\sin(\phi) v$ for some unit vector $v \bot Qe_0^\bot$. Note that due to our assumption $\lambda=1$ the vector $ Qe_0^\bot$ is normalized. 
We will consider only $U_1$, everything works in the same way for $-U_1$. For $t \in (-\frac{1}{2},\frac{1}{2})$ and some $\zeta \in [0,2 \pi]$ we define 
\begin{equation}\label{eq:et}
e(t)=\cos(\theta+\cos(\zeta)t) e_0 + \sin(\theta+\cos(\zeta)t) e_0^\bot
\end{equation}
and $$p(t)=\sin(\phi+\sin(\zeta)t) Qe_0^\bot+\cos(\phi+\sin(\zeta)t) v.$$
Then $(e(t),p(t))$ is a geodesic in $S^2 \times S^2$. We also define $g(t)=f(e(t),p(t))$.
Using that $Q e_0^\bot$ is normalized a simple computation gives 
\begin{multline*}
g(t)=-2\sin(\theta+t\cos(\zeta)) \sin(\phi +t \sin(\zeta))\\ + 5 \sin^2(\theta+t \cos(\zeta)) \langle e_0^\bot, Q e_0^\bot \rangle \langle e(t), p(t) \rangle. 
\end{multline*}
Then using that $|\sin(\theta)| \leq \epsilon$ we find that $$g'(0)=-2 \cos(\zeta) \cos(\theta) \sin(\phi) + O(\epsilon),$$ 
and 
\begin{multline*}
g''(0)=-4 \sin(\zeta) \cos(\zeta) \cos(\theta) \cos(\phi)\\ + 10 \cos(2 \theta) \cos^2(\zeta)\langle e_1^\bot, Q e_1^\bot \rangle \langle e(0), p(0) \rangle +O(\epsilon). 
\end{multline*}
If $|\cos(\phi)| \leq {1}/{\sqrt{2}}$ then we choose $\zeta=0$ and find $|g'(0)|\geq \sqrt{2}-O(\epsilon)$. Thus we may assume that
$|\cos(\phi)| \geq {1}/{\sqrt{2}}$.
We choose $\zeta = \zeta_0$ so that  $-4 \sin(\zeta_0) \cos(\zeta_0) \cos(\theta) \cos(\phi)<0$  and $|\cos(\zeta_0)|={|\sin(\zeta_0)|}/{10}$. Then an elementary computation gives that $g''(0) \leq -\frac{1}{100}+ O(\epsilon)$. Choosing $\delta_0<\frac{1}{100}$ and $\epsilon$ small enough we conclude the proof. 
\end{proof}

\subsection{Dipole-octopole term}
We recall the definition of the octopole moment tensor.
$$ O_{ijk}= \frac{1}{2} \int dz ( 5 z_i z_j z_k - |z|^2  (z_i \delta_{jk}+ z_j \delta_{ik}+ z_k \delta_{ij})) \rho(z) dz e^{ijk}.$$
We denote by $p$ the dipole moment vector of the second molecule. Our goal is to study the function
$f: S^2 \times S^2 \mapsto \R$ with
\begin{equation}
f(e,p)=3 O(e,e,p)-7 O(e,e,e) (e \cdot p).
\label{eq:f_dipole-octopole}
\end{equation}
We denote by $O(v,w,\cdot)$ the vector determined by the equality
$O(v,w,u)= \lan O(v,w,\cdot),u\ran$ and by $O(v,\cdot,\cdot)$ the symmetric matrix such that 
$\lan w, O(v,\cdot,\cdot) u \ran= O(v,w,u)$. To handle the more complicated octopole moment, we need the following auxiliary lemma.

\begin{lemma}\label{lem:Oee0}
	Assume that $O(\cdot,\cdot,\cdot) \neq 0$. Then there exist at most three pairwise non parallel vectors $v_1, v_2, v_3 \in S^2$ such that $$O(v_1,v_1,\cdot)=O(v_2,v_2,\cdot)=O(v_3,v_3,\cdot)=0.$$
\end{lemma}
\begin{proof}
	We assume that two such vectors $v_1,v_2$ exist. Without loss of generality suppose that $v_1=e_1$ and that $v_2=c_1 e_1 + c_2 e_2$, where $c_2 \neq 0$. Then by the assumption we obtain that
	\begin{equation}\label{eq:o11}
	O(e_1,e_1,\cdot)=0 
	\end{equation}
	and
	\begin{equation}\label{eq:o12o220}
	2 c_1 O(e_1,e_2,\cdot)+c_2O(e_2,e_2,\cdot)=0.
	\end{equation}
\smallskip

\noindent \textbf{Case 1}: $c_1=0.$ Then from \eqref{eq:o12o220} it follows that
	\begin{equation}\label{eq:o220}
	O(e_2,e_2,\cdot)=0.
	\end{equation}
	and since $O(e_1,e_1,\cdot)+O(e_2,e_2,\cdot)+O(e_3,e_3,\cdot)=0$ we find from \eqref{eq:o11} and \eqref{eq:o220} that
	\begin{equation}\label{eq:o330}
	O(e_3,e_3,\cdot)=0.
	\end{equation}
	From \eqref{eq:o220} and \eqref{eq:o330} it follows that $O(e_2,e_3,e_2)=O(e_2,e_3,e_3)=0$ and thus there exists $\lambda_1 \in \R$ such that
	\begin{equation}\label{eq:o230}
	O(e_2,e_3,\cdot)=\lambda_1 e_1.
	\end{equation}
	Similarly, we can obtain that
	\begin{equation}\label{eq:o130}
	O(e_1,e_3,\cdot)=\lambda_2 e_2.
	\end{equation}
	and 
	\begin{equation}\label{eq:o120}
	O(e_1,e_2,\cdot)=\lambda_3 e_3.
	\end{equation}
	Taking the inner product with $e_1$ in \eqref{eq:o230}, $e_2$ in  \eqref{eq:o130} and $e_3$ in  \eqref{eq:o120} it follows that $\lambda_1=\lambda_2=\lambda_3=:\lambda$. If $\lambda=0$ then $O(\cdot,\cdot,\cdot)=0$ contradicting our assumption, thus $\lambda \neq 0$. Assume now that 
	$$O(d_1 e_1 + d_2 e_2 + d_3 e_3, d_1 e_1 + d_2 e_2 + d_3 e_3,\cdot) = 0.$$ 
	Then using \eqref{eq:o11} and \eqref{eq:o220}-\eqref{eq:o120} together with $\lambda_1=\lambda_2=\lambda_3=:\lambda$ we find that
	$\lambda (d_1 d_2 e_3 + d_1 d_3 e_2 + d_2 d_3 e_1 )=0$. Therefore, since $\lambda \neq 0$ we obtain that $d_1 d_2= d_2 d_3= d_3 d_1=0$. Thus at least two of the coefficients $d_i$ have to be zero so the vector $d_1 e_1 + d_2 e_2 + d_3 e_3$ is parallel to one of the vectors $e_1,e_2,e_3$.

	\smallskip
	
	\noindent\textbf{Case 2} : $c_1 \neq 0$. Then we may assume without loss of generality that $c_1=1$ and we recall that $c:=c_2 \neq 0$.  Thus \eqref{eq:o12o220} implies that
	\begin{equation}\label{eq:o12o22}
	2  O(e_1,e_2,\cdot)+c\, O(e_2,e_2,\cdot)=0.
	\end{equation}
	Taking the inner product with $e_1$ in~\eqref{eq:o12o22}, it follows from \eqref{eq:o11} that
	\begin{equation}\label{eq:o122}
	O(e_1,e_2,e_2)=0.
	\end{equation}
	Taking in \eqref{eq:o12o22} inner product with $e_2$ and using \eqref{eq:o122}, we find that
	\begin{equation}\label{eq:o222}
	O(e_2,e_2,e_2)=0.
	\end{equation}
	Since $O(e_2,\cdot,\cdot)$ is traceless it follows, using \eqref{eq:o11} and \eqref{eq:o222}, that 
	\begin{equation}\label{eq:o233}
	O(e_2,e_3,e_3)=0.
	\end{equation}
	Similarly, since $O(e_1,\cdot,\cdot)$ is traceless, using \eqref{eq:o11} and \eqref{eq:o122}, we find 
	\begin{equation}\label{eq:o133}
	O(e_1,e_3,e_3)=0.
	\end{equation}
	From \eqref{eq:o133} and \eqref{eq:o233}
	it follows that there exists $\lambda \in \R$ such that
	\begin{equation}\label{eq:o33}
	O(e_3,e_3,\cdot)= \lambda e_3.
	\end{equation}
	Therefore, since 
	$O(e_1,e_1,\cdot) + O(e_2,e_2,\cdot) + O(e_3,e_3,\cdot)=0$ we find using \eqref{eq:o11} and \eqref{eq:o33}
	\begin{equation}\label{eq:o22}
	O(e_2,e_2,\cdot)= - \lambda e_3.
	\end{equation}
	Hence, using \eqref{eq:o12o22} we obtain that
	\begin{equation}\label{eq:o12}
	O(e_1,e_2,\cdot)= \frac{c}{2} \lambda e_3.
	\end{equation}
	From \eqref{eq:o133} and $O(e_1,e_3,e_1)=0$ (see \eqref{eq:o11})
	we find that
	\begin{equation}\label{eq:o13'}
	O(e_1,e_3,\cdot)=\mu_1 e_2.
	\end{equation} 
	Taking the inner product with $e_2$ in \eqref{eq:o13'} and with $e_3$ in \eqref{eq:o12} we see that 
	\begin{equation}\label{eq:mu1lambda}
	\mu_1= \frac{c}{2} \lambda.
	\end{equation}
	and thus 
	\begin{equation}\label{eq:o13}
	O(e_1,e_3,\cdot)=\frac{c}{2} \lambda e_2.
	\end{equation} 
	From \eqref{eq:o233}  if follows that
	\begin{equation}\label{eq:o23'}
	O(e_2,e_3,\cdot)= \mu_2 e_1 + \mu_3 e_2.
	\end{equation}
	Taking the inner product with $e_2$ in \eqref{eq:o23'} and with $e_3$ in \eqref{eq:o22} we find that $\mu_3=-\lambda$. Taking then the inner product with $e_1$ in \eqref{eq:o23'} and with $e_2$ in \eqref{eq:o13} we find that $\mu_2=\frac{c}{2} \lambda$.
	Thus,
	\begin{equation}\label{eq:o23}
	O(e_2,e_3,\cdot)= \lambda( \frac{c}{2}e_1 -  e_2).
	\end{equation}
	From \eqref{eq:o11},  \eqref{eq:o33}, \eqref{eq:o22},    \eqref{eq:o12},  \eqref{eq:o13},   \eqref{eq:o23}, it follows that if $\lambda=0$ then $O(\cdot,\cdot,\cdot)=0$ contradicting our assumption. Thus $\lambda \neq 0$.
	We assume now that for some $d_1, d_2, d_3$ we have that
	\begin{equation}\label{eq:O0}
	O(d_1 e_1 + d_2 e_2 + d_3 e_3, d_1 e_1+d_2 e_2 + d_3 e_3,\cdot)=0.
	\end{equation}
	Since the right hand side of the last equation is a linear combination of the vectors in \eqref{eq:o11},  \eqref{eq:o33}, \eqref{eq:o22},    \eqref{eq:o12},  \eqref{eq:o13} and   \eqref{eq:o23}, and $e_1$ appears only in \eqref{eq:o23} it follows that
	$2 d_2 d_3 \lambda=0$ and since $\lambda \neq 0$ we find that $d_2 d_3=0$. If $d_2=0$ then
	$2 d_1 d_3 O(e_1,e_3,\cdot) + d_3^2 O(e_3, e_3,\cdot)=0$. Therefore, using \eqref{eq:o33}, \eqref{eq:o13} and that 
	$c, \lambda  \neq 0$ it follows that $d_3=0$ and thus the vector $d_1 e_1 + d_2 e_2 + d_3 e_3$ is parallel to $v_1=e_1$.
	We assume now that $d_3=0, d_2 \neq 0$.
	Then from \eqref{eq:O0} we find that
	$2 d_1 d_2 O(e_1,e_2,\cdot) + d_2^2 O(e_2, e_2,\cdot)=0$. Thus using \eqref{eq:o12o22} we find that $d_2=c d_1$ and thus the vector $d_1 e_1 + d_2 e_2 + d_3 e_3$ is parallel to $v_2=e_1+c e_2$, as desired.
\end{proof}

We will now prove the following
\begin{lemma}
	Assume that $O$ satisfies the non-degeneracy assumption~\eqref{eq:octopole_assumption}. Then there exists $\delta>0$ such that the following holds:
	If $\|\nabla_{e,p}f(e,p)\| \leq \delta$ and ${\rm \text{Hess}}_{e,p}f(e,p) \geq -\delta$, then $f(e,p)\leq -\delta$. 
\end{lemma}
\begin{proof}
	The proof is similar to that of part b) in Lemma \ref{lem:dipquad}. We will only explain here the modifications. From Lemma \ref{lem:Oee0} we know that there are at most three pairs of vectors
	$\pm v$ such that $O(\pm v, \pm v,\cdot)=0$. As in part b) of Lemma \ref{lem:dipquad} we will prove that if $e$ is close to one of these points and $\delta$ is small enough  then it is impossible to have 
	 $\|\nabla_{e,p}f(e,p)\| \leq \delta$ and ${\rm Hess}_{e,p}f(e,p) \geq -\delta$.
	 
	 We assume that $O(e_1,e_1,\cdot)=0$ let $e(t)$ be defined as in \eqref{eq:et} with $e_0$ replaced by $e_1$ and 
	 $p(t)=\cos(\phi+t\sin(\zeta)) v + \sin(\phi+t\sin(\zeta))O(e_1, e_1^\bot,\cdot)$, where $v \bot O(e_1,e_1^\bot,\cdot)$.
	 Note that $O(e_1,\cdot,\cdot)$ is by assumption nonzero, hence we  may assume without loss of generality that $\|O(e_1,\cdot,\cdot)\|=1$. Since moreover $O(e_1,\cdot,\cdot)$ is traceless it follows that $\|O(e_1,e_1^\bot,\cdot)\|=1$. So then $p(t) \in S^2$ and $(e(t),p(t))$ is a geodesic in $S^2 \times S^2$. Then
	 \begin{align*}
g(t):=f(e(t),p(t))=&3 \sin(2 \theta_t)\sin(\phi_t) +3 \sin^2(\theta_t)O(e_1^\bot, e_1^\bot, p(t))\\
&\qquad - 7 \sin^3(\theta_t) O(e_1^\bot,e_1^\bot,e_1^\bot)(e(t) \cdot p(t))\\
&\qquad -21 \sin^2(\theta_t) \cos(\theta_t) O(e_1,e_1^\bot,e_1^\bot)(e(t) \cdot p(t)),
	 \end{align*}
where $\theta_t=\theta+t\cos(\zeta)$ and $\phi_t=\phi+t\sin(\zeta)$. 
Then we have 
\begin{equation}
g'(0)=6 \cos(2 \theta) \cos(\zeta) \sin(\phi) + O(\epsilon),
\end{equation}
and
\begin{multline}
g''(0)=12 \cos(2 \theta) \sin(\zeta) \cos(\zeta) \cos(\phi) + 6\cos(2\theta)\ \cos(\zeta)^2 O(e_1^\bot,e_1^\bot,p) \\
-42 \cos(2 \theta) \cos(\theta) \cos^2(\zeta)  O(e_1,e_1^\bot,e_1^\bot)(e \cdot p) + O(\epsilon),
\end{multline}
The rest of the argument works like in Lemma \ref{lem:dipquad}.
\end{proof}

Now we will prove the connectedness of the sets $\{(e,p): f(e,p)<-\delta\}$.

\begin{lemma}
	Suppose that $O(\cdot,\cdot,\cdot) \neq 0$. Then there exists $\delta_0>0$ such that for all $\delta<\delta_0$ the set $\{(e,p): f(e,p)<-\delta\}$ is nonempty and pathwise connected.
\end{lemma}

\begin{proof}
	Like in the proof of part (a) of Lemma \ref{lem:dipquad} we may  assume that 
	$$p=-\frac{3 O(e,e,\cdot)-7 O(e,e,e) e }{\|3 O(e,e,\cdot)-7 O(e,e,e) e\|}.$$ 
	We note that this $p$ is well defined on the paths that we will construct. Indeed, from Lemma \ref{lem:Oee0} if the paths avoid the values $e=\pm v_i$ then  $O(e,e,\cdot) \neq 0$ and thus $3 O(e,e,\cdot)-7 O(e,e,e) e \neq 0$, as the inner product of the left hand side with $e$ is nonzero. Thus we will have that
	\begin{equation}
	-f(e,p)=g(e):=\|3 O(e,e,\cdot)-7 O(e,e,e) e\|.
	\end{equation}
	It is therefore enough to show that there exists $\delta_0>0$ such that for all $\delta<\delta_0$ the set 
	$\{e \in S^2: g(e)>\delta\}$ is nonempty and connected. To do this we proceed as follows. In view of Lemma \ref{lem:Oee0} there are at most six points on $S^2$ on which $g$ vanishes. For each such $v_i$ we will show that there exists a neighborhood $U_i \subset S^2$ of $v_i$ with the following property: if $v \in U_i$ but $v \neq v_i$ then there exists a path along which g is increasing until we are out of $U_i$. We will then choose $\delta_0=\min\{g(e): e \in S^2/(\cup U_i) \}$. As we may take the neighborhoods $U_i$ as small as we want, we may also assume that $S^2/(\cup U_i)$ is pathwise connected. Thus if we have two points $e_1,e_2$ for which $g(e_1), g(e_2) \geq \delta$ then we can translate them along an increasing path for $g$ to $S^2/(\cup U_i)$. Since $g \geq \delta_0 > \delta$ on $S^2/(\cup U_i)$ and $S^2/(\cup U_i)$ is pathwise connected, if we construct the neighborhoods $U_i$ the lemma is proven. 
	
	We now consider a point at which $g$ vanishes. After a change of basis we may assume without loss of generality that this point is $e_1$. For any angles $\phi, \theta$ we define $R_{\theta,\phi}= \cos(\theta) e_1 + \sin(\theta) \cos(\phi) e_2 + \sin(\theta) \sin(\phi) e_3$.
	
	\smallskip
	
	\noindent\textbf{Case 1:} Assume that $O(e_1,\cdot,\cdot) \neq 0$.
	 We have in light of the property $O(e_1,e_1,\cdot)=0$ that
	\begin{equation}\label{eq:grthph}
	3O(R_{\theta,\phi},R_{\theta,\phi},\cdot)-7O(R_{\theta,\phi},R_{\theta,\phi},R_{\theta,\phi}) R_{\theta,\phi}=3\sin(2\theta) w + O(\sin^2(\theta)),
	\end{equation} 
 where $w= O(e_1,cos(\phi)e_2 +\sin(\phi) e_3,\cdot) \neq 0$. Indeed, by assumption  the operator $O(e_1,\cdot,\cdot)$ is not zero and  since  it is traceless $O(e_1,\cdot,\cdot)$ it cannot have two linearly independent eigenvectors of eigenvalue zero. But then since $O(e_1,e_1,\cdot)=0$ we find that $w \neq 0$ for any $\phi$. Therefore $\|w\|$ is as a function of $\phi$ bounded away from zero. Thus, from \eqref{eq:grthph} it follows that we can find $\theta_0>0$ such that for all $\theta \in [0, \theta_0)$ the function $g(R_{\theta,\phi}) $ is an increasing function of $\theta$ for any fixed $\phi$. Then the neighborhood we were looking for is just the set $U=\{R_{\theta,\phi}: \theta \in [0,\theta_0) \}$. This completes the proof of the lemma in case 1.

 \smallskip
 
	\noindent\textbf{Case 2:} Assume that $O(e_1,\cdot,\cdot) \equiv 0$. Then we have that
		\begin{equation}
		3O(R_{\theta,\phi},R_{\theta,\phi},\cdot)-7O(R_{\theta,\phi},R_{\theta,\phi},R_{\theta,\phi}) R_{\theta,\phi}=3 \sin^2(\theta) w +O(\sin^3(\theta)),
		\end{equation} 
where $w = O(\cos(\phi)e_2 +\sin(\phi) e_3,\cos(\phi)e_2 +\sin(\phi) e_3,\cdot).$ It suffices then to prove that $w \neq 0$ for all $\phi$ and we can repeat the argument of the proof of case 1. Assume that there is a $\phi$ such that $w=0$. Then without loss of generality we can assume that $\phi=0$. Thus $O(e_2,e_2,\cdot)=0$ and since $O(e_1,e_1,\cdot)=0$ it follows that $O(e_3,e_3,\cdot)=0$. Using that
$O(e_j,e_j,\cdot)=0$, $j=2,3$ and $O(e_1,\cdot,\cdot)=0$ we find that $O(e_2,e_3,\cdot)=0$ and thus $O(\cdot,\cdot,\cdot)=0$ which is a contradiction. Therefore $w \neq 0$ for all $\phi$ and the argument of Case 1 can be repeated to conclude the proof of the lemma.
\end{proof}

\subsection{Quadrupole-quadrupole term}

The situation in the quadrupole-quadrupole term is more complicated than the dipole-multipole terms. We will therefore state and prove some auxiliary lemmas which will be used later to prove Propositions \ref{prop:localmin} and \ref{prop:connected}.

The quadrupole moment of a molecule, is a symmetric matrix $Q$ with vanishing trace. 
The quadrupole-quadrupole term reads (up to a multiplicative positive constant)
\begin{align}
\cF(Q_1, Q_2)&=35 \langle e_1, Q_1 e_1 \rangle  \langle e_1, Q_2 e_1 \rangle- 20  \langle e_1, Q_1 Q_2 e_1 \rangle+ 2 \tr(Q_1 Q_2)\nn\\
&=\tr\big(L(Q_1)Q_2\big),\label{def:cF}
\end{align}
where $p=|e_1\rangle\langle e_1|$ and
\begin{equation*}
L(Q_1)=35pQ_1p-10pQ_1-10Q_1p+2Q_1.
\end{equation*}
In our problem the matrices $Q_1$ and $Q_2$ can only be rotated, which means that 
$$Q_1\in W_1=\{U^T Q_1^0 U,\ U \in SO(3)\},\qquad Q_2\in W_2=\{U^T Q_2^0 U,\ U \in SO(3)\}$$
with fixed symmetric real matrices $Q^0_1$ and $Q_2^0$, such that 
$$\tr(Q_1^0)=\tr(Q_2^0)=0.$$

Before proving the propositions, we will first state and prove some Lemmas that will be useful for the proof. 

\begin{lemma}\label{lem:immerspur0}
	Let $A,B \in \BR^{3 \times 3} $ be two symmetric matrices with $\tr(B)=0$ but $B \neq 0$. If
	\begin{equation}\label{immerspur0}
	\tr(A U^T B U)=0, \quad  \forall U \in SO(3),
	\end{equation}
	then there exists
	$\lambda \in \R$ so that $A=\lambda I$.
\end{lemma}
\begin{proof}
	Since the Lie algebra of $SO(3)$ is the set of antisymmetric matrices $so(3)$, from \eqref{immerspur0} we obtain that
	$$\tr(A [w, U^T BU])=0, \quad \forall U \in SO(3), \ \forall w \in so(3)$$
	or that
	$$\tr([U^T BU,A] w)=0, \quad \forall U \in SO(3),\ \forall w \in so(3).$$
Since $[U^T BU,A]$ is an antisymmetric matrix, it follows that  
	\begin{equation}\label{UTBU}
	[U^T BU,A]=0, \quad \forall U \in SO(3).
	\end{equation}
	Suppose that $A$ is not multiple of the identity. Then $A$ has an eigenvalue of multiplicity 1 and let $\vec{v}$ be a corresponding eigenvector. From $\eqref{UTBU}$ it follows that $\vec{v}$ is an eigenvector of $U^T BU$
	for all $U \in SO(3)$ or equivalently $U\vec{v}$ is an eigenvector of $B$ for all $U \in SO(3)$. This implies that $B$ is a multiple of the identity contradicting the fact that $B$ is nonzero and traceless.
	Thus $A$ must be a multiple of the identity.
\end{proof}

Later we will use Lemma \ref{lem:immerspur0} for the functional $\cF(Q_1,Q_2)$ with $B=Q_2$ and $A=35 p Q_1 p-10 p Q_1-10 Q_1 p +2 Q_1$. To this end it is useful to prove that $35 p Q_1 p-10 p Q_1-10 Q_1 p+2 Q_1$ is never a multiple of the identity.
\begin{lemma}\label{lem:niemalslambdai}
	Let $A=35 p Q_1 p-10 p Q_1-10 Q_1 p+2 Q_1$, where $Q_1=Q_1^T$ and $\tr(Q_1)=0$ and assume that there exists $\lambda \in \R$ so that $A=\lambda I$. Then $Q_1=0$. 
\end{lemma}
\begin{proof}
	From $A e_1=\lambda e_1$ it easily follows that $e_1$ is an eigenvector of $Q_1$ and that the eigenvalue is $\frac{\lambda}{17}$. As a consequence $A=\frac{15}{17}\lambda p+ 2 Q_1= \lambda I$.
	It follows immediately that  $Q_1 e_j = \frac{\lambda}{2} e_j$, for $j=2,3$ so that $e_2, e_3$ are also eigenvectors of $Q_1$. Using that $\tr(Q_1)=0$ we obtain that $\lambda=0$ and therefore $Q_1=0$. 
\end{proof}

\begin{lemma}\label{lem:mittelwert0}
	For all $Q_1 \in W_1, Q_2 \in W_2$ we have that
	\begin{equation}
	\int_{SO(3)}\cF(Q_1, U^T Q_2 U) dU=0,
	\end{equation}
	where the integral is with respect to the Haar measure.
\end{lemma}

\begin{proof}
For any symmetric matrix, we have
	$$\int_{SO(3)}U^T A U dU=\frac{\tr(A)}{3}.$$
Indeed, after changing variables $U'=UV$,  we see that the left side commutes with all $V\in SO(3)$, hence it must be a multiple of the identity. Taking the trace gives the claimed relation. Thus we obtain 
	$$\int_{SO(3)}\cF( Q_1 , U^T Q_2 U) dU=\cF\left(Q_1, \int_{SO(3)}U^T Q_2 U dU\right)=0$$
since $\tr(Q_2)=0$. 
\end{proof}

\begin{lemma}\label{Satz:nuetzlich}
\noindent$(i)$ There exists $c_0>0$ so that 
	$$g(Q_1):=\min_{Q_2 \in W_2} \cF(Q_1, Q_2)<-c_0$$ 
	for all $Q_1 \in W_1$, and 
	$$h(Q_1):=\max_{Q_2 \in W_2} \cF(Q_1, Q_2)>c_0$$ 
	for all $Q_1 \in W_1$.
	
\medskip
	
\noindent $(ii)$  For any $Q_1$ in $W_1$ and any critical point $Q_2$ of $\cF_{Q_1}:=\cF(Q_1,\cdot)|_{W_2}$ on which  $\cF_{Q_1}(Q_2)$ is nonnegative or not a minimum, ${\rm Hess}_{Q_2} \cF_{Q_1}(Q_2)$ has at least one negative eigenvalue.

\medskip
	
\noindent $(iii)$ There exists  $\epsilon>0$ such that for all $Q_1 \in W_1$ and $Q_2 \in W_2$, if $\|\nabla_{Q_2} \cF_{Q_1}(Q_2)\| \leq \epsilon$ and ${\rm Hess}_{Q_2} \cF_{Q_1}(Q_2) \geq -\epsilon$ then $\cF_{Q_1}(Q_2)<-\epsilon$.

\medskip

\noindent $(iv)$ Let $c_0$ be as in part (i). For all $\delta<c_0$ if $\cF_{Q_1}(Q_2)<-\delta$ then there exists a continuous path $\gamma: [0,1] \rightarrow W_2$ such that $\gamma(0)=Q_2$, $\cF(Q_1,\gamma(t))<-\delta$ for all $t \in [0,1]$ and  $\cF(Q_1,\gamma(1))=g(Q_1)$.
In other words $Q_2$ can be connected to a minimizer of $\cF_{Q_1}$ with a path along which $\cF$ stays smaller than $-\delta$.
\end{lemma}

\begin{proof}
	$(i)$ As there exists a constant $C>0$ so that
	\begin{equation}\label{est:Fschanke}
	|\cF(B,D)| \leq C \| B\|  \| D\|, \quad \forall B,D \in \R^{3 \times 3},
	\end{equation}
	where $\|.\|$ is the standard operator norm with respect to the Euclidean norm in $\R^3$, it follows that $g$ is a continuous function on $W_1$. Therefore, it is enough to show that $g(Q_1)<0$ for all $Q_1 \in W_1$. Let $Q_1 \in W_1$. Applying Lemma \ref{lem:immerspur0} for $B=Q_2$ and $A=35 p Q_1 p-10 p Q_1-10 Q_1 p +2 Q_1$  and Lemma \ref{lem:niemalslambdai}, we obtain that  there exists $Q_2 \in W_2$
	so that $\cF(Q_1,Q_2) \neq 0$. If $\cF(Q_1,Q_2)<0$, then $g(Q_1)<0$. If $\cF(Q_1,Q_2)>0$ then from the continuity of $\cF$ and from Lemma \ref{lem:mittelwert0}, it follows that there exists a $Q_2' \in W_2$ so that $\cF(Q_1,Q_2')<0$. Therefore, $g(Q_1)<0$ in any case. The function $h$ can be handled similarly. 
	
	\medskip
	
	$(ii)$ For fixed $Q_1$, we have $\cF_{Q_1} (Q_2)=\tr(L(Q_1)Q_2)$. By arguing  as in the proof of Equation \eqref{UTBU} it follows that
	a critical point  of $\cF_{Q_1}|_{W_2}$
	commutes with $L(Q_1)$, and as a consequence so does a minimizer $Q_{2,min} \in W_2$. Therefore,  there exists an orthonormal basis $\vec{v_1}, \vec{v_2}, \vec{v_3}$ of $\R^3$ such that  $\vec{v_1}, \vec{v_2}, \vec{v_3}$ are eigenvectors of both $L(Q_1)$ and $Q_{2,min}$.
	Let $a_1, a_2, a_3$ be the eigenvalues of $L(Q_1)$ so that $L(Q_1) \vec{v_i}= a_i \vec{v_i}$, $i=1,2,3$ and let $b_1,b_2,b_3$ be the eigenvalues of $Q_2^0$ and therefore of $Q_2$.  
	Without loss of generality we assume that $a_1 \leq a_2 \leq a_3$ and $b_1 \geq b_2 \geq b_3$.
	Then 
	$$\tr(L(Q_1) Q_{2,min})=a_1 b_1+ a_2 b_2 + a_3 b_3,$$ 
	i.e. the trace is minimized if $Q_2 \vec{v_i}= b_i \vec{v_i}$. For all other critical points we have $Q_2 \vec{v_i}= b_{\sigma(i)} \vec{v_i}$
	$i=1,2,3$, where $\sigma \in \gS_3$ is a permutation.
	Then unless $Q_2$ is another minimizer, 
	it is not a local minimum. To see this
	we note that 
	\begin{equation}
	\cF_{Q_1}(Q_2)=A_{\sigma(1) \sigma(2) \sigma(3)},
	\end{equation}
       where
	\begin{equation}
	A_{ijk}:=a_1 b_i + a_2 b_j + a_3 b_k.
	\end{equation}
	 If $Q_2$ is not a minimizer, then there exist
	$i,j$ with $a_i<a_j$ and $b_{\sigma(i)}< b_{\sigma(j)}$. We consider the rotation $U_\theta$ for which 
	$U_\theta(\vec{v_i})= \cos(\theta) \vec{v_i}+ \sin(\theta) \vec{v_j}$, 
	$U_\theta(\vec{v_j})= \cos(\theta) \vec{v_i}- \sin(\theta) \vec{v_j}$
	and $U_\theta$ leaves the orthogonal complement of $\text{span}\{\vec{v_i}, \vec{v_j}\}$ invariant. Let 
	$f(\theta):=\tr(L(Q_1) U_\theta Q_2 U_\theta^T)$. A computation shows that
	\begin{equation}\label{eq:decpath}
	f(\theta)-f(0)=-(A_{\sigma(1) \sigma(2) \sigma(3)}- A_{\pi\sigma(1) \pi\sigma(2) \pi\sigma(3)}) \sin^2(\theta) \leq 0,
	\end{equation}
	where $\pi$ is the permutation exchanging $\sigma(i), \sigma(j)$.
	So
	$$f''(0)=-2(A_{\sigma(1) \sigma(2) \sigma(3)}- A_{\pi\sigma(1) \pi\sigma(2) \pi\sigma(3)}) <0$$ 
	and hence the critical point has a Hessian with at least one negative eigenvalue and thus it is not a local minimum. 
	Since, due to part (i), critical points of $\cF_{Q_1}$ on which $\cF_{Q_1} \geq 0$ are not minimizers, it follows immediately that the Hessian of $\cF_{Q_1}$ at $Q_2$ has at least one negative eigenvalue. In fact its lowest eigenvalue is given by
	\begin{multline}\label{eq:lambdaminHess}
\lambda_{min,Hess}=-2 \max\big\{(A_{\sigma(1) \sigma(2) \sigma(3)}- A_{\pi\sigma(1) \pi\sigma(2) \pi\sigma(3)}),\\
\pi\ \text{permutation of two indices}\big\}.
	\end{multline} 
	
\medskip

$(iii)$	The main difficulty of the proof arises from the fact that $\epsilon$ has to be independent of $Q_1$. We are first going to find an appropriate $\epsilon_0$ for all $Q_1, Q_2$ satisfying the additional assumption $\nabla_{Q_2} \cF_{Q_1}(Q_2)=0$, i.e. $Q_2$ is a critical point of $\cF_{Q_1}$. For this $Q_1, Q_2$ the matrices $L(Q_1), Q_2$ are simultaneously diagonalizable, as we already mentioned in the proof of part $(ii)$ of Lemma \ref{Satz:nuetzlich}. Let $a_1 \leq a_2 \leq a_3$ be the eigenvalues of $L(Q_1)$ and $b_1\geq b_2 \geq b_3$ be the eigenvalues of $Q_2$. Depending on the orientation of $Q_2$ the possible values of
$\cF_{Q_1}(Q_2)$ are $A_{123}$, $A_{132}$, $A_{213}$, $A_{231}$, $A_{321}$, $A_{312}$.
The values $A_{123}, A_{321}$ are the smallest and biggest values respectively and because of part $(i)$ of Lemma \ref{Satz:nuetzlich} we have that 
\begin{equation}\label{eq:A123}
A_{123}<-c_0,
\end{equation}
and $A_{321} > c_0$ where it is important that $c_0$ does not depend on $Q_1$. Observing in addition that $\tr(Q_2)=0$ implies that $A_{321} + A_{213} + A_{132}=0$, it follows that
\begin{equation}\label{eq:A213}
\min\{A_{213},A_{132}\}<-\frac{c_0}{2}.
\end{equation}
But for every permutation  $\sigma \in \gS_3$ there exists a permutation $\pi \in \gS_3$ exchanging two elements so that $\pi \sigma \in \{(123),(213)\}$.
The set  $\{(123),(132)\}$ has the same property.
In the setting of our problem, in light of equations \eqref{eq:A123} and \eqref{eq:A213}, this means that independently on which critical point we study, we can rotate around one axis 90 degrees so that we arrive at the critical point with value less than $-\frac{c_0}{2}$. Thus, using~\eqref{eq:lambdaminHess} it follows that the lowest eigenvalue of the Hessian or the value of the function has to be less than $-{c_0}/{4}$, i.e.
\begin{equation}\label{eq:fHesscritical}
\nabla_{Q_2}\cF(Q_1, Q_2)=0 \implies \min\{\cF(Q_1,Q_2), \lambda_{min, {\rm Hess}_{Q_2}\cF(Q_1,Q_2)}\} \leq -\frac{c_0}{4}.
\end{equation}

To remove the assumption $\nabla_{Q_2} \cF_{Q_1}(Q_2)=0$ and therefore finish the proof, observe that the functions $\cF$, $\nabla_{Q_2}\cF$ and ${\rm Hess}_{Q_2} \cF$ are uniformly continuous. Thus there is a $\delta>0$ so that if $(Q_1,Q_2)$ has distance less than $\delta$ from a critical point $(\widetilde{Q_1},\widetilde{Q_2})$, then $$|\cF(Q_1,Q_2)-\cF(\widetilde{Q_1},\widetilde{Q_2})| \leq \frac{c_0}{8}$$
and
$$| {\rm Hess}_{Q_2}\cF(Q_1,Q_2)- {\rm Hess}_{Q_2}\cF(\widetilde{Q_1},\widetilde{Q_2})| \leq \frac{c_0}{8},$$
which, together with \eqref{eq:fHesscritical}, implies that
\begin{equation}\label{eq:fHessfastcritical}
\min\{\cF(\widetilde{Q_1},\widetilde{Q_2}),  \lambda_{min, {\rm Hess}_{Q_2}\cF(\widetilde{Q_1},\widetilde{Q_2})}\} \leq -\frac{c_0}{8}.
\end{equation}
For the points away from a $\delta$ neighborhood of the critical points, the norm $\|\nabla_{Q_2}\cF\|$ is bounded away from zero uniformly in $Q_1$. This together with \eqref{eq:fHessfastcritical} concludes the proof of part $(iii)$.

\medskip
	
	$(iv)$ We consider a $\delta$  as in the statement of the lemma and an arbitrary $Q_2$ with the property that $\cF_{Q_1}(Q_2)<-\delta$.  We will show that there exists a path connecting $Q_2$ with a minimizer $Q_{2,min}$ along which the functional $\cF_{Q_1}$ remains smaller or equal than $\cF_{Q_1}(Q_2)$.
	
	We first consider the case that  $Q_2$ is a critical  point of $\cF_{Q_1}$ namely $[L(Q_1), Q_2]=0$. Then, as we showed in part $(ii)$ (unless $Q_2$ is a minimizer in which case there is nothing to prove), we can apply the argument with the rotations of the previous paragraph until we reach a minimizer along a decreasing path. The function $f$ in \eqref{eq:decpath} is decreasing in $[0,\frac{\pi}{2}]$, until the eigenvectors to the eigenvalues $b_{\sigma(i)}, b_{\sigma(j)}$ are exchanged.
	
	We now consider the case, that $Q_2$ is not a critical point of $\cF_{Q_1}$. Then,
	\begin{equation}\label{eq:notcritical}
	[L(Q_1), Q_2] \neq 0.
	\end{equation} 
	But this means that not all directional derivatives are zero so there is a direction of strict increase. As a consequence we can apply Lemma \ref{lem:flowpseudolocmin} to prove that we can reach a critical point $Q_{2,min}$ of $\cF_{Q_1}$, for which the Hessian is nonnegative. It follows from  part $(ii)$ that $Q_{2,min}$ must be a global minimizer of $\cF_{Q_1}$.
\end{proof}

Our proof does \emph{not} show that $\{Q_2 \in W_2: \cF(Q_1,Q_2)<-\delta\}$ is connected for $\delta<c_0$. The reason is that even in case of non-degeneracy of the eigenvalues $a_1, a_2, a_3, b_1, b_2, b_3$ coming up in the proof, the minimizer of $\cF_{Q_1}$ is not unique. This is because $\cF_{Q_1}$ is invariant under rotations of 180 degrees around the axes of the eigenvectors, and therefore in the case of non-degeneracy there are already four minimizers. It is not clear if it is possible to connect any two of them with a path along which $\cF$ stays negative. 
To deal with this problem we will prove that this is possible for some $Q_1$.

Since $Q_2$ is traceless we have that
\begin{equation}
\cF(Q_1,Q_2)=\tr(M(Q_1) Q_2)
\end{equation}
where the operator 
$$M(Q_1):=L(Q_1)-5\langle e_1, Q_1 e_1\rangle=L(Q_1)-\tr(L(Q_1))$$  is also traceless. The following lemma gives a sufficient condition for the set
$\{Q_2 \in W_2: \cF_{Q_1}(Q_2)<-\delta\}$ to be connected when $\delta>0$ is small enough.

\begin{lemma}\label{lem:eigsigns}
	$(i)$
	Suppose that $M(Q_1)$ has two positive eigenvalues and $Q_2^0$ has two nonnegative eigenvalues. Then there exists $\delta_0>0$ so that for all $\delta<\delta_0$ the set $\{Q_2 \in W_2: \cF_{Q_1}(Q_2)<-\delta\}$ is pathwise connected.
	
	\medskip
	
\noindent $(ii)$ The same conclusion holds if $M(Q_1)$ has two negative eigenvalues and $Q_2^0$ has two nonpositive eigenvalues.
\end{lemma}

\begin{proof}
	Part $(ii)$ follows from $(i)$ and the observation that  $\cF$ remains invariant if we replace $Q_1, Q_2$ with $-Q_1,-Q_2$. So we only prove $(i)$. Let $\lambda_1, \lambda_2,-(\lambda_1+\lambda_2)$, with $\lambda_1 \geq \lambda_2 > 0$, be  the eigenvalues of $M(Q_1)$ and $\mu_1, \mu_2,-(\mu_1+\mu_2)$, with $\mu_1 \geq \mu_2 \geq 0$, be the eigenvalues of $Q_2$. 

	\medskip
	
\noindent\textbf{Case 1:} There is no degeneracy of eigenvalues, namely $\lambda_1 \neq \lambda_2$ and $\mu_1 \neq \mu_2$. Then as we explained above there are four minimizers of $\cF_{Q_1}$ due to the invariance of $\cF_{Q_1}$ under  rotations of 180 degrees around the axis of any eigenvector. The corresponding minimum is then the sum
	$$-(\lambda_1 + \lambda_2 ) \mu_1 + \lambda_2 \mu_2 -\lambda_1 (\mu_1 + \mu_2)  $$
If we exchange $\mu_1$ and $\mu_2$ or if we exchange $\lambda_1$ with $\lambda_2$ the sum still remains negative.
It follows that if we rotate $Q_2$  around the eigenvector of the eigenvalue $\lambda_1$ or $-\lambda_1-\lambda_2$ then $\cF_{Q_1}$ remains negative all along the rotation. But with these rotations by 180 degrees and their combination we can connect all the minimizers of $\cF_{Q_1}$ with paths along which $\cF_{Q_1}$ remains negative and therefore remains less than some number $-\delta'<0$. Using this and part $(iii)$ of Proposition \ref{Satz:nuetzlich} it follows that for all $\delta<\delta_0$, where 
	\begin{equation}\label{def:d0}
	\delta_0:=\min\left\{\delta', \frac{c_0}{2}\right\}
	\end{equation}
	the set  $\{Q_2 \in W_2: \cF(Q_1,Q_2)<-\delta\}$ is pathwise connected.

		\medskip
	
\noindent\textbf{Case 2:} There is a degeneracy of eigenvalues, so $\lambda_1 = \lambda_2$ or $\mu_1 = \mu_2$ or both equalities hold. Then the argument above works with no modifications. The rotations described above work and under one or both of them the sets of minimizers is invariant.
\end{proof}

Since $Q_2$ is traceless it always has two nonnegative or two nonpositive eigenvalues.
In order to successfully  use Lemma \ref{lem:eigsigns} we will show  that there is always a suitable choice of $Q_1$ for which $M(Q_1)$ has two positive eigenvalues and another choice for which $M(Q_1)$ has two negative eigenvalues.

\begin{lemma}\label{lem:tptn}
	There is always a choice of $Q_1\in W_1$ for which $M(Q_1)$ has two positive eigenvalues (counting multiplicity) and  always a choice of $Q_1\in W_1$ for which $M(Q_1)$ has two negative eigenvalues.
\end{lemma}

\begin{proof}
	We write $Q_1=c_1 P_{v_1} + c_2 P_{v_2} + c_3 P_{v_3}$, where $c_1 \geq c_2 \geq c_3$ are the eigenvalues of $Q_1$ and $P_{v_i}$ is the orthogonal projection onto the eigenvector $v_i$. Then $c_1 > 0 > c_3$. If we choose $Q_1$ so that $p=P_{v_1}$, then 
	$M(Q_1)=12 c_1 P_{v_1} + (2 c_2 -5 c_1) P_{v_2} + (2 c_3- 5 c_1) P_{v_3}$ so that $M(Q_1)$ has one positive and two negative eigenvalues. Similarly, if we choose $Q_1$ so that $p=P_{v_3}$ then $M(Q_1)$ has one negative and two positive eigenvalues.
\end{proof}

We note that Proposition \ref{prop:localmin} is an immediate corollary of part $(iii)$ of Lemma \ref{Satz:nuetzlich}. Thus it remains to provide the

\begin{proof}[Proof of Proposition \ref{prop:connected} for the quadrupole-quadrupole term]
The traceless matrix $Q_2$ has always two nonnegative eigenvalues or two nonpositive eigenvalues. Applying Lemma \ref{lem:tptn} we choose $\widetilde{Q_1}$ so that $M(\widetilde{Q_1})$ has two positive or, respectively, two negative eigenvalues. Then from Lemma \ref{lem:eigsigns} it follows that for all $\delta \in (0, \delta_0)$, where $\delta_0$ is the same as in \eqref{def:d0}, the set
$\{Q_2 \in W_2: \cF(\widetilde{Q_1},Q_2)<-\delta\}$ is pathwise connected.

We fix $\delta>0$ with $\delta<\delta_0$ and consider $(Q_{1,0}, Q_{2,0}),  (Q_{1,1}, Q_{2,1}) \in W_1 \times W_2$ with $\cF(Q_{1,0}, Q_{2,0})<-\delta$ and  $\cF(Q_{1,1}, Q_{2,1})<-\delta$. We would like to connect these points through a path
in $W_1 \times W_2$ along which $\cF<-\delta$. 
Due to the connectedness of $\{Q_2 \in W_2: \cF(\widetilde{Q_1},Q_2)<-\delta\}$ it is enough to connect $(Q_{1,0}, Q_{2,0}) $ to some point $(\widetilde{Q_1},Q_2')$
for some $Q_2' \in W_2$ with a path along which $\cF<-\delta$ and do in a similar way the same for the point $(Q_{1,1}, Q_{2,1})$. We will show how to do this for the point  $(Q_{1,0}, Q_{2,0})$.

We first consider a differentiable function $Q: [0,1] \rightarrow W_1$, with $Q(0)=Q_{1,0}$, $Q(1)=\widetilde{Q_1}$ and $\|Q'(t)\|=\ell$ for all $t \in [0,1]$. Here $\ell>0$ is the length of the path connecting
$Q_{1,0}$ and $\widetilde{Q_1}$. We start with the construction of the path in $W_1 \times W_2$ as follows. By part $(iii)$ of Proposition \ref{Satz:nuetzlich}, we can connect $Q_{2,0}$ to a minimizer $Q_{2,0}'$ of $\cF_{Q_{1,0}}$ with a path along which $\cF_{Q_{1,0}}<-\delta$.
Then $\cF(Q_{1,0},Q_{2,0}')=g(Q_{1,0})$. But then by part $(i)$ of Proposition \ref{Satz:nuetzlich} we have that $\cF(Q_{1,0}, Q_{2,0}')<-c_0$. Now we keep $Q_{2,0}'$ fixed and we change $Q_{1,0}$ along the path $Q(t)$ until the time $t_1$, where 
$$t_1=\sup\left\{t \in [0,1]: \cF(Q(s),Q_{2,0}') < -\frac{c_0}{2},\quad \forall s \leq t\right\}.$$
If $t_1=1$ the conclusion follows. If not then we have that 
$$\cF(Q(t_1), Q_{2,0}')=-\frac{c_0}{2}<-\delta.$$
Therefore, 
\begin{equation}\label{est:bigdistance}
\left|\cF(Q(t_1), Q_{2,0}')-\cF(Q_{ 1,0}, Q_{2,0}')\right| \geq \frac{c_0}{2}.
\end{equation}
Moreover using that $\|Q'(t)\|=\ell$ for all $t \in [0,t_1]$ and that $Q(0)=Q_{1,0}$, we obtain that
$\|Q(t_1)-Q_{1,0}\| \leq \ell t_1$. This together with \eqref{est:bigdistance}, the bilinearity of $\cF$ and \eqref{est:Fschanke}   gives that
$t_1 \geq \frac{c_0}{2C \ell \|Q_2^0 \|}$, where $Q_2^0$ is the same as in the definition of $W_2$. If we iterate this procedure then each time we move forward a time step which is at least equal to $\frac{c_0}{2C \ell \|Q_2^0\|}$ along the path $Q$. Therefore, after repeating the procedure above finitely many times, we can construct a path in $W_1 \times W_2$ along which $\cF<-\delta$, and connecting $(Q_{1,0}, Q_{2,0})$ with $(\widetilde{Q_1}, Q_{2}')$, where $Q_{2}' \in W_1$ is such that   $\cF(\widetilde{Q_1}, Q_{2}')<-\delta$. 
This concludes the proof of Proposition \ref{prop:connected} for the quadrupole-quadrupole term.
\end{proof}

\appendix

\section{Dressing paths with wavefunctions}\label{app:dressed_path}

In this appendix we state and prove a result from~\cite{Lewin-04b,Lewin-PhD,Lewin-06}, in a more general setting and with a slightly different approach. We show that it is always possible to choose a continuous family of wavefunctions $\Psi(t)$ and have almost the same maximum value of the energy along any given path. 

\begin{lemma}[Dressing paths with wavefunctions]\label{lem:dressing_paths}
Let $(H(t))_{t\in[0,1]}$ be a one-parameter family of operators on a Hilbert space $\gH$, which are assumed to be
\begin{itemize}
 \item bounded from below, uniformly in $t\in[0,1]$;
 \item of constant operator domain, that is, 
 $$C^{-1}\big(H(0)^2+1\big)\leq H(t)^2+1\leq C\big(H(0)^2+1\big)$$
 for all $t\in[0,1]$;
 \item continuous in the sense that 
 $$t\mapsto \big(C+H(0)\big)^{-\frac12}\,H(t)\,\big(C+H(0)\big)^{-\frac12}$$
 is norm-continuous on $[0,1]$;
 \item such that 
 $$E(t):=\min\sigma\big(H(t)\big)<\min\sigma_{\rm ess}\big(H(t)\big)$$
 for all $t\in[0,1]$.
\end{itemize}
Let $x_0$ and $x_1$ be any normalized eigenvectors of $H(0)$ and $H(1)$ associated with the first eigenvalue $E(0)$ and $E(1)$, respectively. Then for any $\eps>0$ there exists a continuous map $t\mapsto x(t)$ in $D(H(0))$ of normalized vectors such that $x(0)=x_0$, $x(1)=x_1$ and 
$$\max_{t\in[0,1]}\pscal{x(t),H(t)x(t)}\leq \max_{t\in[0,1]}E(t)+\eps.$$
\end{lemma}

The exact same result holds under the weaker assumption that the operators have a constant form domain, but then $t\mapsto x(t)$ can be chosen continuous only in $Q(H(0))$.

\begin{figure}[t]
\centering
\includegraphics[width=7cm]{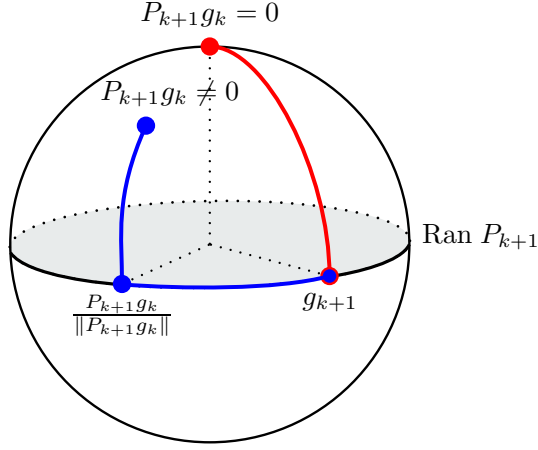}
\caption{The two cases in the proof of Lemma~\ref{lem:dressing_paths}.\label{fig:sphere}}
\end{figure}

\begin{proof}
For $N$ a large integer, we let $t_k:=k/N$ for $k=0,\dots ,N$ and choose any ground state $x_k$ of $H(k/n)$, for $k=1,\dots ,N-1$. Then we define a continuous path $x(t)$ such that $x(t_k)=x_k$ for all $k=0,\dots ,N$. 

On $[t_k,t_{k+1}]$ we define the path $x(t)$ as follows. Let $P_{k+1}$ denote the orthogonal projection onto the first eigenspace of $H((k+1)/N)$.  If $P_{k+1}x_{k}=0$, we take
$$x(t)=\cos\left(\frac{\pi}2(N t-k)\right)x_{k}+\sin\left(\frac{\pi}2(N t-k)\right) x_{k+1}$$
for $t_k\leq t\leq t_{k+1}$. 
If $P_{k+1}x_{k}\neq 0$, we construct a path in two parts as displayed in Figure~\ref{fig:sphere}. We first use
\begin{multline}
x(t)=\cos\left(\alpha_{k}(2k+1-2Nt)\right)\frac{P_{k+1}x_{k}}{\|P_{k+1}x_{k}\|}\\+\sin\left(\alpha_{k}(2k+1-2Nt)\right) \frac{(1-P_{k+1})x_{k}}{\|(1-P_{k+1})x_{k}\|} 
\label{eq:path_matrix_GS}
\end{multline}
for $t_k\leq t\leq t_k+1/(2N)$, with
$$\alpha_k=\arccos\|P_{k+1}x_{k}\|$$
where it is understood that the second term in~\eqref{eq:path_matrix_GS} vanishes when $(1-P_{k+1})x_{k}=0$. In the second part $t_k+1/(2N)\leq t\leq t_{k+1}$ we link the two vectors $P_{k+1}x_{k}\|P_{k+1}x_{k}\|^{-1}$ and $x_{k+1}$ by any path within the unit sphere of the first eigenspace of $H((k+1)/N)$. In particular if the two vectors are colinear, this amounts to just adding the appropriate complex phase. 

The above paths have been chosen to ensure that 
$$t\mapsto \pscal{x(t),H\left(\frac{k+1}{N}\right)x(t)}$$
is non-increasing for $t\in [t_k,t_{k+1}]$. Under our assumptions on $H(t)$ and since 
$$H\left(\frac{k}{N}\right)x_k=E\left(\frac{k}{N}\right)x_k$$
for all $k$, we have that $x_k$ is bounded in $D(H(0))$, by a constant independent of $N$. Moreover, 
$$t\mapsto \pscal{x(t),\left(C+H\left(\frac{k+1}{N}\right)\right)^2x(t)}$$
is also non-decreasing for $C$ large enough such that $(C+E((k+1)/N)$ is the first eigenvalue of $(C+H((k+1)/N))^2$, which proves that the map $x(t)$ stays uniformly bounded in $D(H(0))$ (hence in $Q(H(0))$) for all $t\in[0,1]$. 

Now we have for $t\in [t_k;t_{k+1}]$
\begin{align*}
&\pscal{x(t),H(t)x(t)}\\
&\qquad=\pscal{x(t),H\left(\frac{k+1}{N}\right)x(t)}+\pscal{x(t),\left(H(t)-H\left(\frac{k+1}{N}\right)\right)x(t)}\\
&\qquad\leq\pscal{x_{k},H\left(\frac{k+1}{N}\right)x_{k}}+\pscal{x(t),\left(H(t)-H\left(\frac{k+1}{N}\right)\right)x(t)}\\
&\qquad\leq E\left(\frac{k}{N}\right)+\pscal{x_{k},\left(H\left(\frac{k+1}{N}\right)-H\left(\frac{k}{N}\right)\right)x_{k}}\\
&\qquad\qquad +\pscal{x(t),\left(H(t)-H\left(\frac{k+1}{N}\right)\right)x(t)}\\
&\qquad\leq E\left(\frac{k}{N}\right)+\max_{t\in[0,1]}\norm{(H(0)+C)^{\frac12}x(t)}^2\times\\
&\qquad\quad \times\max_{\frac{k}{N}\leq t\leq\frac{k+1}{N}}\norm{(H(0)+C)^{-\frac12}\left(H(t)-H\left(\frac{k+1}{N}\right)\right)(H(0)+C)^{-\frac12}}.
\end{align*}
This shows that the maximum along the path is
\begin{multline*}
\max_{t\in[0,1]}\pscal{x(t),H(t)x(t)}\leq \max_{t\in[0,1]}E(t)\\+C\max_{|t-s|\leq 1/N}\norm{(H_0+C)^{-\frac12}\left(H(t)-H(s)\right)(H_0+C)^{-\frac12}}
\end{multline*}
where the second term is small for $N$ large. 
\end{proof}

\section{Multipolar expansion: proof of Lemmas~\ref{lem:multipexp} and~\ref{lem:multipolar_expansion}}\label{app:multipolar_expansion}

\subsection{Proof  of Lemma \ref{lem:multipexp}}
We have that
\begin{equation*}
\frac{1}{|Le_1-h|}=\frac{1}{L|e_1-u|}=\frac{1}{L\sqrt{1+|u|^2-2u\cdot e_1}}\qquad\text{with $u:=\frac{h}{L}$.}
\end{equation*}
Both $x\mapsto(1+x)^{-1/2}$ and $u\mapsto |u|^2-2u\cdot e_1$ are real analytic in a neighborhood of the origin, and so is $u\mapsto |e_1-u|^{-1}$ on the ball $|u|\leq1/2$. We conclude that 
	\begin{equation*}
	\bigg|\partial_u^\alpha \bigg( \frac{1}{|e_1-u|}-
	\sum_{n=0}^N \cM^{(n)}_{\delta_u}(e_1,\dots ,e_1)  \bigg) \bigg|\leq C |u|^{N+1-|\alpha|}
	\end{equation*}
for a constant $C$ depending on $N\geq1$ and $|\alpha|\leq N$, but not on $u$ for $|u|\leq1/2$.
Replacing $u$ by $h/L$ we find
	\begin{equation*}
	\bigg|\partial_h^\alpha \bigg( \frac{1}{|Le_1-h|}-
	\sum_{n=0}^N \frac{\cM^{(n)}_{\delta_h}(e_1,\dots ,e_1) }{L^{n+1}} \bigg) \bigg|\leq C \frac{|h|^{N+1-|\alpha|}}{L^{N+2}}\leq C\frac{1+|h|^{N+1}}{L^{N+2}}
	\end{equation*}
as desired.\qed

\subsection{Proof of Lemma~\ref{lem:multipolar_expansion}}

This is a well known result, see for example ~\cite{BurBon-81}. The Taylor expansion of the Coulomb potential in Lemma \ref{lem:multipexp}  gives
\begin{multline*}
\bigg\|\iint_{\R^3\times\R^3}\frac{{\rm d}\rho_1(x)\;{\rm d}\rho_2(y)}{|Ux-Vy+Le_1|}-
\sum_{k=0}^N B^{(k)}(U,V) \bigg\|_{C^2(SO(3)^2)}\\
\leq \frac{C}{L^{N+2}}\int_{\R^3}(1+|x|^{N+1})\,{\rm d}|\rho_1|(x)\; \int_{\R^3}(1+|y|^{N+1})\,{\rm d}|\rho_2|(y).
\end{multline*}
where
\begin{multline*}
B^{(k)}(U,V)\\=
\frac{1}{k!}
\iint_{\R^3\times\R^3}{\rm d}\rho_1(x)\;{\rm d}\rho_2(y) (U x_{a_1}-V y_{a_1})\dots (U x_{a_k}-V y_{a_k})\partial_{a_1}\dots \partial_{a_k}\frac{1}{|Le_1|}
\end{multline*}
and we have used the Einstein convention for summation in $a_j \in \{1,2,3\}$. By abuse of notation $\partial_{a_1}\dots \partial_{a_N}\frac{1}{|Le_1|}$ means partial derivatives of the Coulomb potential evaluated at $Le_1$.
It thus suffices to prove that
\begin{equation}\label{eq:BUVdec}
B^{(k)}(U,V)=\frac{1}{L^{k+1}}\sum_{m=1}^{k-1} \cF^{(m,k-m)}(U,V).
\end{equation}
We assume without loss of generality that $U=V=I$.
Expanding the product $\prod (x_{a_j}-y_{a_j})$ and using Schwarz theorem, due to which the order of the partial derivatives does not matter, we can rewrite
\begin{multline}
B^{(k)}(U,V)= \sum_{n+m=k} (-1)^m
\left(\frac{1}{n!} \int_{\R^3}{\rm d}\rho_1(x) x_{a_1}\dots x_{a_n}\right)\times\\
\times\left(\frac{1}{m!}\int_{\R^3} {\rm d}\rho_2(y) y_{b_1}\dots y_{b_m}\right)\partial_{a_1}\dots \partial_{a_n}\partial_{b_1}\dots \partial_{b_m}\frac{1}{|Le_1|},
\label{eq:pre_multipole_expansion}
\end{multline}
because there are $k \choose m$  $=\frac{k!}{n!m!}$ terms with $x$ and $y$ appearing $n,m$ times, respectively in the expansion of $\prod (x_{a_j}-y_{a_j})$.
We will now prove that replacing each integral 
$$\frac{1}{n!} \int_{\R^3}{\rm d}\rho_1(x) x_{a_1}\dots  x_{a_n}$$
with 
$$\frac{1}{\prod_{j=1}^n(2j-1)}\cM^{(n)}_{\rho_1}(e_{a_1},\dots ,e_{a_n})$$ 
does not change the sum with respect to $a_j$'s and we can similarly replace the integrals with respect to $\rho_2$. Together with \eqref{def:nmult_int} and the homogeneity of degree $-(k+1)$ of $k$-th order partial derivatives of the Coulomb potential, this concludes the proof of \eqref{eq:BUVdec} and therefore of Lemma \ref{lem:multipolar_expansion}. For $n=1$ this is trivial. We will do it for $n=2$ and then explain how this  works for all $n$ inductively. 

We have 
\begin{equation*}
\cM^{(n)}_{\rho_1}(e_{a_1},\dots ,e_{a_n})=\frac{(-1)^n}{n!} \int_{\R^3} {\rm d}\rho_1(x) |x|^{2n+1} \partial_{a_1} \cdots \partial_{a_n} \left(\frac{1}{|x|}\right),
\end{equation*}
so in particular for $n=2$ 
\begin{equation*}
\cM^{(2)}_{\rho_1}(e_{a_1},e_{a_2})=\frac{1}{2!} \int_{\R^3} {\rm d}\rho_1(x) |x|^{5} \partial_{a_1} \partial_{a_2} \left(\frac{1}{|x|}\right).
\end{equation*}
Now observe that
$$ \partial_{a_2} \left(\frac{1}{|x|}\right)=-\frac{x_{a_2}}{|x|^3}$$
from which it follows that
\begin{equation}\label{eq:Coulompar2} \partial_{a_1} \partial_{a_2} \left(\frac{1}{|x|}\right)=-
\frac{\delta_{a_1 a_2}}{|x|^3} + \frac{3 x_{a_1} x_{a_2}}{|x|^5}.
\end{equation}
Thus 
\begin{equation*}
\frac{1}{3} \cM^{(2)}_{\rho_1}(e_{a_1},e_{a_2})=\frac{1}{2!} \int_{\R^3}  \left( x_{a_1} x_{a_2} -\frac{1}{3}\delta_{a_1 a_2}|x|^2\right){\rm d}\rho_1(x).
\end{equation*}
Due to the fact that the Coulomb potential is harmonic we have that $\delta_{a_1 a_2} \partial_{a_1} \partial_{a_2} \frac{1}{|Le_1|}=0$ in the summation. So replacing $\frac{1}{2!} \int_{\R^3} {\rm d}\rho_1(x) ( x_{a_1} ,x_{a_2})$ with $ \frac{1}{3} \cM^{(2)}_{\rho_1}(e_{a_1},e_{a_2})$ in the right hand side of \eqref{eq:pre_multipole_expansion} will not change the right hand side of \eqref{eq:pre_multipole_expansion}.

For general $n$, we argue as follows. Using that 
$$\partial_a \frac{1}{|x|^{2k+1}}=-\frac{x_a}{(2k+3)|x|^{2k+3}}$$ 
for any $k \in \mathbb{N}$ we obtain by induction that
$$  \partial_{a_1} \cdots \partial_{a_n} \left(\frac{1}{|x|}\right)=(-1)^n \frac{x_{a_1}\dots  x_{a_n}}{|x|^{2n+1}}\prod_{j=1}^n(2j-1)  + R $$
where $R$ has Kronecker deltas and therefore does not change anything for similar reasons as for $n=2$. The $(-1)^n$ in the last formula cancels with the $(-1)^n$ in the definition of the multipole moment. Thus we can similarly replace  $\frac{1}{n!} \int_{\R^3}{\rm d}\rho_1(x) x_{a_1}\dots x_{a_n}$
with $\prod_{j=1}^n(2j-1)^{-1}\cM^{(n)}_{\rho_1}(e_{a_1},\dots ,e_{a_n})$ without changing the right hand side of  \eqref{eq:pre_multipole_expansion}. We thus obtain~\eqref{eq:multipolar_expansion} as desired.\qed

\subsection{Computation of $\cF^{(n,m)}$ for $n+m\leq5$}

Here we compute the multipole-multipole  terms explicitly, using \eqref{def:nmult_int}.

\subsubsection*{Dipole-dipole $m=n=1$} It is given by
\begin{align*}
- (D_1 \cdot e_{a_1}) (D_2 \cdot e_{b_1}) \partial_{a_1} \partial_{b_1}\frac{1}{|e_1|}&\stackrel{\eqref{eq:Coulompar2}}{=}-(D_1 \cdot e_{a_1}) (D_2 \cdot e_{b_1}) (-
\delta_{a_1 b_1} + 3 \delta_{a_1 1} \delta_{b_1 1})\\
&\;\,=D_1 \cdot D_2-3 (D_1 \cdot e_1)(D_2 \cdot e_1). 
\end{align*}

\subsubsection*{Quadrupole-dipole $n=2, m=1$} It is given by
$$-\frac{1}{3}Q_1(e_{a_1}, e_{a_2}) (D_2 \cdot e_{b_1})\partial_{a_1} \partial_{a_2}\partial_{b_1}\frac{1}{|e_1|}.$$
Using \eqref{eq:Coulompar2} we obtain that
\begin{equation}\label{eq:Coulompar3} \partial_{i} \partial_{j} \partial_k \left(\frac{1}{|x|}\right)=3
	\frac{\delta_{ij} x_k + \delta_{ik} x_j + \delta_{jk} x_i}{|x|^5} -15   \frac{x_i x_j x_k}{|x|^7}.
\end{equation}
Thus the quadrupole-dipole term reads
\begin{align*}
&-\frac{1}{3}Q_1(e_{a_1}, e_{a_2}) (D_2 \cdot e_{b_1})\partial_{a_1} \partial_{a_2}\partial_{b_1}\frac{1}{|e_1|}\\
&\quad = Q_1(e_{a_1}, e_{a_2}) (D_2 \cdot e_{b_1}) \left( 5   \delta_{a_1 1}  \delta_{a_2 1}  \delta_{a_3 1}-
 (\delta_{a_1 a_2} \delta_{b_1 1} + \delta_{a_1 b_1} \delta_{a_2 1}  + \delta_{a_2 b_1} \delta_{a_1 1}) \right)\\
 &\quad =5 Q_1(e_1, e_1) (D_2 \cdot e_1)-2 Q_1(e_1,D_2), 
\end{align*}
 where in the last step we used that $Q_1$ is traceless and that $Q_1(e_1,D_2)=Q_1(e_1, e_{a_k}) (D_2 \cdot e_{a_k})$.
 
 Before computing the octople-dipole and the quadrupole-quadrupole terms we note that from \eqref{eq:Coulompar3} it follows that
 \begin{equation}\label{eq:Coulompar4} \partial_{i} \partial_{j} \partial_k \partial_l \left(\frac{1}{|x|}\right)  =3
 \frac{\delta_{ij} \delta_{kl} + \delta_{ik} \delta_{jl} + \delta_{jk} \delta_{il}}{|x|^5} -15   \frac{\delta_{ij} x_k x_l + cyclic}{|x|^7}+ 105 \frac{x_i x_j x_k x_l}{|x|^9}.
 \end{equation}
 
\subsubsection*{Quadrupole-quadrupole $n=2, m=2$} It is given by
\begin{align*}
&\frac{1}{9}Q_1(e_{a_1}, e_{a_2}) Q_2(e_{b_1}, e_{b_2})\partial_{a_1} \partial_{a_2}\partial_{b_1}
 \partial_{b_2} \frac{1}{|e_1|}\\
 &\qquad  \stackrel{\eqref{eq:Coulompar4}}{=} \frac{1}{9}Q_1(e_{a_1}, e_{a_2}) Q_2(e_{b_1}, e_{b_2}) \Big( 3(\delta_{a_1 a_2} \delta_{b_1 b_2} + \delta_{a_1 b_1} \delta_{a_2 b_2} + \delta_{a_1 b_2} \delta_{a_2 b_1}) \\
 &\qquad\qquad  -15(\delta_{a_1 a_2} \delta_{b_1 1} \delta_{b_2 1} + cyclic)+ 105 \delta_{a_1 1}  \delta_{a_2 1}  \delta_{b_1 1} \delta_{b_2 1}\Big) \\
&\qquad  \;\,=\frac{1}{3}\Big(35 Q_1(e_1,e_1) Q_2(e_1,e_1)-20 \langle e_1, Q_1 Q_2 e_1 \rangle + 2\tr(Q_1 Q_2)\Big), 
\end{align*}
 where in the last step we used that $Q_1, Q_2$ are traceless.
  
\subsubsection*{Octopole-dipole $n=2, m=2$} It is given by
\begin{align*}
&-\frac{1}{15} O_1(e_{a_1}, e_{a_2}, e_{a_3}) (D_2 \cdot e_{b_1})\partial_{a_1}\partial_{a_2}\partial_{a_3}
  \partial_{b_2} \frac{1}{|e_1|}\\
  &\qquad \stackrel{\eqref{eq:Coulompar4}}{=} -\frac{1}{15} O_1(e_{a_1}, e_{a_2}, e_{a_3}) (D_2 \cdot e_{b_1}) \Big( 3(\delta_{a_1 a_2} \delta_{a_3 b_1} + \delta_{a_1 a_3} \delta_{a_2 b_1} + \delta_{a_1 b_1} \delta_{a_2 a_3})\\
  &\qquad \qquad-15(\delta_{a_1 a_2} \delta_{a_3 1} \delta_{b_1  1} + cyclic)+ 105 \delta_{a_1 1}  \delta_{a_2 1}  \delta_{a_3 1} \delta_{b_1 1}\Big)\\
  &\qquad \;\,= 3 O(e_1, e_1, D_2) - 7 O(e_1, e_1, e_1) (D_2 \cdot e_1),
\end{align*}
where in the last step we used that $O(v,.,.)$ is traceless for any $v$ and that $O_1(e_1,e_1, e_{a_k}) (D_2 \cdot e_{a_k})= O(e_1, e_1, D_2) $.
 
The dipole-hexadecapole and quadrupole-octapole terms, whose exact expressions do not play an important role in our analysis, can be similarly derived.  
\qed

\section{Interaction energy of several molecules}\label{app:vdWseveral}

In Theorem~\ref{thm:expansion_energy} we have given an expansion for the interaction energy of two molecules, because we were studying isomerizations involving two molecules only. However, our proof given in Section~\ref{sec:ProofvdWdom} can be generalized to the case of several molecules, with little modifications. It is remarkable that, at the order $L^{-6}$, there are three-body corrections which are in general non zero for polar molecules. Those are known in the Physics literature (see, e.g.,~\cite{BucMicZol-07}) but should not be confused with the more famous  Axilrod--Teller--Muto $L^{-9}$ three-body correction which plays an important role for atoms~\cite{AxiTell-43,Muto-43,LilTka-10,StaGobTka-14}. Here we state the appropriate generalization of Theorem \ref{thm:expansion_energy} in the case of several molecules and we briefly sketch its proof.

We consider a system of $K$ molecules, where each molecule is represented by its nuclear positions $Y_k=(y_{k,1},...,y_{k,M_k})$ and nuclear charges $Z_k=(z_{k,1},...,z_{k,M_k})$. We then place each molecule about one point $X_k\in\R^3$ and rotate it by applying the rotation $U_k\in SO(3)$. Denoting for simplicity our main variable by
$$\tau:=(X_1,U_1,...,X_K,U_K),$$
the whole system is described by the nuclear positions
$$Y(\tau):=(X_1+U_1Y_1,...,X_K+U_KY_K)$$
and the nuclear charges
$$Z=(Z_1,...,Z_K).$$
We denote by $L_{ij}=|X_i-X_j|$ the distance between the molecules $i,j$ and assume that $L_{ij}\gg1$ for all $1\leq i<j\leq K$. Finally, let 
$$\cE(\tau):=\min\sigma\Big(H_N(Y(\tau),Z)\Big)$$
be the corresponding ground state energy, with $N=\sum_{k=1}^K|Z_k|$. We call $H_k$ the Hamiltonians of the individual molecules and $E_k=E_{|Z_k|}(Y_k,Z_k)$ the corresponding ground state energies. 

The analogue of \eqref{eq:main_assumption2} for several molecules is
\begin{equation}\label{eq:main_assumption2sev}
	\sum_{k=1}^KE_{|Z_k|}(Y_k,Z_k) < \min_{\substack{(N_1',\dots ,N_K') \neq (|Z_1|,\dots ,|Z_K|)\\ \sum_{k=1}^KN'_k=N}}\; \sum_{k=1}^K E_{N'_k}(Y_k,Z_k).
\end{equation}
In an analogous way to $f_{U,V}$ defined in \eqref{eq:def_f_U_V} we can define $f_{ij,U_i U_j}$ and then for any two ground states  $\Psi_i, \Psi_j$ of $H_i$ and $H_j$ we define  the van der Waals correlation function 
$C_{\rm vdW}(\Psi_i, \Psi_j, U_i, U_j)$ similarly as in 
\eqref{eq:def_C_vdW_Psi_1_2}.
We use the notation 
$$R_{ij}= \Pi_{ij}^\perp\left(H_i \otimes\1+\1\otimes H_j-E_i-E_j\right)^{-1}_{|(\cG_i\otimes\cG_j)^\perp}\Pi_{ij}^\perp,$$
to denote the inverse of the two-body Hamiltonian $H_i \otimes\1+\1\otimes H_j-E_i-E_j$ restricted to the orthogonal complement of $\cG_i\otimes \cG_j$ in $\bigwedge_1^{N_i}L^2\otimes \bigwedge_1^{N_j}L^2$.
The new feature of the case of several molecules is the three body potential
\begin{equation}
W(U_i,U_j,U_k,\Psi_i,\Psi_j,\Psi_k):=  \langle  f_{ik,U_i U_k} \Psi_i\otimes  \Psi_j\otimes  \Psi_k, R_{ij}\, f_{ij,U_i U_j} \Psi_i \otimes \Psi_j\otimes  \Psi_k
 \rangle.
 \end{equation}
Note that here we could replace $R_{ij}$ by $R_{ik}$ without changing the result, or even by
\begin{multline*}
R_{ijk}=\\ \Pi_{ijk}^\perp\left(H_i \otimes\1_2+\1\otimes H_j\otimes\1+\1_2\otimes H_k-E_i-E_j-E_k\right)^{-1}_{|(\cG_i\otimes\cG_j\otimes \cG_k)^\perp}\Pi_{ijk}^\perp. 
\end{multline*}
This is because $H_i + H_j+H_k-E_i-E_j-E_k$ is equal to $H_i+ H_j-E_i-E_j$ when acting on functions in the form $\Xi\otimes \Psi_k$.

\begin{theorem}[Expansion of the ground state energy for several molecules]\label{thm:several_molecules}
Let $\Psi_1,...,\Psi_K$ be any normalized ground states of, respectively, $H_1,...,H_K$. We have the upper bound
\begin{multline}
	\cE(X_1,U_1,...,X_K,U_K)\leq \sum_{k=1}^K E_k  +\\ \sum_{1 \leq i<j \leq K}  \left(\sum_{2\leq n+m\leq 5}\frac{\cF^{(n,m)}(\Psi_i,\Psi_j,U_i,U_j)}{(L_{ij})^{n+m+1}}
	-\frac{C_{\rm vdW}(\Psi_i,\Psi_j,U_i,U_j)}{(L_{ij})^6} \right)
	\\  - \sum_{i=1}^K \sum_{j\neq k \in\{1,...,K\}\setminus\{i\}} \frac{W(U_i,U_j,U_k,\Psi_i,\Psi_j,\Psi_k)}{(L_{ij})^3 (L_{ik})^3} + O\left(\frac{1}{\min(L_{ij})^7}\right),
	\label{eq:upper_several}
\end{multline}
where the $O(\cdots)$ is uniform in $U_1,\dots ,U_K \in SO(3)$.
If in addition the ground states of $H_1,...,H_K$ are all irreducible for $1\leq k\leq K$ and~\eqref{eq:main_assumption2sev} holds, then~\eqref{eq:upper_several} is an equality and each of the above terms is independent of the chosen $\Psi_1,..., \Psi_K$.
\end{theorem}

It is possible to provide a more precise dependence of the error term with respect to the distances $L_{ij}$, but we do not state the corresponding bound for simplicity. 

Note that if the $k$-th molecule does not have a dipole moment then the corresponding three-body energy cancels, by symmetry. 
A better way of writing the full quantum correction is
\begin{equation}
-\pscal{\left(\sum_{1\leq i<j\leq K}\frac{f_{ij,U_i,U_j}}{(L_{ij})^3}\right)\bigotimes_{n=1}^K\Psi_n\;,\;\cR_{1,...,K}\left(\sum_{1\leq k<\ell\leq K}\frac{f_{k\ell,U_k,U_\ell}}{(L_{k\ell})^3}\right)\bigotimes_{n=1}^K\Psi_n}
\label{eq:full_vdW_molecules}
\end{equation}
where
$$\cR_{1,...,K}:=\Pi_{1,...,K}^\perp\left(\sum_{n=1}^K(H_n-E_n)\right)^{-1}_{\big|\left(\bigotimes_{n=1}^K\cG_n\right)^\perp}\Pi_{1,...,K}^\perp$$
and where $\Pi_{1,...,K}$ is the projection onto the ground state eigenspace $\bigotimes_{n=1}^K\cG_n$. When expanding the two sums in~\eqref{eq:full_vdW_molecules}, we obtain the two-body van der Waals term for $\{i,j\}=\{k,\ell\}$, and the three-body term for $\#\{i,j\}\cap\{k,\ell\}=1$. If $\{i,j\}\cap\{k,\ell\}=\emptyset$, the term vanishes as explained below. It is now clear from the formula~\eqref{eq:full_vdW_molecules} that the full quantum correction is $\leq0$ and a slight adaptation of Proposition~\ref{prop:positivity_vdW} shows that it is indeed $<0$. 

\begin{proof}[Sketch of the proof of Theorem~\ref{thm:several_molecules}]
The proof is similar to that given in Section~\ref{sec:ProofvdWdom} for $K=2$. We may define the two-molecule interaction terms $I_{ij}$ and the cut-off characteristic function $\chi_{ij}$. We localize the ground states similarly as for two molecules, at a distance $\frac{1}{4}\min_{k \neq j}L_{kj}$ from each of the $X_j$. 
We denote the cut-off and rotated-translated ground states by $\Phi_{1,\tau},\dots , \Phi_{K,\tau}$ where we recall that $\tau=(X_1,U_1,...,X_K,U_K)$.
The test  function we consider is, similarly as for $K=2$, 
\begin{equation}\label{eq:Phi0sev}
\Phi_\tau:=\bigotimes_{k=1}^K\Phi_{k,\tau}
- \sum_{1 \leq i < j \leq k} \chi_{ij,\tau} R_{ij,\tau}I_{ij,\tau} \bigotimes_{k=1}^K\Phi_{k,\tau}
\end{equation}
where the index $\tau$ indicates that the quantities are translated and rotated in the obvious manner. 
Using that $(H_{j,\tau}-E_j) \Phi_{j,\tau} =O(e^{-cL})$ where $L=\min{L_{ij}}$ and similar computations as for $K=2$, we find for the energy
\begin{align}
&\pscal{\Phi_\tau,\left(H_N(Y(\tau),Z)-\sum_{k=1}^K E_k\right) \Phi_\tau }\nn\\
&=\sum_{1\leq j<k\leq K}\pscal{\Phi_{1,\tau}\otimes \dots \otimes \Phi_{K,\tau},I_{jk,\tau}\Phi_{1,\tau}\otimes \dots \otimes \Phi_{K,\tau}}\nn\\
&\quad - \sum_{1 \leq i < j \leq K} \sum_{1 \leq k < \ell \leq K}\pscal{I_{k\ell,\tau}\Phi_{1,\tau}\otimes \dots \otimes \Phi_{K,\tau},  R_{ij,\tau} I_{ij,\tau}\Phi_{1,\tau}\otimes \dots \otimes \Phi_{K,\tau}}\nn\\ 
&\quad+O\left(\frac1{L^7}\right).\label{eq:HminusEexp}
\end{align}
The first term gives the classical multipolar terms, after applying Lemma~\ref{lem:multipolar_expansion} to the total density
$$\rho=\sum_{k=1}^K\sum_{m=1}^{M_k}z_k\delta_{X_k+U_ky_{k,m}}-\rho_{\bigotimes_{k=1}^K\Phi_{k,\tau}}=\sum_{k=1}^K\left(\sum_{m=1}^{M_k}z_k\delta_{X_k+U_ky_{k,m}}- \rho_{\Phi_{k,\tau}}\right).$$
If $\{k,\ell\} \cap \{i,j\}= \emptyset $ then the corresponding term in the second sum of~\eqref{eq:HminusEexp} vanishes due to the fact that $R_{ij,\tau} \Phi_{i,\tau}\otimes \Phi_{j,\tau}=0$. Splitting then the sum according to whether $\{k,\ell\}= \{i,j\}$ or  $\#\{k,\ell\} \cap \{i,j\}=1$ we obtain the claimed upper bound with similar arguments as for $K=2$.

Assuming that the ground state eigenspaces are irreducible and that \eqref{eq:main_assumption2sev} holds, the lower bound is proved similarly as for $K=2$. We can define the orthogonal projection $\Pi_\tau$ onto the space spanned by our trial functions~\eqref{eq:Phi0sev} for all possible choices $\Psi_1,\dots,\Psi_K$. Arguing with the IMS localization formula as in \cite{MorSim-80,AnaSig-17}, it follows from \eqref{eq:main_assumption2sev} that 
$$\Pi_\tau^\perp\big(H(Y(\tau),Z)-\cE(\tau)\big)\Pi_\tau^\perp \geq \eps\Pi_\tau^\perp $$ 
for $L=\min L_{ij}$ large enough. The Feshbach-Schur method is then applicable.  The off-diagonal terms are small,
$$\norm{\Pi^\bot_\tau H(Y(\tau),Z)\Pi_\tau}=O\left(\frac1{L^4}\right)$$
for the same reasons as for $K=2$. 
Now, our trial state space still has the property of being invariant under permutations of the spins of each molecule. Hence $\Pi_\tau H(Y(\tau),Z)\Pi_\tau$ is a multiple of $\Pi_\tau$ and we obtain the lower bound. This concludes our brief sketch of the proof of the Theorem~\ref{thm:several_molecules}. 
\end{proof}


\end{document}